\documentclass[prd,tightenlines,nofootinbib,superscriptaddress]{revtex4}

\usepackage{amsfonts,amssymb,amsthm,bbm,amsmath}
\usepackage[colorlinks=true,citecolor=magenta,linkcolor=blue]{hyperref}

\usepackage{color,psfrag}
\usepackage[dvips]{graphicx}
\usepackage{epstopdf}

\usepackage{subcaption}
\usepackage{tikz}
\usetikzlibrary{calc}
\usetikzlibrary{decorations.pathmorphing}
\usetikzlibrary{shapes.geometric}
\usetikzlibrary{arrows,decorations.markings}

\newcommand{\alink}[4]
{\draw[decoration={markings,mark=at position 0.6 with {\arrow[scale=1.5,>=stealth]{>}}},postaction={decorate}] (#1) -- node[#3,pos=.5]{$#4$}(#2)}
\newcommand{\link}[2]
{\draw[decoration={markings,mark=at position 0.6 with {\arrow[scale=1.5,>=stealth]{>}}},postaction={decorate}] (#1) --(#2)}

\numberwithin{paragraph}{subsection}

\newcommand{\C}{{\mathbb C}}
\newcommand{\Q}{{\mathbb Q}}
\newcommand{\N}{{\mathbb N}}
\newcommand{\R}{{\mathbb R}}
\newcommand{\Z}{{\mathbb Z}}

\newcommand{\ch}{\mathrm{ch}}
\newcommand{\sh}{\mathrm{sh}}

\newcommand{\cA}{{\mathcal A}}
\newcommand{\cB}{{\mathcal B}}
\newcommand{\cE}{{\mathcal E}}
\newcommand{\cF}{{\mathcal F}}

\newcommand{\cI}{{\mathcal I}}

\newcommand{\cH}{{\mathcal H}}
\newcommand{\cM}{{\mathcal M}}

\newcommand{\cO}{{\mathcal O}}
\newcommand{\cP}{{\mathcal P}}
\newcommand{\cT}{{\mathcal T}}
\newcommand{\cV}{{\mathcal V}}
\newcommand{\cD}{{\mathcal D}}
\newcommand{\cC}{{\mathcal C}}
\newcommand{\cS}{{\mathcal S}}
\newcommand{\cU}{{\mathcal U}}

\newcommand{\SU}{\mathrm{SU}}

\newcommand{\U}{\mathrm{U}}

\newcommand{\be}{\begin{equation}}
\newcommand{\ee}{\end{equation}}
\newcommand{\beq}{\begin{eqnarray}}
\newcommand{\eeq}{\end{eqnarray}}
\newcommand{\bes}{\begin{eqnarray}}
\newcommand{\ees}{\end{eqnarray}}

\newcommand{\mat} [2] {\left ( \begin{array}{#1}#2\end{array} \right ) }

\newcommand{\su}{{\mathfrak{su}}}

\def\arcosh{\mathrm{arcosh}}

\newcommand{\la}{\langle}
\newcommand{\ra}{\rangle}

\newcommand{\tr}{{\mathrm{Tr}}}
\newcommand{\f}{\frac}
\newcommand{\tl}{\widetilde}
\def\nn{\nonumber}
\def\pp{\partial}

\def\tG{\widetilde{G}}

\def\vphi{\varphi}
\def\eps{\epsilon}

\newcommand{\id}{\mathbb{I}}

\def\act{\triangleright}

\def\vu{\vec{u}}
\def\vsigma{\vec{\sigma}}
\def\vJ{\vec{J}}

\def\om{\omega}
\def\rd{\mathrm{d}}

\newcommand{\PR}{{\text{PR}}}

\def\tl{\tilde{l}}
\def\vcC{\vec{\cC}}
\def\hh{h^{(h)}}
\def\hv{h^{(v)}}

\newtheorem{theo}{Theorem}[section]
\newtheorem{lemma}[theo]{Lemma}
\newtheorem{prop}[theo]{Proposition}

\theoremstyle{definition}

\theoremstyle{remark}

\tikzset{->-/.style={decoration={
			markings,
			mark=at position #1 with {\arrow{>}}},postaction={decorate}}}	

\begin{document}

\title{Non-Perturbative 3D Quantum Gravity: \vspace*{1mm}\\ Quantum Boundary States \& Exact Partition Function}

\author{{\bf Christophe Goeller}}\email{christophe.goeller@ens-lyon.fr}
\affiliation{Laboratoire de Physique, ENS Lyon, CNRS-UMR 5672, 46 all\'ee d'Italie, Lyon 69007, France}
\affiliation{Perimeter Institute, 31 Caroline St North, Waterloo ON, Canada N2L 2Y5}
\author{{\bf Etera R. Livine}}\email{etera.livine@ens-lyon.fr}
\affiliation{Laboratoire de Physique, ENS Lyon, CNRS-UMR 5672, 46 all\'ee d'Italie, Lyon 69007, France}
\affiliation{Perimeter Institute, 31 Caroline St North, Waterloo ON, Canada N2L 2Y5}
\author{{\bf Aldo Riello}}\email{ariello@perimeterinstitute.ca}
\affiliation{Perimeter Institute, 31 Caroline St North, Waterloo ON, Canada N2L 2Y5}

\date{\today}

\begin{abstract}

We push forward the investigation of holographic dualities in 3d quantum gravity formulated as a topological quantum field theory, studying the correspondence between boundary and bulk structures. Working with the Ponzano-Regge topological state-sum model defining an exact discretization of 3d quantum gravity, we analyze how the partition function for a solid twisted torus depends on the  boundary quantum state. This configuration is relevant to the AdS$_{3}$/CFT$_{2}$ correspondence. We introduce boundary spin network states with coherent superposition of spins on a square lattice on the boundary surface.
This allows for the first exact analytical calculation of Ponzano-Regge amplitudes with  extended 2D boundary (beyond the single tetrahedron). We get a regularized finite truncation of the BMS character formula obtained from the one-loop perturbative quantization of 3d gravity. This hints towards the existence of an underlying symmetry and the integrability of the theory for finite boundary at the quantum level for  coherent boundary spin network  states.

\end{abstract}

\maketitle
\tableofcontents

\section*{Introduction}

The holographic principle has become the main point of convergence for the various approaches to quantum gravity.
%
The key insight is that  the dynamics of the quantum geometry in a  space-time region can be entirely encoded in the theory induced on the boundary and thereby be faithfully represented by boundary observables.
This thread has  developed from black-hole thermodynamics and the  discovery of the area-entropy law.
In the last two decades these ideas have been substantially deepened on the one hand by the study of the AdS/CFT correspondence from a string theory perspective, and on the other by the study of quantum gravity as an almost-topological quantum field theory from a loop quantum gravity perspective.
A crucial aspect of the holographic principle is that it interlaces the quantum dynamics of gravity with the renormalization flow of quantum geometries and the physics of gravitational edge modes. This materializes into holographic dualities between bulk and boundary theories, which could ultimately provide a non-perturbative definition of quantum gravity. 

Here we  push  further the investigations of  holographic dualities in non-perturbative 3d quantum gravity as initiated in \cite{Short,PRholo1,PRholo2}. Gravity in three space-time dimensions is indeed a topological field theory \cite{Witten:1988hc,Witten:1989sx}, for which gauge-invariant bulk observables should be entirely determined by the space-time topology and appropriate boundary  conditions. This makes it the perfect arena to explore possible realizations of the holographic principle in quantum gravity and the related bulk-boundary dualities \cite{Carlip:2005zn}. 

Formulated as a topological quantum field theory (TQFT), 3D quantum gravity can be exactly discretized and quantized. This is realized by the Turaev-Viro topological invariant \cite{Turaev:1992hq}, whose relation to the Reshetikhin-Turaev invariant \cite{Reshetikhin:1990pr,Reshetikhin:1991tc} reflects the relation between 3D gravity and Chern-Simons theories in the continuum \cite{Witten:1988hc,Witten:1989sx}.
Formulated as a state-sum model, the Turaev-Viro theory reduces to the Ponzano-Regge state-sum model \cite{PR,Regge:2000wu,Freidel:2004vi,Barrett:2008wh,Freidel:2005bb} in the vanishing cosmological constant case  \cite{Ooguri:1991ni,Freidel:2005bb}, which can thus be understood as 3D quantum gravity at $\Lambda =0$.  The interested reader will find more details between the Ponzano-Regge state-sum and topological invariants of 3-manifolds in  \cite{Barrett:2008wh, Freidel:2004nb}.

The Ponzano-Regge model is an intrinsically discrete approach to 3D quantum gravity. We define the partition function and amplitudes on 3D triangulations, or more generally 3D cellular complexes. The topological invariance means that the partition function does not depend on the precise bulk discretization, but solely on its topology and boundary conditions.
The boundary states are defined on the 2d boundary cellular complex. They are defined, both for Ponzano-Regge and Turaev-Viro, as spin networks living on the graph dual to the 2d boundary triangulation, i.e. we dress the links and nodes of the boundary graph with respectively $\SU(2)$ representations and group invariants. This framework allows to study 3D quantum gravity from a quasi-local point of view with arbitrary quantum boundaries, without restricting to   (semi-)classical boundary conditions or asymptotic boundary states at infinity.

Following this line of research developed in \cite{Short,PRholo1,PRholo2,Riello:2018anu}, we propose here to use a class of coherent boundary states, defined as quantum superpositions of spins---thus lengths---admitting a critical regime with a manifest scale invariance in the semi-classical limit.
Introduced in \cite{Freidel:2012ji,Bonzom:2012bn, Bonzom:2015ova,Girelli:2017dbk}, we show that these coherent spin network wave-packets allow for an exact evaluation of the Ponzano-Regge amplitudes with quantum boundary. Applying this to the case of a solid twisted torus, relevant to the AdS${}_{3}$/CFT${}_{2}$ correspondence and the BTZ black hole, we recover a regularized version of the  character formula for the Bondi-Metzner-Sachs (BMS) group  formally defined as a modular form in terms of the Dedekind $\eta$-function. This regularized BMS character formula establishes a clear bridge between the Ponzano-Regge framework and the other approaches to computing the 3D quantum gravity partition function, either from perturbative renormalization of 3D gravity or from the vanishing cosmological constant limit of the CFT formulas based on the Virasoro group  \cite{Barnich:2015mui,Oblak:2015sea}.
%
%
At the end of the day, the surprising simplification of the final result for coherent boundary states points towards the existence of powerful discrete symmetries on the boundary of 3D quantum gravity, which should allow to control the continuum limit of the Ponzano-Regge model.

\medskip

In the first section, we will give a quick review of the Ponzano-Regge state-sum and its boundary spin network states. In particular, we will recall that the resulting partition function does not depend on the 3D bulk triangulation but only depends on the boundary state. The second section will introduce coherent spin network wave-packets, as quantum superpositions of spins on the boundary. We will clarify their geometrical interpretation and discuss their scale invariance and critical regime in the semi-classical limit. We will present the main result of the paper in the third section. We will focus on the solid twisted torus and provide the exact analytical computation of the 3D quantum gravity amplitude for anisotropic coherent boundary spin networks defined on a square lattice on the boundary torus. We will show that we obtain a totally regularized version of the BMS character formula, which leads back to the standard BMS formula in the asymptotic limit for the lattice size.
We will conclude with a discussion of possible extensions of our methods and of their relation to the existence of discrete symmetries of the Ponzano-Regge boundary theory and its potential integrability.

\section{Ponzano-Regge state-sum with boundary state}

In 1968, Ponzano and Regge introduced a model to describe 3D quantum gravity without cosmological constant based on the fact that the algebraic structures of the representation theory of the $\SU(2)$ Lie group can be interpreted as quantum version of elements of 3D geometry and that, as a consequence, the 6j spin recoupling symbols reproduced the Regge action for a  tetrahedron in the large spin limit \cite{PR} (see \cite{Barrett:2008wh} for a modern presentation).
It was later understood by Rovelli in the 90's that the Ponzano-Regge state-sum realizes a discrete path integral for 3D gravity, interpreted as a sum over histories of spin network states for 3D loop quantum gravity  \cite{Rovelli:1993kc}.
From this perspective, the Ponzano-Regge state-sum is a quantized version of Regge calculus for discretized 3D gravity \cite{Regge:2000wu}.
It is possible to derive the model directly as a path integral for 3D gravity as a topological field theory (of the BF type) \cite{Freidel:2004vi,Freidel:2005bb} or as a canonical quantization of 3D gravity  \cite{Noui:2004iy,Noui:2004iz,Bonzom:2011hm,Bonzom:2011nv}.
The model was also generalized to the case of a non vanishing cosmological constant through a $q$-deformation of the $\SU(2)$ Lie group by Turaev and Viro \cite{Turaev:1992hq,Mizoguchi:1991hk,Freidel:1998ua,Bonzom:2014bua}, to a Lorentzian signature \cite{Davids:1998bp,Freidel:2000uq,Davids:2000kz}, and to the four-dimensional case either as a standard $BF$ gauge theory \cite{Ooguri:1992eb,Crane:1993if,Crane:1994ji} or as a higher-gauge generalization thereof \cite{Baratin:2006gy,Baratin:2014era,Asante:2019lki}.

In this section, we review the basics of the Ponzano-Regge partition function for a 3D cellular complex and discuss 2D boundary states and the resulting quantum amplitudes.

\subsection{Quantizing 3d Gravity as Discretized BF Theory}

Classical 3D gravity in its first order formulation is defined in terms of a triad 1-form $e^{a}=e^{a}_{\mu}dx^{\mu}$ valued in the Lie algebra $\su(2)$ and a $\su(2)$-connection $\omega$.
%
%
Its action for a closed 3-manifold $\cM$, and a vanishing cosmological constant $\Lambda=0$, reads:
\be
S_{3D}[e,\omega]
=
\int   e \wedge F[\omega] 
\equiv
\int_{\cM} d^{3}x \; \tr( e \wedge F[\omega] )
\,,
\ee
where $F[\omega]=\text{d}\omega + \omega \wedge \omega$ is the $\su(2)$ valued 2-form curvature  of the connection.
This is the action for a topological field theory of the BF type, where the frame field $e$ plays the role of the B field.
The theory has no local degree of freedom in the bulk and its path integral is a topological invariant related to the Ray-Singer torsion \cite{Blau:1989dh,Blau:1989bq} and the Reshetikhin-Turaev invariant \cite{Freidel:2004nb},
\be
\label{ZBF}
Z[\cM] = \int \cD e \cD \omega \; e^{-i S_{3D}[e,\omega]}
=
\int \cD e \cD \omega \; e^{-i \int e \wedge F[\omega]}
\ee
One can integrate\footnotemark{} over the frame field $e$, which plays the role of a Lagrange multiplier enforcing the flatness of the connection $F[\omega]=0$. Formally, this leads to:
\begin{equation}
	Z[\cM] = \int \cD\mu[ \omega] \, \delta(F[\omega]).
	\label{eq:BF_integrated_e}
\end{equation}
\footnotetext{%
	Instead of integrating over the frame field $e$, one could instead integrate over the connection $\om$. Solving for the connection in terms of $e$ gives the Levi-Civita connection. Integrating over $\omega$ amounts to plugging this solution back in the action, yielding the usual Einstein-Hilbert action for 3D gravity in terms of the metric ${}^{3D}g_{\mu\nu}=e^{a}_{\mu}e^{b}_{\nu}\delta_{ab}$.
}
In order to cleanly perform this integration, one needs to suitably gauge-fix the action and path integral. Introducing the ghosts and BRST formalism leads to the measure $\mu[ \omega]$ over the moduli space of flat connection given by the Ray-Singer torsion \cite{Blau:1989dh,Blau:1989bq}. The Ponzano-Regge model is the discrete equivalent of this computation and can be written in terms of the  Reidemeister torsion for a twisted cohomology, which is shown to be equal to the Ray-Singer torsion \cite{Barrett:2008wh}.

\medskip

To give a precise meaning to the partition function \eqref{ZBF}, the natural method is to discretize it. There are two ways to proceed.
Historically, one focuses on the frame field $e$ and the associated metric. Working on a 3D cellular complex, we associate algebraic structures from the representation theory of the $\SU(2)$ Lie group to geometrical elements. For instance, an irreducible representation is associated to each edge of the cellular complex such that its spin gives the edge length in Planck unit. Then invariant tensors intertwining between representations are associated to each face and finally an amplitude constructed from the representations and intertwiners is associated to every 3D cell and ultimately to the whole 3D cellular complex.

The second path is to focus on the connection $\om$ and to treat 3D gravity as a gauge theory, to be discretized as a lattice gauge theory. Working on the dual 3D cellular complex, the connection is discretized into $\SU(2)$ group elements running along every dual edge (thus going through every face). The curvature is discretized into the $\SU(2)$ holonomies going along the boundary of every dual face (thus going around every edge, intuitively measuring the deficit angle around the edge \`a la Regge). The discrete path integral is then defined as the integral over the discretized connection of the product of $\delta$-distributions constraining  all these $\SU(2)$ holonomies to be  the identity, i.e. constraining the discrete connection to be  flat.

It can be shown that these two methods are ultimately equivalent and yield the same Ponzano-Regge amplitudes. 
In both approaches, the key ingredient  to fix every possible ambiguities consists in requiring  that the model be topologically invariant, i.e. that the resulting amplitudes do not depend on the details of the cellular complex but only on its overall topology. In the case that the 3D cellular complex is simplicial, i.e. a 3D triangulation, this amounts to ensuring that the model be invariant under Pachner moves. This strong criteria is enough to ensure the uniqueness of the model at the end of the day.

\medskip

Let us explain in more details these two ways to construct the Ponzano-Regge amplitudes.
We work on an 3D oriented cellular complex $\Delta$, which can be constructed as a cellular decomposition of a 3D oriented compact manifold $\cM$. Let us label the four layers of cells as $v$, $e$, $f$ and $\sigma$, for the 0-cells (vertices), 1-cells (edges), 2-cells (faces) and 3-cells, respectively. The dual cellular complex is again an oriented 3D complex, which we call $\Delta^*$, and is made of dual vertices $v^{*}$, dual to the 3-cells, dual edges $e^{*}$, dual to the faces, 2D plaquettes $p$, dual to the edges, and 3D bubbles $b$, dual to the original vertices. We summarize these notations in the following table:
\begin{center}
	\begin{tabular}{|c|c|c|c|}
		\hline
		~~dimension (in $\Delta$)~~ & ~~co-dimension (in $\Delta^{*})~~$  & ~~~~~~$\Delta$~~~~~~ & ~~~~~$\Delta^{*}$~~~~~ \\ \hline
		0& 3  & $v$ & $b$ \\ \hline
		1& 2  & $e$ & $p$ \\ \hline
		2& 1  & $f$ & $e^{*}$ \\ \hline
		3& 0  & $\sigma$ & $v^{*}$ \\
		\hline
	\end{tabular}
\end{center}

\medskip

{\bf Ponzano-Regge as a State-Sum:}
The first method to define the Ponzano-Regge model is through a quantization of the geometrical elements of the 3D cellular complex $\Delta$, associating to each cell an element from the representation theory of the $\SU(2)$ Lie group depending on the dimension of the considered cell.
To each edge $e$ is attached a $\SU(2)$ irreducible representation, i.e. a spin $j_{e}\in\f\N2$ defining the $\SU(2)$ irreducible representation based on a Hilbert space $\cV^{j_{e}}$ with dimension $\dim \cV^{j_{e}}=2j_{e}+1$. States in the Hilbert space $\cV^{j}$ can be interpreted as  quantum vectors of norm $j$. The standard basis states $|j,m\ra$ are labeled by the magnetic moment $m$ running by integer steps from $-j$ to $+j$. Another interesting basis is the overcomplete basis of coherent states, $g\,|j,j\ra$ for $g\in\SU(2)$, which minimize the uncertainty for the $\su(2)$ generators  and can be interpreted as semi-classical 3D vectors (see e.g. \cite{Dupuis:2010iq}).

Then, each face $f$ of the cellular complex $\Delta$ forms a polygon, whose boundary is a sequence of edges $e\in\pp f$. We associate to the face the space of intertwiners, i.e. $\SU(2)$-invariant states or singlets, between the spins attached to its boundary edges,
\be
\cH_{f}[\{j_{e}\}]
:=
\textrm{Inv}_{\SU(2)}\left[
\bigotimes_{e\in\pp f}\cV^{{j_{e}}}
\right]
\,,
\qquad
\dim \cH_{f}
=
\int_{\SU(2)} \rd g\,
\prod_{e\in\pp f}\chi_{j_{e}}(g)
\,,
\ee
where the characters, $\chi_{j}(g)=\tr\,D^{j}(g)$, are the traces of the Wigner matrices representing the group elements $g\in\SU(2)$ in the corresponding representation of spin $j$.
%
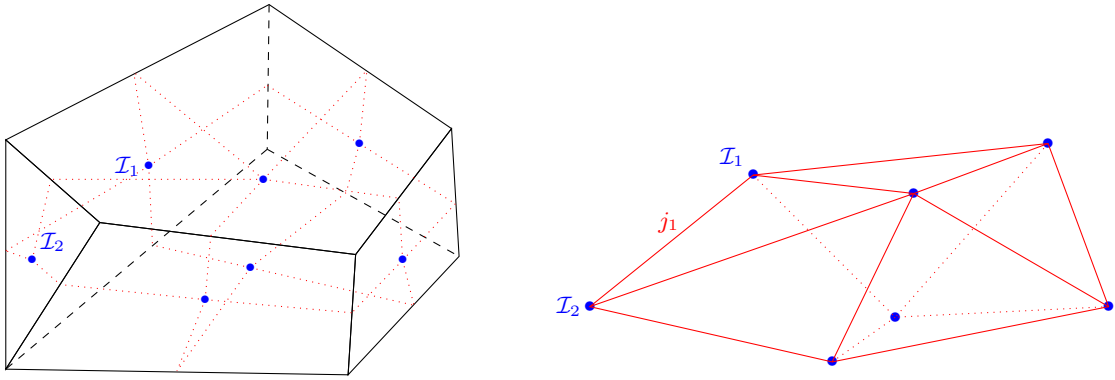
\begin{figure}[h!]
	\begin{subfigure}[t]{.4\linewidth}
		\begin{tikzpicture}[scale=2.5]
		\centering
		\coordinate (A1) at (0.31,1.66);
		\coordinate (A2) at (-1.09,0.94);
		\coordinate (A3) at (-0.59,0.5);
		\coordinate (A4) at (0.77,0.33);
		\coordinate (A5) at (1.28,1);
		\coordinate (A6) at (-1.09,-0.28);
		\coordinate (A7) at (0.73,-0.31);
		\coordinate (A8) at (1.32,0.32);
		\coordinate (A9) at (0.3,0.89);
		
		\coordinate (C1) at (-0.33,0.8);
		\coordinate (C2) at (-0.95,0.3);
		\coordinate (C3) at (-0.03,0.09);
		\coordinate (C4) at (1.02,0.3);
		\coordinate (C5) at (0.79,0.92);
		\coordinate (C6) at (0.28,0.73);
		\coordinate (C7) at (0.21,0.26);
		
		\coordinate (1E1) at (0.3,1.23);
		\coordinate (1E2) at (-1.09,0.35);
		\coordinate (1E3) at (-0.41,1.3);
		\coordinate (1E4) at (-0.31,0.38);
		\coordinate (2E1) at (-0.85,0.73);
		\coordinate (2E2) at (-0.8,0.17);
		\coordinate (3E1) at (0.0,0.43);
		\coordinate (3E2) at (-0.18,-0.29);
		\coordinate (3E3) at (0.75,0.02);
		\coordinate (4E1) at (1,0.63);
		\coordinate (4E2) at (1.08,0.06);
		\coordinate (4E3) at (1.3,0.6);
		\coordinate (5E1) at (0.72,0.66);
		\coordinate (5E2) at (0.83,1.31);
		
		\draw (A1)--(A2)--(A3)--(A4)--(A5)--cycle;
		\draw (A2)--(A3)--(A6)--cycle;
		\draw (A6)--(A3)--(A4)--(A7)--cycle;
		\draw (A4)--(A5)--(A8)--(A7)--cycle;
		\draw[dashed] (A1)--(A9); \draw[dashed] (A9)--(A8); \draw[dashed] (A9)--(A6);
		
		\draw[red,dotted] (C1)--(1E1); \draw[red,dotted] (C1)--(1E2); \draw[red,dotted] (C1)--(1E3); \draw[red,dotted] (C1)--(1E4);   
		\draw[red,dotted] (C2)--(2E1); \draw[red,dotted] (C2)--(2E2); \draw[red,dotted] (C2)--(1E2);
		\draw[red,dotted] (C3)--(3E1); \draw[red,dotted] (C3)--(3E2); \draw[red,dotted] (C3)--(3E3); \draw[red,dotted] (C3)--(2E2);
		\draw[red,dotted] (C4)--(4E1); \draw[red,dotted] (C4)--(4E2); \draw[red,dotted] (C4)--(4E3); \draw[red,dotted] (C4)--(3E3);
		\draw[red,dotted] (C5)--(5E1); \draw[red,dotted] (C5)--(5E2); \draw[red,dotted] (C5)--(1E1); \draw[red,dotted] (C5)--(4E3);
		\draw[red,dotted] (C6)--(1E3); \draw[red,dotted] (C6)--(2E1); \draw[red,dotted] (C6)--(3E1); \draw[red,dotted] (C6)--(4E1); \draw[red,dotted] (C6)--(5E2);
		\draw[red,dotted] (C7)--(1E4); \draw[red,dotted] (C7)--(3E2); \draw[red,dotted] (C7)--(4E2); \draw[red,dotted] (C7)--(5E1);
		
		\draw[blue] (C1) node[scale=0.8]{$\bullet$}; \draw[blue] (C1) node[left]{$\cI_1$};
		\draw[blue] (C2) node[scale=0.8]{$\bullet$}; \draw[blue] (C2) node[above right]{$\cI_2$};
		\draw[blue] (C3) node[scale=0.8]{$\bullet$};
		\draw[blue] (C4) node[scale=0.8]{$\bullet$};
		\draw[blue] (C5) node[scale=0.8]{$\bullet$};
		\draw[blue] (C6) node[scale=0.8]{$\bullet$};
		\draw[blue] (C7) node[scale=0.8]{$\bullet$};
		\end{tikzpicture}
		\caption{Boundary graph to the boundary of a 3-cell $\pp\sigma$, represented in black. The blue dots are the duals of the boundary faces, carrying intertwiners, while the red dotted line are the dual of the edges, carrying spins.}
		
	\end{subfigure}
	\hspace{4mm}
	\begin{subfigure}[t]{.4\linewidth}
		\centering
		\begin{tikzpicture}[scale=3.5]
		\coordinate (C1) at (-0.33,0.8);
		\coordinate (C2) at (-0.95,0.3);
		\coordinate (C3) at (-0.03,0.09);
		\coordinate (C4) at (1.02,0.3);
		\coordinate (C5) at (0.79,0.92);
		\coordinate (C6) at (0.28,0.73);
		\coordinate (C7) at (0.21,0.26);
		
		\draw[blue] (C1) node{$\bullet$}; \draw[blue] (C1) node[above left]{$\cI_1$};
		\draw[blue] (C2) node{$\bullet$}; \draw[blue] (C2) node[left]{$\cI_2$};
		\draw[blue] (C3) node{$\bullet$};
		\draw[blue] (C4) node{$\bullet$};
		\draw[blue] (C5) node{$\bullet$};
		\draw[blue] (C6) node{$\bullet$};
		\draw[blue] (C7) node{$\bullet$};
		
		\draw[red] (-0.65,0.55) node[above]{$j_1$};
		\draw[red] (C1)--(C2)--(C3)--(C4)--(C5)--(C1);
		\draw[red] (C1)--(C6); \draw[red] (C2)--(C6); \draw[red] (C3)--(C6); \draw[red] (C4)--(C6); \draw[red] (C5)--(C6);
		\draw[red,dotted] (C1)--(C7); \draw[red,dotted] (C3)--(C7); \draw[red,dotted] (C4)--(C7); \draw[red,dotted] (C5)--(C7);
		\end{tikzpicture}
		\caption{The boundary spin network of a 3-cell, dressed with spins $j_l$ (on the red links) and intertwiners $\cI_{f}$ (on the blue dots). From the spin network viewpoint, the links between the intertwiners are straight.}
	\end{subfigure}
	
	\caption{Representation of the boundary graph for a 3-cell (a) and the associated spin network (b).}
	\label{fig:boundaryspinnet}
\end{figure}

The next step is attributing an amplitude to each 3-cell $\sigma$. The boundary of the 3-cell forms a polyhedron, made of the faces $f\in\pp\sigma$ glued together by their shared edges $e\in\pp\sigma$. Each edge $e$ carries its spin $j_{e}$, while each face $f$ carries an intertwiner state $\cI_{f}\in\cH_{f}[\{j_{e}\}_{e\in\pp f}]$. It is customary to formalize this in terms of the  topological dual of the 3-cell boundary $\pp \sigma$, called the boundary graph, on which lives the boundary spin network as illustrated in figure \ref{fig:boundaryspinnet}.
The nodes of this boundary spin network correspond to the boundary faces and are dressed with the intertwiner states $\cI_{f}$, while the links of the boundary spin network correspond to the boundary edges and are dressed with the spins $j_{e}$. We then define the evaluation of the boundary spin network of the 3-cell as the contraction\footnotemark{} of the intertwiner states along the boundary graph,
\footnotetext{%
A way to avoid using the structure maps and clarify the role of orientations on the boundary $\pp\sigma$ is to orient the intertwiners themselves. On the oriented boundary spin network dual to $\pp\sigma$, the nodes intertwine in a $\SU(2)$-invariant way between the incoming links and the outgoing links:
\be
\cI_{n}\,:\,\, \bigotimes_{l,t(l)=n}\cV^{{j_{l}}}\rightarrow\bigotimes_{l,s(l)=n}\cV^{{j_{l}}}
\,,
\nn
\ee
where $s(l)$, resp. $t(l)$, denotes the source, resp. target, of the link $l$ on the boundary graph. It is straightforward to contract such oriented intertwiners states together by tracing over the spin spaces:
\be
\cE\big{[}j_{l},\cI_{n}\big{]}
=
\tr_{\{\cV^{j_{l}}\}}\bigotimes_{n}\cI_{n}
=
\sum_{\{m^{s,t}_{l}\}}
\prod_{n}\left\la \{j_{l},m^{s}_{l}\}_{s(l)=n}|\cI_{n}|\{j_{l},m^{t}_{l}\}_{t(l)=n}\right\ra
\,.
\nn
\ee
}
i.e. the trace over the spin Hilbert spaces $\cV^{(j_{e})}$ with the structure maps $\eps_{j}$ switching the orientation\footnotemark{} along the links as $\eps_{j}\,|j,m\ra=(-1)^{j-m}\,|j,-m\ra$:
\be
\cE_{\pp\sigma}\big{[}j_{e\in\pp\sigma},\cI_{f\in\pp\sigma}\big{]}
\,=\,
\tr_{\bigotimes_{e}\cV^{j_{e}}}\,
\left[
\bigotimes_{e\in\pp\sigma}\eps_{j_{e}}
\otimes
\bigotimes_{f\in\pp\sigma}\cI_{f}
\right]
\,.
\ee
\footnotetext{%
Despite the structure maps, the orientation of the boundary graph is still relevant. Indeed, the structure maps squares to minus the identity, $\eps_{j}^{2}=-\id$ on $\cV^{j}$. This leads to a sign ambiguity $(-1)^{2j}$ in the definition of the boundary spin network evaluation. Once the 3-cell amplitudes are glued back together to define the overall amplitude for the whole 3D cellular complex, this sign ambiguity disappears as long as the 3-cell boundaries are all consistently oriented, inherited from the edge orientations and a clockwise planar orientation around every 3-cell boundary.
}
Let us summarize the hierarchy of the correspondence between geometric elements and algebraic objects in the following table\footnotemark:
\begin{center}
	\begin{tabular}{|c|c|rl|}
		\hline
		~~dimension (in $\Delta$)~~ &   ~~~~~~$\Delta$~~~~~~ & \multicolumn{2}{c|}{$\SU(2)$ representation theory} \\ \hline
		0&  $v$ & & \\ \hline
		1&$e$ & spin\,\,& $j_{e}\in\N/2$\\ \hline
		2& $f$ & intertwiner\,\, &$\cI_{f}\in\textrm{Inv}_{\SU(2)}[\otimes_{e\in\pp f} j_{e}]$~ \\ \hline
		3&  $\sigma$ &  ~~ spin network evaluation\,\, &$\cE_{\pp\sigma}[\{j_{e},\cI_{f}\}_{e,f\in\pp\sigma}]$~ \\
		\hline
	\end{tabular}
\end{center}
\footnotetext{
It can seem awkward that no algebraic structure is associated to the vertices of the 3D cellular complex. The situation becomes clearer in the lattice gauge theory picture, where the Bianchi identities live at the vertices and require gauge-fixing \cite{Freidel:2004vi}. Moreover, when introducing defects and generalizing the Ponzano-Regge model to an extended topological theory, they will naturally be topological defects associated to 2D, 1D and 0D structures, respectively interpretable as boundaries, particle world-lines and particle interactions \cite{Freidel:2004vi,Freidel:2005bb}.
}

Finally, the 3-cells are assembled together into the full 3D cellular decomposition. The Ponzano-Regge  amplitude for a closed cellular complex $\Delta$ is defined as a state-sum by summing the product of the 3-cell amplitudes over the spins and taking the trace over the intertwiner spaces:
\be
Z^{\PR}[\Delta]
\,=\,
\sum_{\{j_{e}\}_{e\in\Delta}}\prod_{e} (-1)^{2_{j_e}}(2{j_e}+1)
\prod_{f}(-1)^{\sum_{e\in \pp f}j_{e}}
\,\tr\,\bigotimes_{\sigma\in\Delta}\cE_{\pp\sigma}\big{[}\{j_{e},\cI_{f}\}_{e,f \in\pp\sigma}\big{]}
\,,
\ee
where the trace is taken by summing over an orthonormal basis of the intertwiner spaces associated to the faces for each global assignment of spins $\{j_{e}\}_{e\in\Delta}$.

When the 3D cellular complex is simplicial, i.e. when $\Delta$ is a triangulation, then this definition leads back to the original Ponzano-Regge partition function in terms of the $\{6j\}$-symbol for spin recoupling from the theory of $\SU(2)$ representations. Indeed, all the faces $f$ are now triangles and the corresponding 3-valent intertwiners are unique once the spins are given: they are given by the Clebsh-Gordan coefficients, or equivalently Wigner's 3j-symbols. 3-cells are all tetrahedra, each made of four triangles. A tetrahedron is dressed with 6 spins $j_{e}$ attached to its edges and interpreted as their lengths in Planck unit, $\ell_{e}=j_{e}\,l_{Planck}$. The corresponding amplitude for the tetrahedron is then obtained by contracting the corresponding four Clebsh-Gordan coefficients into a 6j-symbol. The Ponzano-Regge state-sum for the triangulation finally reads as:
\be
Z_{\PR}[\Delta]
=
\sum_{\{j_{e}\}_{e\in\Delta}}
\prod_{e} (-1)^{2_{j_e}}d_{j_e}
\prod_{f}(-1)^{\sum\limits_{e\in \pp f} j_{e}}
\prod_{\sigma} \Big\lbrace 6 j\Big\rbrace_{\sigma}
\,.
\ee
Due to the Biedenharn-Elliot identity (or pentagonal identity) satisfied by the 6j-symbols, this state-sum is invariant under Pachner moves (see e.g.\cite{Barrett:1997,Barrett:2008wh,Bonzom:2009zd})---switching 3 tetrahedra for 2 tetrahedra keeping the same boundary, or switching a single tetrahedron for 4 tetrahedra by adding an extra vertex---and thus depends only on the overall topology of $\Delta$. 

\medskip

{\bf Ponzano-Regge as Lattice Gauge Theory:}
The second method focuses on the connection $\om$ and discretizes 3D gravity as a lattice gauge theory.
The connection is discretized as a set of $\SU(2)$ group elements living on the dual edges in the dual cellular decomposition $\Delta^{*}$.
A group element $g_{e^{*}}$ on the oriented dual edge $e^{*}$ can be equivalently seen as a group element $g_{f}$ going through the face $f$ dual to $e^{*}$. This group element is interpreted as the $\SU(2)$ change of frame between the two 3-cells sharing the face $f$, as illustrated in figure \ref{fig:discrete connection} .
We then switch the flatness constraint $F[\om]=0$ for the requirement that the holonomies of the discrete connection around every dual face, or plaquette, in $\Delta^{*}$ be equal to the holonomy. More precisely, for each edge $e\in\Delta$ and its dual plaquette $p\in\Delta^{*}$, we define the holonomy, as illustrated in figure \ref{fig:discrete connection}:
\be
G_{e\in\Delta}=\overrightarrow{\prod_{f\ni e}} g_{f}
=
\overrightarrow{\prod_{e^{*}\in \pp p}} g_{e^{*}}=G_{p\in\Delta^{*}}
\,.
\ee
\begin{figure}[h!]
\begin{subfigure}[t]{.38\linewidth}
	\centering
		\begin{tikzpicture}[scale=2.5]
			\coordinate (A1) at (-0.17,1.13);
			\coordinate (A2) at (-0.57,0.27);
			\coordinate (A3) at (-0.43,-0.35);
			\coordinate (A4) at (0.21,-0.23);
			\coordinate (A5) at (0.28,0.59);
			\coordinate (C1) at (-0.84,0.75);
			\coordinate (C2) at (0.77,-0.11);
			\coordinate (C) at (-0.12,0.36);
			\coordinate (Cbis) at (-0.45,0.54);

			\draw (C1) node{$\bullet$} node[above]{$s$}; \draw (C2) node {$\bullet$} node[above]{$t$};
			
			\draw[draw=red,fill=red!20] (A1)--(A2)--(A3)--(A4)--(A5)--cycle;
			\draw[->-=0.5] (C1)-- node[pos=0.5,sloped,above]{$e^{*}$} (Cbis);
			\draw[dotted] (Cbis)--(C);
			\draw (C) -- node[pos=0.63,above,sloped]{$g_{e^*}=g_{f}$}(C2);
			\draw[red] (C) node{$\bullet$};
			\draw[red] (A1) node[below=1cm]{$f$};
		\end{tikzpicture}
\caption{$\SU(2)$ group element $g_{e^{*}}=g_{f}$ on the dual edge $e^{*}$, dual to the face (in red) $f$, defining the change of frames between the two 3-cells (dual to the source and target node $s$ and $t$ respectively) sharing the face $f$.}
\end{subfigure}
\hspace{4mm}
\begin{subfigure}[t]{.55\linewidth}
	\centering
		\begin{tikzpicture}[scale=2.5,>=stealth]
			\coordinate (A1) at (-1.1,0.06);
			\coordinate (A2) at (-0.4,-0.2);
			\coordinate (A3) at (0.6,0);
			\coordinate (A4) at (0.72,0.41);
			\coordinate (A5) at (0.05,0.59);
			\coordinate (A6) at (-0.66,0.56);
			
			\coordinate (C1) at (-0.49,1.09);
			\coordinate (C2b) at (0.01,-0.12);
			\coordinate (C2) at (0.32,-0.86);
			\coordinate (C) at (-0.12,0.21);
			
			\draw[fill=black!20] (A1)--(A2)--(A3)--(A4)--(A5)--(A6)--cycle;
			\draw[->-=0.5] (A1)-- node[sloped,below]{$g_1$} (A2); \draw[->-=0.5] (A2)-- node[sloped,below]{$g_2$} (A3); \draw[->-=0.5] (A4)-- node[right]{$g_3$} (A3); \draw[->-=0.5] (A4)-- node[sloped,above]{$g_4$} (A5); \draw[->-=0.3] (A5)-- node[sloped,above,pos=0.3]{$g_5$} (A6); \draw[->-=0.5] (A6)-- node[left]{$g_6$} (A1);
			\draw (A1) node[right=1cm]{$p$};
			\draw[red,thick] (C1)-- node[pos=0.2,right]{$e$}(C); \draw[red,dotted,thick] (C)--(C2b); \draw[red,thick] (C2b)--(C2); \draw[red] (C) node{$\bullet$};
		\end{tikzpicture}

\caption{The $\SU(2)$ holonomy around an edge $e$ is defined as the ordered product of the group elements along the boundary of the plaquette $p$ dual to $e$ and gives a measure of the curvature of the discrete connection. In the example depicted above, considering that the cycle starts with $g_1$ following the orientation of the dual edge $1$, $G_{p} = g_1 g_2 g_3^{-1} g_4 g_5 g_6$.}
\end{subfigure}
\caption{A discrete $\su(2)$-connection as $\SU(2)$ group elements along the dual edges $e^{*}\in\Delta^{*}$.}
\label{fig:discrete connection}
\end{figure}
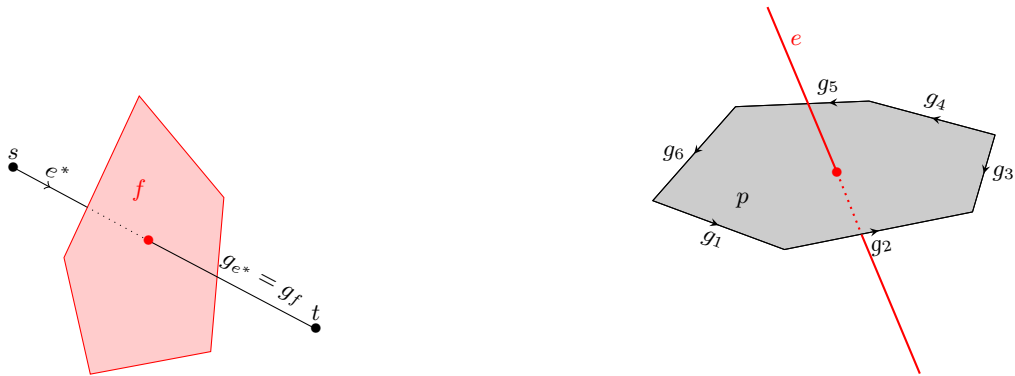%
%
Notice that one should choose a starting point (or root) on the plaquette to define the ordered product around it. However, since we constrained the $G$'s to be equal to the identity, this does not matter for now. 
The partition function is then formulated as a lattice gauge theory on $\Delta^{*}$:
\begin{equation}
Z^{\PR}[\Delta] =
\int[\rd g_{e^{*}}]_{e^{*}\in\Delta^{*}}\,
\prod_{p\in\Delta^{*}}\delta\left(\overrightarrow{\prod_{e^{*}\in \pp p}} g_{e^{*}}\right)
\,.
\label{eq:PR_model_connection}
\end{equation}
The topological invariance of the model is straightforward to show diagrammatically \cite{Girelli:2001wr}.

Upon using the Plancherel decomposition of the $\delta$-distribution over $\SU(2)$ as a sum over spins according to the Peter-Weyl theorem, 
\be
\delta(G)=\sum_{j\in\f\N2}(2j+1)\chi_{j}(G)\,,
\ee
and performing the integration over $\SU(2)$ in terms of intertwiners, one can easily show that the two definitions, as a state-sum or as a lattice gauge theory, of the Ponzano-Regge model exactly match (see \cite{Freidel:2004vi,Barrett:2008wh} for a review). The difference between the two formulations can be understood as a change of basis from frame field (``metric'') to connection.

\medskip

{\bf Divergence and Gauge-Fixing:}
In general, the Ponzano-Regge amplitude for a generic cellular complex $\Delta$ diverges. These divergences emerge due to the (non-compact) shift symmetry of the $BF$ action, which consists of arbitrary shifts of the frame field $e\mapsto e + d_\omega \lambda$  \cite{Freidel:2004vi}. At the discrete level, this translates into a divergence, in $\delta(\id)$ or equivalently $\sum_{j}(2j+1)^{2}$, for each bulk vertex $v\in\Delta$. This is clear in the lattice gauge theory formulation. The bubble dual to a vertex is a closed surface, made of several plaquettes, each of those coming with a $\delta$-distribution $\delta(G_{p})$ enforcing the flatness of the holonomy around it. One of these distributions is redundant due to the Bianchi identity.
Removing one such distribution  at the interface between two adjacent bulk vertices amounts to removing one of the two vertex while merging their corresponding 3-cells; this can be reiterated until one gets a topologically-equivalent cellular complex without any bulk vertex.
Mathematically, this amounts to choosing a maximal tree $T$ on the cellular complex $\Delta$ and removing all the $\delta$-distribution $\delta(G_{e})$ for the edges $e\in T$. It was shown in \cite{Freidel:2004vi} that this gauge-fixing procedure has a trivial Fadeev-Popov determinant and does not depend on the choice of gauge-fixed plaquettes. Explicit examples of gauge-fixings for various cellular complexes can be found in \cite{Freidel:2005bb,PRholo2}.
A complete discussion of the divergence problem formalized in terms of cohomology can be found in \cite{Barrett:2008wh,Bonzom:2010ar,Bonzom:2011br}.  On closed manifolds (with defects), it results in a formula for the Ponzano-Regge partition function in terms of the Reidemeister torsion for the twisted cohomology of flat connections on $\Delta^{*}$, which matches the Ray-Singer analytic torsion as computed from the continuous path integral for BF theory.


\subsection{Partition function with Boundary States}

We would like to consider the more general case of a manifold with boundaries, in order to study the boundary theories induced by the Ponzano-Regge model and possible holographic dualities between bulk and boundary theories along the lines sketched in \cite{Short,Riello:2018anu}.
Let us consider a 3D cellular complex $\Delta$ with a 2D boundary $\pp\Delta$. The definition of a 2D boundary state extends the definition of the boundary spin network for a 3-cell. As in the construction of the Ponzano-Regge model as a state-sum presented in the previous section, we introduce the 2D cellular complex $(\pp\Delta)^{*}$  dual to the 2D boundary and focus on its 1-skeleton graph $\Gamma$.

Quantum boundary states are defined as wave-functions of a discrete connection on the boundary graph $\Gamma$. They are (with a slight abuse of language) referred to as spin networks state\footnotemark.
\footnotetext{%
A spin network mathematically refers to the basis of the space of $\SU(2)$-gauge invariant $L^{2}$ wave-functions diagonalizing the $\SU(2)$ Casimirs and constructed from $\SU(2)$ Wigner matrices and intertwiners \cite{Rovelli:1994ge,Rovelli:1995ac}.However, it is often used  to indicate more broadly a generic wave-function in that Hilbert space. For a discussion of the gauge invariance of spin networks and their extensions under $\SU(2)$ transformations, the interested reader can refer to \cite{Charles:2016xwc}.
}
These are functions of one $\SU(2)$ group element along each link of the boundary graph. Let us call $L$ the number of links of $\Gamma$. The space of boundary wave-functions is provided with the natural scalar product defined by the Haar measure on $\SU(2)$,
\be
\psi(\{h_{l}\}_{l\in\Gamma})\in \cH_{\Gamma}
\,,\qquad
\cH_{\Gamma}=L^{2}(\SU(2)^{{\times L}})
\,,\qquad
\la \psi|\widetilde{\psi}\ra
=
\int_{\SU(2)^{{\times L}}}
\prod_{l\in\Gamma}\rd h_{l}\,\,
\overline{\psi(\{h_{l}\}}\,\,
\widetilde{\psi}(\{h_{l}\})
\,.
\ee
Each link $l\in\Gamma$ is dual to an edge on the boundary of the cellular complex $e\in\pp\Delta$. We will thus write equivalently $h_{l}$ and $h_{e}$, implicitly assuming the equivalence of labeling the group elements on the boundary by the edges $e\in\pp\Delta$ or links $l\in\Gamma=(\pp\Delta)^{*}$.

We define the Ponzano-Regge amplitude with  boundary  by ``closing'' the cellular decomposition $\Delta$ with the boundary state $\psi_{\pp \Delta}$:
\be
Z_{\Delta}[\psi_{\pp\Delta}]
\,=\,
\int\prod_{e\in\pp\Delta}\rd h_{e}\,\prod_{f\in\Delta}\rd g_{f}\,\,\,
\psi_{\pp\Delta}\big{(}\{h_{e}\}_{e\in\pp\Delta}\big{)}
\,
\prod_{e\in\pp\Delta}
\delta\left(h_{e}\,\overrightarrow{\prod_{f\ni e}} g_{f}\right)
\,
\prod_{e\in\Delta\setminus\pp\Delta}
\delta\left(\overrightarrow{\prod_{f\ni e}} g_{f}\right)\,.
\ee
The integration over the group elements on the boundary faces, $g_{f}$ for $f\in\pp\Delta$, automatically projects the wave-functions on the gauge invariant sector i.e. satisfying
\be
\forall \{g_{n}\}_{n\in\Gamma} \in \SU(2)^{N} ,\, \forall l \in \pp\Delta^* \quad
\psi\left(\left\lbrace g_{t(l)} h_{l} g_{s(l)}^{-1}\right\rbrace \right)
\,=\,
\psi(\left\lbrace h_{l}\right\rbrace ),
\ee
where  $s(l)$ (resp. $t(l)$) stands for the source (resp. target) node of the link $l$. The integer $N$ denotes the number of nodes of the boundary graph $\Gamma$. This means that all non-gauge invariant components of the boundary wave-function lead to a vanishing contribution in the integral $Z_{\Delta}[\psi_{\pp\Delta}]$. It is therefore natural to assume that the boundary wave-functions are all invariant under $\SU(2)$ gauge transformations at the boundary nodes.

The integration over the group elements on the bulk faces implies the flatness of the discrete connection on the boundary, i.e. the holonomy around every boundary vertex $v_{\pp}\in\pp\Delta$ is constrained to be the identity, with a $\delta$-distribution contribution to the integral
\be
\delta\left(\overrightarrow{\prod_{\substack{e\ni v_{\pp}\\e\in\pp\Delta}}}\, h_{e}\right)
\label{eq:flatness_bdry_vertex}
\, ,
\ee
as illustrated in figure \ref{fig:boundaryflatness}.
Indeed, each dual face on the boundary $\pp \Delta^*$ belongs to a polyhedron (a bubble dual to a vertex of the original cellular complex) whose all others faces are in the bulk. The holonomies around these faces are constrained by the model to be flat, which implies that the holonomy around the remaining face of the polyhedra be also flat.
Let us point out that this is not the only effect of the integration over the bulk group elements. The integration will lead to a specific distribution over the space of flat boundary connection, reflecting the topologies of the 3D cellular complex $\Delta$ and of its 2D boundary $\pp\Delta$. We will see this in more detail below, where we will focus on  2D boundaries with the topology of the 2-sphere and 2-torus.
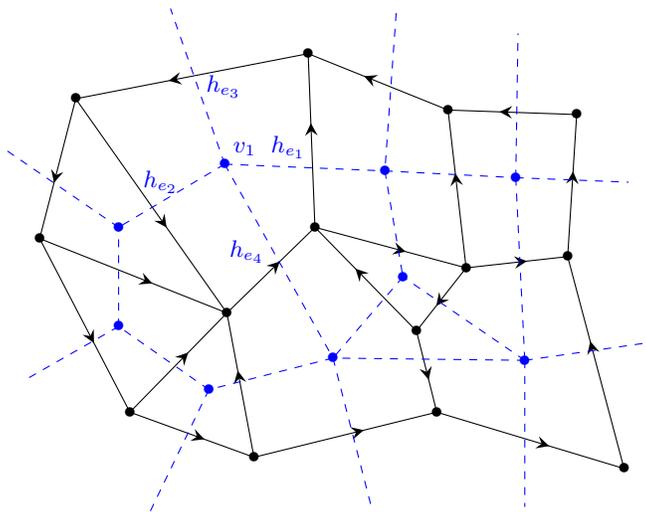
\begin{figure}[h!]
	\centering
	\begin{tikzpicture}[scale=3]

		\coordinate (A1) at (-0.39,0.81);
		\coordinate (A2) at (-0.55,0.19);
		\coordinate (A3) at (-0.15,-0.58);
		\coordinate (A4) at (0.4,-0.78);
		\coordinate (A5) at (1.21,-0.58);
		\coordinate (A6) at (2.04,-0.83);
		\coordinate (A7) at (1.79,0.11);
		\coordinate (A8) at (1.83, 0.74);
		\coordinate (A9) at (1.26, 0.76);
		\coordinate (A10) at (0.64, 1.01);
		\coordinate (I1) at (0.28, -0.14);
		\coordinate (I2) at (0.67, 0.24);
		\coordinate (I3) at (1.34, 0.06);
		\coordinate (I4) at (1.12, -0.22);
		
		\draw (A1) node{$\bullet$}; \draw (A2) node{$\bullet$}; 
		\draw (A3) node{$\bullet$}; \draw (A4) node{$\bullet$}; 
		\draw (A5) node{$\bullet$}; \draw (A6) node{$\bullet$}; 
		\draw (A7) node{$\bullet$}; \draw (A8) node{$\bullet$}; 
		\draw (A9) node{$\bullet$}; \draw (A10) node{$\bullet$}; 
		\draw (I1) node{$\bullet$}; \draw (I2) node{$\bullet$}; 
		\draw (I3) node{$\bullet$}; \draw (I4) node{$\bullet$}; 
		
		\link{A1}{A2};
		\link{A2}{A3};
		\link{A3}{A4};
		\link{A4}{A5};
		\link{A5}{A6};
		\link{A6}{A7};
		\link{A7}{A8};
		\link{A8}{A9};
		\link{A9}{A10};
		\link{A10}{A1};
		\link{A1}{I1};
		\link{A2}{I1};
		\link{A3}{I1};
		\link{A4}{I1}; 
		\link{I1}{I2};
		\link{I2}{A10};
		\link{I2}{I3};
		\link{I3}{A9};
		\link{I3}{A7};
		\link{I3}{I4};
		\link{I4}{A5};
		\link{I4}{I2};    
		
		\coordinate (B1) at (0.27, 0.52); \coordinate (B2) at (0.98, 0.49); \coordinate (B3) at (1.56, 0.46); \coordinate (B4) at (1.6, -0.35); \coordinate (B5) at (1.06, 0.02); \coordinate (B6) at (0.75, -0.34); \coordinate (B7) at (0.2, -0.48); \coordinate (B8) at (-0.2, -0.2); \coordinate (B9) at (-0.2, 0.24); \coordinate (B10) at (0.03, 1.21); \coordinate (B11) at (1.03, 1.19); \coordinate (B12) at (1.57, 1.1); \coordinate (B13) at (2.06, 0.44); \coordinate (B14) at (2.16, -0.27); \coordinate (B15) at (1.6, -1); \coordinate (B16) at (0.92, -1.0); \coordinate (B17) at (-0.06, -1.02); \coordinate (B18) at (-0.62, -0.44); \coordinate (B19) at (-0.71, 0.59); 
		
		\draw[dashed,blue] (B1)-- node[above,pos=0.4]{$h_{e_1}$} (B2); \draw[dashed,blue] (B1)-- node[above,pos=0.6]{$h_{e_2}$} (B9); \draw[dashed,blue] (B1)-- node[right,pos=0.5]{$h_{e_3}$} (B10); \draw[dashed,blue] (B1)-- node[left,pos=0.45]{$h_{e_4}$} (B6);
		\draw[dashed,blue] (B2)--(B3); \draw[dashed,blue] (B2)--(B5); \draw[dashed,blue] (B11)--(B2); 
		\draw[dashed,blue] (B3)--(B4); \draw[dashed,blue] (B3)--(B13); \draw[dashed,blue] (B3)--(B12);   
		\draw[dashed,blue] (B4)--(B5);  \draw[dashed,blue] (B4)--(B6); \draw[dashed,blue] (B4)--(B14); \draw[dashed,blue] (B4)--(B15);  
		\draw[dashed,blue] (B5)--(B6);
		\draw[dashed,blue] (B6)--(B7); \draw[dashed,blue] (B6)--(B16);
		\draw[dashed,blue] (B7)--(B8); \draw[dashed,blue] (B7)--(B17);  
		\draw[dashed,blue] (B8)--(B9); \draw[dashed,blue] (B8)--(B18);
		\draw[dashed,blue] (B9)--(B19); 
		
		\draw[blue] (B1) node{$\bullet$} node[above right]{$v_1$}; \draw[blue] (B2) node{$\bullet$}; 
		\draw[blue] (B3) node{$\bullet$}; \draw[blue] (B4) node{$\bullet$}; 
		\draw[blue] (B5) node{$\bullet$}; \draw[blue] (B6) node{$\bullet$}; 
		\draw[blue] (B7) node{$\bullet$}; \draw[blue] (B8) node{$\bullet$}; 
		\draw[blue] (B9) node{$\bullet$};
	\end{tikzpicture}
\caption{Induced flatness of the discrete connection on the 2D boundary, with the boundary spin network links in continuous line and the boundary of the cellular complex in blue dashed lines. A dual face on the boundary $\pp\Delta$ corresponds to a elementary loop of the boundary graph $\Gamma$, dual to a boundary vertex $v_{\pp}$. In the example above, the holonomy around the boundary vertex $v_{1}$ is given by the ordered product of the holonomies $h_{e_1}$, $h_{e_2}$, $h_{e_3}$ and $h_{e_4}$, associated to the links of the plaquette dual to $v_{1}$ and the flatness is imposed by \eqref{eq:flatness_bdry_vertex}.}
\label{fig:boundaryflatness}
\end{figure}

The Ponzano-Regge amplitude with boundary is still generally divergent, as for the case of a closed 3D cellular complex, due to the same reason of redundant $\delta$-distributions at the bulk vertices reflecting the gauge invariance of the theory under local 3D translations. This is addressed by gauge-fixing this symmetry along a maximal tree of the opened cellular complex $\Delta\setminus\pp\Delta$ (see \cite{Freidel:2004vi,Barrett:2008wh,Bonzom:2010ar} for more details). This procedure effectively amounts to removing as many redundant $\delta$-distributions in the bulk usually reducing the combinatorial structure of the cellular complex to a single bulk vertex without affecting its topology (see \cite{PRholo1,Barrett:2008wh} for examples).

Up to now, we have considered  boundary states in the connection polarization. We can also consider  boundary states in the metric polarization, and in particular with fixed induced metric data. This simply corresponds to a Fourier transform of the boundary wave-functions, switching to the spin network basis. Indeed, using the Peter-Weyl theoreom for the decomposition of $L^{2}$ functions over the $\SU(2)$ Lie group, one can decompose $\SU(2)$-gauge invariant wave-functions on the spin network basis:
\begin{equation}
\psi_{\pp\Delta}(\{h_{l}\}_{l\in\Gamma})
\,=\,
\sum_{\{j_{l},\cI_{n}\}}\psi_{\{j_{l},\cI_{n}\}}\Psi_{\pp\Delta}^{\{j_{l},\cI_{n}\}}(\{h_{l}\})
\qquad\textrm{with}\quad
\Psi_{\pp\Delta}^{\{j_{l},\cI_{n}\}}(\{h_{l}\})
\,=\,
\tr_{\{\cV^{j_{l}}\}}
\Big{[}
\bigotimes_{n}\cI_{n}\otimes\bigotimes_{l} \sqrt{d_{j_l}}  D^{j_l}(h_l) 
\Big{]}
\,,
\label{eq:def_spin_network}
\end{equation}
where the Fourier coefficients of the wave-function $\psi_{\pp\Delta}$ are $\psi_{\{j_{l},\cI_{n}\}}$. The Fourier transform is a sum over spins $j_{l}\in\N/2$ and over a basis of intertwiners living at the nodes $\cI_{n}$. The expression $D^j(h)$ stands for the Wigner matrix representing the group element $h\in\SU(2)$ in the spin-$j$ representation. We use the notation $d_{j}=2j+1$ for the dimension of the irreducible representation of spin $j$. 
The trace  stands for the contraction of all the magnetic indices as prescribed by the combinatorial structure of the graph $\Gamma$:
\be
\Psi_{\pp\Delta}^{\{j_{l},\cI_{n}\}}(\{h_{l}\})
\,=\,
\sum_{\{m^{s,t}_{l}\}}
\prod_{l}  \sqrt{d_{j_l}}  \la j_{l},m^{t}_{l}| h_{l}|j_{l},m^{s}_{l}\ra\,
\prod_{n}\left\la \{j_{l},m^{s}_{l}\}_{t(l)=n}|\cI_{n}|\{j_{l},m^{t}_{l}\}_{s(l)=n}\right\ra
\,.
\ee
These spin network wave-functions are normalized thanks to the factors $\sqrt{d_{j_l}}$,
\be
\la \Psi_{\pp\Delta}^{\{j_{l},\cI_{n}\}}|{\Psi}_{\pp\Delta}^{\{j'_{l},\cI'_{n}\}}\ra
=
\int \prod_{l\in\Gamma}\rd h_{l}\,
\overline{\Psi_{\pp\Delta}^{\{j_{l},\cI_{n}\}}}(\{h_{l}\})\,\,
{\Psi}_{\pp\Delta}^{\{j'_{l},\cI'_{n}\}}(\{h_{l}\})
=
\prod_{l}\delta_{j_{l},j'_{l}}
\,
\prod_{n}\la \cI_{n}|\cI'_{n}\ra
\,.
\ee
This leads to Ponzano-Regge amplitudes for basis states,
\be
Z_{\Delta}\big{[}\{j_{l},\cI_{n}\}_{l,n\in\pp\Gamma}\big{]}
=
Z_{\Delta}\big{[}\Psi_{\pp\Delta}^{\{j_{l},\cI_{n}\}}\big{]}
\,,
\ee
where the spins on the boundary are to be interpreted as the quantized lengths of the boundary edges $\ell_{e}=j_{e}\,l_{{\text{Planck}}}$ for $e\in\pp\Delta$.
The intertwiners determine the area and shape of the boundary faces.

\medskip

Committing a slight abuse of notation, we will sometimes write $\la Z_{\Delta} | \psi_{\pp\Delta} \ra=Z_{\Delta}[\psi_{\pp\Delta}]$ for the Ponzano-Regge amplitude with quantum boundary state $\psi_{\pp\Delta}$. This underlines that the Ponzano-Regge partition function for the 3D cellular complex $\Delta$ with boundary is to be understood as a linear form on the space of boundary states  $\psi_{\pp\Delta}$. This linear form  necessarily projects on flat boundary connections in a way that depends on the topology of the bulk $\Delta$.

One can then switch the logic between bulk and boundary and focus first on the boundary.
Starting with a 2D cellular complex, one can consider the space of linear forms on the space of spin network states living on it and wonder if the Ponzano-Regge partition function for 3D embeddings of this 2D cellular complex can cover this space of linear forms (or at least define a dense subset).
We leave this general mathematical question for future investigation and focus on the low genus topologies of a spherical and toroidal boundary topologies.

Let us start with a cellular decomposition of the 3-ball with a 2-sphere boundary, $\Delta=\cB_{3}$ with $\pp\Delta=\cS_{2}$. This case was already considered in the original work of Ponzano and Regge \cite{PR}.  The Ponzano-Regge amplitude simply amounts to the evaluation of the boundary wave-function on the identity. Indeed,  there does not exist any non-trivial flat connection on the 2-sphere: the only flat connection up to gauge transformations  is the trivial connection. In the spin basis, these amplitudes give the 3nj-symbols corresponding to the combinatorics of the boundary graph $\Gamma$.

\begin{prop}

For a cellular decomposition $\Delta$ of a 3-ball, with a connected boundary graph $\Gamma = (\pp \Delta)^{*}$, the Ponzano-Regge partition function for a quantum state on the boundary 2-sphere is given by its evaluation at the identity,
\be
\la Z_{\Delta} | \psi \ra_{\cB_{3}} = \psi(\{h_{l}=\id\}_{l\in\Gamma})
\,.
\ee
\end{prop}
\begin{proof}
 To prove this result, we start with the planar boundary graph $\Gamma$ and its dual 2D cellular decomposition of the 2-sphere, and rely on the simplest 3D cellular complex $\Delta$, compatible with this boundary data, defined with a single bulk vertex. We call $\Omega$ this bulk vertex and the full cellular decomposition $\Delta$ is built by drawing edges from $\Omega$ to every vertex on the boundary $v_{\pp}\in\pp\Delta$. We call $e_{v}$ the resulting bulk edge.  We distinguish boundary faces consisting of vertices and edges entirely in the boundary 2-sphere $\pp\Delta$ and bulk faces which involve the bulk vertex $\Omega$. These bulk faces are all triangles linking $\Omega$ to two boundary vertices, say $v^{1}_{\pp}$ and $v^{2}_{\pp}$, and thus made of two bulk edges $e_{v^{1}_{\pp}}$ and $e_{v^{2}_{\pp}}$ and the boundary edge $(v^{1}_{\pp}v^{2}_{\pp})$.
 
Thus a boundary edge $e_{\pp}\in\pp\Delta$ is shared by three faces: one bulk face which we call $f(e_{\pp})$ linking it to the bulk vertex $\Omega$, and two boundary faces corresponding to the source and target node of the link of the boundary graph $\Gamma$ dual to $e$, which we call its source and target  faces $s(e)$ and $t(e)$.
 
We associate a group elements $g_{f}$ to each bulk face and each boundary face, leading to the Ponzano-Regge partition function with boundary defined as:
\be
Z_{\Delta}[\psi_{\pp\Delta}]
\,=\,
\int\prod_{e_{\pp}\in\pp\Delta}\rd h_{e}\,\prod_{f\in\Delta}\rd g_{f}\,\,\,
\psi_{\pp\Delta}\big{(}\{h_{e_{\pp}}\}\big{)}
\,
\prod_{e_{\pp}}
\delta\left(h_{e_{\pp}}^{-1}\,g_{s(e)}g_{f(e_{\pp})}g_{t(e)}^{-1}\right)
\,
\prod_{v_{\pp}\in\pp\Delta}
\delta\left(\overrightarrow{\prod_{f\ni e_{v}}} g_{f}\right)
\,.
\ee
The integration over the boundary faces is completely absorbed by the $\SU(2)$ gauge-invariance of the boundary spin network, so that the $\delta$-distribution corresponding to a boundary edge $e_{\pp}$ simply identifies $h_{e_{\pp}}=g_{f(e_{\pp})}$. We can thus entirely write the Ponzano-Regge amplitude directly  in terms of the combinatorics of the boundary cellular complex:
\be
Z_{\Delta}[\psi_{\pp\Delta}]
\,=\,
\int\prod_{e\in\pp\Delta}\rd h_{e}\,
\psi_{\pp\Delta}\big{(}\{h_{e}\}\big{)}
\,
\prod_{v\in\pp\Delta}
\delta\left(\overrightarrow{\prod_{\substack{e\ni v\\e\in\pp\Delta}}} h_{e}\right)
=
\int\prod_{l\in\Gamma}\rd h_{l}\,
\psi_{\pp\Delta}\big{(}\{h_{l}\}\big{)}
\,
\prod_{p\in\pp\Delta}
\delta\left(\overrightarrow{\prod_{l\in p}} h_{l}\right)
\,,
\ee
where we have reformulated the integration variables in terms of the elements of the boundary graph $\Gamma$ and the plaquettes (i.e. dual faces) of the boundary complex $\pp\Delta$ in terms of cycles on the boundary graph.

The gauge-fixing of the Ponzano-Regge amplitude amounts to removing one $\delta$-distribution per bulk vertex, which means removing one $\delta$-distribution around a boundary plaquette. Let's call $p_{\emptyset}$ this gauge-fixed plaquette.
Calling $N$, $L$ and $P$ the number of nodes, links and plaquettes on the boundary $\Gamma$, this means that we are integrating over $L$ group elements $h_{l}$ with $(P-1)$ $\delta$-functions. One elegant way to see that this is enough to fix all the holonomies to the identity is to use the $\SU(2)$ gauge-invariance of the boundary state to fix as many group elements to the identity and then fix the remaining group elements to the identity by flatness.
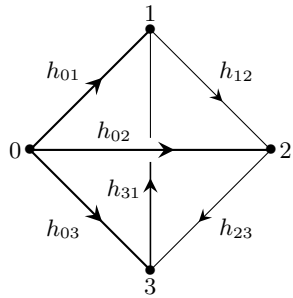
\begin{figure}[h!]
	\begin{tikzpicture}[scale=1.6]
	
		\coordinate (A0) at (0,0);
		\coordinate (A1) at (1,1);
		\coordinate (A2) at (2,0);
		\coordinate (A3) at (1,-1);
		\coordinate (a) at (1,0.1);
		\coordinate (b) at (1,-0.1);
			
		\draw (A0) node{$\bullet$} node[left]{0}; 
		\draw (A1) node{$\bullet$} node[above]{1};
		\draw (A2) node{$\bullet$} node[right]{2}; 
		\draw (A3) node{$\bullet$} node[below]{3};
		
\draw[thick,decoration={markings,mark=at position 0.6 with {\arrow[scale=1.5,>=stealth]{>}}},postaction={decorate}] (A0) -- node[above left,pos=.5]{$h_{01}$}(A1);
\draw[thick,decoration={markings,mark=at position 0.6 with {\arrow[scale=1.5,>=stealth]{>}}},postaction={decorate}] (A0) -- node[below left,pos=.5]{$h_{03}$}(A3);
\draw[thick,decoration={markings,mark=at position 0.6 with {\arrow[scale=1.5,>=stealth]{>}}},postaction={decorate}] (A0) -- node[above,pos=.35]{$h_{02}$}(A2);
		\alink {A2}{A3}{below right}{h_{23}};
		\alink {A1}{A2}{above right}{h_{12}};		
\draw[semithick,decoration={markings,mark=at position 0.85 with {\arrow[scale=1.5,>=stealth]{>}}},postaction={decorate}] (A3) -- node[left,pos=.75]{$h_{31}$}(b);
\draw (a)--(A1);

	\end{tikzpicture}

\caption{Gauge-fixing and flatness on the Tetrahedral graph embedded on the 2-sphere: we choose a node $0$ as the root and the maximal tree $T=\{(01),(02),(03)\}$ (in bold). The $\SU(2)$ gauge-invariance allows to fix the group variables on the tree links to $\id$. The remaining links $(12)$, $(23)$ and $(31)$ are respectively matched to the plaquettes $(012)$, $(023)$ and $(031)$, leaving the non-independent cycle $(123)$ as the gauge-fixed plaquette necessary to define a finite Ponzano-Regge amplitude.}
\label{fig:tetrasphere}
\end{figure}

We follow the gauge-fixing procedure for spin networks introduced in \cite{Freidel:2002xb}. We choose a root node $n_{0}$ and choose a maximum tree $T$ on the graph $\Gamma$. This tree $T$ goes through every node of the graph and consists in $L_{T}=N-1$ links. Moreover, for every node $n$, there exist a unique path along $T$ from the root node $n_{0}$ to $n$.

Using the $\SU(2)$ gauge-invariance at the nodes, we can fix every group elements on the tree links to the identity by using the invariance of the Haar measure under the $\SU(2)$ action, $h_{l\in T}=\id$. Using the sphere Euler characteristic $P-L+N=2$, this leaves $L-L_{T}=L-N+1=(P-1)$ links not belonging to the tree. These exactly corresponds to the plaquettes of the boundary graph. Indeed, there is a one-to-one correspondence between independent cycles on $\Gamma$ and links not belonging to the tree, with a link  $l\in\Gamma\setminus T$ matched with the loop following the path in $T$ from the root node $n_{0}$ to its source node $s(l)$, then along the link $l$, and back along the tree $T$ from the target node $t(l)$ to the root node $n_{0}$.
We illustrate this for the case of the tetrahedral case in figure \ref{fig:tetrasphere}.
At the end, combining the $\SU(2)$ gauge invariance and the $\delta$-functions on the boundary plaquettes allows us to fix all the group elements to the identity and get the required result, $\la Z_{\Delta} | \psi \ra_{\cB_{3}} = \psi(\{h_{l}=\id\}_{l\in\Gamma})$.

\end{proof}

It has been further proved that the generating functions for these ``planar'' Ponzano-Regge amplitudes for a 3-valent boundary graph $\Gamma$ are dual to the 2D Ising model living on the 2D cellular complex \cite{Bonzom:2015ova,costantino}. This duality involves the same class of coherent spin network states that we will use in the present work and that have been shown to allow for exact analytic computations of spinfoam amplitudes \cite{Freidel:2012ji,Bonzom:2012bn,Dittrich:2013jxa}.
Finally, general results for the Ponzano-Regge partition function for the 3-sphere and handlebodies can be found in \cite{Freidel:2005bb,Barrett:2008wh,Barrett:2011qe}.

\medskip

Here, we are interested in the next non-trivial topology beyond the 3-ball. We consider a solid torus, i.e. a 3D solid cylinder with a disk base whose final disk is glued back on the initial disk, so that the 2D boundary is a torus.
The boundary graph is defined as follows. We consider an open graph on a tube (i.e. 2D cylinder) such that it has $P$ links puncturing the initial circle and similarly $P$ links puncturing the final circle. The cyclical order of the links on the initial and final circles is important. We then glue these initial and final links together respecting the orientation of the circles, as drawn in figure \ref{fig:twistedcylinder}. A twist is obtained by off-setting the identification between initial and final links. 

Although there might not exist a natural way to distinguish a ``zero-twist'' configuration for an arbitrary graph on the cylinder, we can always talk about the difference of twist between two different gluing of the solid cylinder into a solid torus.
Let us underline that different twists a priori lead to different closed graphs once the opened links are glued along the cut.
\begin{figure}[h!]
	\includegraphics[scale=0.65]{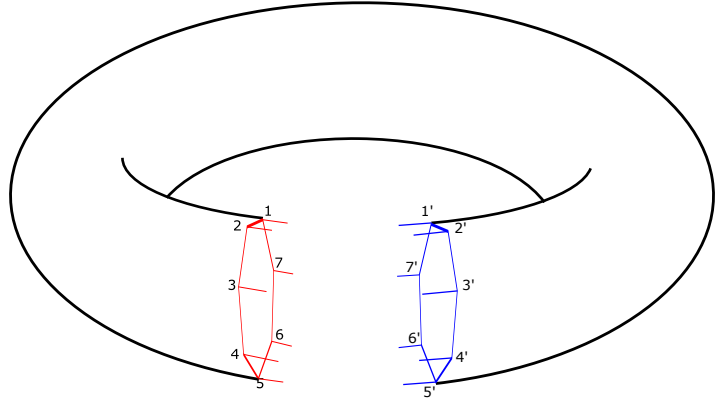}
	\caption{Example of an initial and final circles for an open graph on a tube for $P=7$ punctures. These initial and final circles are identified in order to reconstruct the torus. A twist is obtained by off-setting the identification the initial and final circles.}
	\label{fig:twistedcylinder}
\end{figure}

The  solid torus has two cycles, the contractible cycle around the tube and the non-contractible cycle along the tube.
The resulting Ponzano-Regge amplitude is given by the integration over the holonomy along the non-contractible cycle. This result for an arbitrary graph on the torus extends the computation for a regular square lattice done in \cite{PRholo1,PRholo2} and is a particular case of the general formula for the Ponzano-Regge amplitude for handlebodies proven in \cite{Freidel:2005bb,Dowdall:2009eg}.

\begin{prop}
\label{prop:torus}
For a cellular decomposition $\Delta$ of a solid torus, with $\Gamma = (\pp \Delta)^*$ the Ponzano-Regge partition function for a quantum state on the boundary 2-torus is given by its integral over the non-contractible cycle. Let us call $\cC$ the initial and final circle on which the gluing is performed to create the graph on the torus from the graph on the cylinder. We distinguish the links that cross this circle, which we orient all in the same direction, from the remaining links.Then the Ponzano-Regge amplitude reads:
\be
\label{Ztorus}
\la Z_{\Delta} | \psi \ra_{\cT}
=
\int \rd H\,\,\psi\Big{(}\{h_{l}=H\}_{l\in\cC},\{h_{l}=\id\}_{l\in\Gamma\setminus \cC^*}\Big{)}
\,.
\ee
where $\cC^*$ corresponds to the set of links in $\Gamma$ that cross the circle $\cC$.
\end{prop}
\begin{proof}
The proof is very similar to the case of the 3-ball. 
Starting from the   2D boundary complex, dual to the boundary graph $\Gamma$ with its $N$, $L$ and $P$ nodes, links and plaquettes, an admissible cellular decomposition of the 3D cylinder is  obtained by just adding a single bulk vertex $\Omega$ and a central edge or ``bone'' going from $\Omega$ to itself defining the non-contractible cycle of the solid cylinder. An example is given in figure \ref{fig:twovertextorus}.
\begin{figure}[h!]
\begin{subfigure}[t]{.45\linewidth}
	\centering
	\begin{tikzpicture}[scale=.75]
	
		\coordinate (A0) at (0,0);
		\coordinate (A1) at (0,4);
		\coordinate (A2) at (5,4);
		\coordinate (A3) at (5,0);
		\draw[opacity=0.6] (A0)--(A1)--(A2)--(A3)--(A0);

		\coordinate (a) at (2.5,1.25);
		\coordinate (b) at (2.5,2.75);
			
		\draw (a) node{$\bullet$}; 
		\draw (b) node{$\bullet$};

		\coordinate (N) at (2.5,4);
		\coordinate (S) at (2.5,0);
		\coordinate (Ea) at (5.3,1.25);
		\coordinate (Eb) at (5.3,2.75);
		\coordinate (Wa) at (-.3,1.25);
		\coordinate (Wb) at (-.3,2.75);

\draw[thick,decoration={markings,mark=at position 0.6 with {\arrow[scale=1.5,>=stealth]{>}}},postaction={decorate}] (a) -- node[left,pos=.5]{$h_{3}$}(b);
\draw[decoration={markings,mark=at position 0.6 with {\arrow[scale=1.5,>=stealth]{>}}},postaction={decorate}] (b) -- node[left,pos=.5]{$h_{4}$}(N);
\link {S}{a};
\alink{a}{Ea}{above}{h_{2}};
\alink{b}{Eb}{above}{h_{1}};
\link{Wa}{a};
\link{Wb}{b};

	\end{tikzpicture}
\caption{Gluing the horizontal lines straight, we compute the integral of the wave-function $\psi(\{h_{l}\}_{l=1..4}$ contracted with $\delta(h_{1}h_{3}^{-1}h_{2}h_{3})$, $\delta(h_{1}h_{4}h_{2}h_{4}^{-1})$ for the boundary flatness and $\delta(h_{3}h_{4})$ for the cycle around the central bone. Gauge-fixing $h_{3}=\id$ by a $\SU(2)$ transformation,  this leads to an amplitude $\int \rd H\,\psi(H,H,\id,\id)$.}
\end{subfigure}
\hspace{4mm}
\begin{subfigure}[t]{.45\linewidth}
	\centering
	\begin{tikzpicture}[scale=.75]
	
		\coordinate (A0) at (0,0);
		\coordinate (A1) at (0,4);
		\coordinate (A2) at (5,4);
		\coordinate (A3) at (5,0);
		\draw[opacity=0.6] (A0)--(A1)--(A2)--(A3)--(A0);

		\coordinate (a) at (2.5,1.25);
		\coordinate (b) at (2.5,2.75);
			
		\draw (a) node{$\bullet$}; 
		\draw (b) node{$\bullet$};

		\coordinate (N) at (2.5,4);
		\coordinate (S) at (2.5,0);
		\coordinate (Ea) at (5,1.25);
		\coordinate (Eb) at (5,2.75);
		\coordinate (Wa) at (0,1.25);
		\coordinate (Wb) at (0,2.75);

\draw[thick,decoration={markings,mark=at position 0.6 with {\arrow[scale=1.5,>=stealth]{>}}},postaction={decorate}] (a) -- node[left,pos=.5]{$h_{3}$}(b);
\draw[decoration={markings,mark=at position 0.6 with {\arrow[scale=1.5,>=stealth]{>}}},postaction={decorate}] (b) -- node[left,pos=.5]{$h_{4}$}(N);
\link {S}{a};
\alink{a}{Ea}{above}{h_{2}};
\alink{b}{Eb}{above}{h_{1}};
\link{Wa}{a};
\link{Wb}{b};
\draw[rounded corners=3 pt] (Eb) --(5.4,2.75)-- (5.6,1.25)--(5.9,1.25);
\draw[rounded corners=3 pt] (-.9,2.75) --(-.6,2.75)-- (-.4,1.25)--(0,1.25);
\draw[rounded corners=3 pt] (Ea) --(5.4,1.25)-- (5.5,.2);
\draw[rounded corners=3 pt] (-.5,3.7) --(-.4,2.75)-- (Wb);

	\end{tikzpicture}

\caption{Gluing the horizontal lines with a twist, the Ponzano-Regge integral involves the product of distributions $\delta(h_{1}h_{4}^{-1}h_{2}h_{3})\delta(h_{1}h_{3}h_{2}h_{4}^{-1})\delta(h_{3}h_{4})$ still leading to the same formula for the amplitude in terms of the boundary wave-function  $\int \rd H\,\psi(H,H,\id,\id)$.}
\end{subfigure}
\caption{The two-node boundary graph on the torus  and the resulting calculation of the Ponzano-Regge amplitude for  the straight and twisted versions.}
\label{fig:twovertextorus}
\end{figure}

Without going into the details of the holonomies and $\delta$-functions as in the case of the 3-ball and the 2-sphere, we prefer to stress the differences with the previous proof. For the same number of nodes and links, we will have two less plaquettes on the boundary, due to the smaller Euler characteristic, $P-L+N=0$. This corresponds to two $\delta$-functions less in the Ponzano-Regge integral. But there will be an extra $\delta$-function coming from the holonomy around the central bone, corresponding to the contractible cycle of the solid cylinder. Following the same gauge-fixing procedure, using the $\SU(2)$ gauge transformations of the boundary state to set as many group variables as possible to the identity by choosing  a maximal tree in $\Gamma\setminus \cC^*$, the fact that we have one less $\delta$-function translates into one group element not set to the identity. This group element is the one living on the boundary links around the cut circle $\cC$. Calling it $H$, we easily see that it is the same group element on all these links due to the flatness of the boundary connection.  This leads to the desired result, as illustrated in figure \ref{fig:twovertextorus}.

\end{proof}

Another way to describe the 2-torus geometry with twists and construct boundary graphs is to start directly with the 2-torus instead of the 2-cylinder. Instead of taking the simple circle slice with zero slope, thus winding only once, one can choose a cut circle with non-vanishing rational slope, thus non-trivially winding around the torus. The foliation of the torus obtained by sliding this cut along the torus defines a Morse function, allowing to define a 2D flat coordinate systems on the torus. This can be understood as a Dehn twist winding the cylinder geometry potentially several times (in integer multiples of $2\pi$) around its main axis and amounts to a modular transformation -or large diffeomorphisms- acting on the original $(t,x)$ coordinate system of the torus. This method, as described in appendix \ref{app:torus} leads to  boundary linear forms which depend on the Dehn twist acting on boundary states living on the same boundary graph drawn on the torus. All these boundary linear forms still project on the moduli space of flat connections but define a priori different projectors corresponding to different windings of the torus around its non-contractible cycle\footnotemark.
We illustrate this on simple examples (the tetrahedral graph embedded on the 2-torus and the lattice with two squares) in appendix \ref{app:torus}.
\footnotetext{%
Calling $H_{1}$ and $H_{2}$ the holonomies around the two cycles on the 2-torus, satisfying the flatness condition $H_{1}H_{2}=H_{2}H_{1}$. The original linear form obtained in \eqref{Ztorus} corresponds to considering $H_{1}=H$ and $H_{2}=\id$. Winding the second cycle around the first cycle, thereby creating a twist, amounts to considering $H_{1}=H$ and $H_{2}=H^{m}$ for a winding number $m$. This modifies the Ponzano-Regge boundary amplitude as discussed in the appendix.
}
It would be interesting to compare this approach, which allows to directly  study the behavior of the Ponzano-Regge amplitude under large boundary diffeomorphisms, with the simpler procedure for constructing boundary graphs and boundary states used here. We leave this for future investigation.

\medskip

The rest of the paper will consist in applying the above proposition \ref{prop:torus} to coherent spin network states on a square lattice on the twisted torus and showing that the resulting Ponzano-Regge amplitude with boundary can be exactly computed with an explicit dependence on the twist angle.

\section{Quantum Boundary States as Coherent Spin Networks}

The Ponzano-Regge amplitude for a toroidal boundary was already considered in \cite{PRholo2}. More precisely, this previous work considered  a restricted choice of boundary spin network states on a square lattice defined at fixed spins (i.e. fixed edge lengths) and with coherent intertwiners (flexible squares as made from two triangles glued together). This allows to perform a saddle point analysis and derive an approximative formula in the asymptotic limit for large spins.

Here, we go one step further and consider fully coherent spin networks, with coherent superpositions of both spins and intertwiners. These states involve scale parameters defining Poisson probability distributions on the spins and become scale-invariant in a critical limit of those scale parameters. Moreover, they can be understood as generating functions for the semi-coherent Ponzano-Regge amplitdues at fixed spins. Finally, we will show in the next section that they allow for an exact analytical computation of the Ponzano-Regge amplitudes.

%

\subsection{From Coherent Intertwiners to Coherent Spin Superpositions}
\label{cohstates}

To construct coherent boundary spin network states, we use the spinor or holomorphic formalism. 
This formalism was developed for loop gravity \cite{Borja:2010rc,Livine:2011gp,Livine:2013wmq}, and led to the definition of coherent intertwiners which can be glued into coherent spin networks  \cite{Livine:2007vk,Freidel:2009nu,Freidel:2010tt,Dupuis:2010iq,Bianchi:2010gc,Bonzom:2012bn,Alesci:2016dqx}.
 We denote by $|w \ra \in \C^2$ a spinor, by $\la w |$ its conjugate and by $| w ] $ its dual
\begin{equation}
	| w \ra = 
	\begin{pmatrix}
	w^0 \\
	w^1
	\end{pmatrix}\,,
	\qquad
	\la w | = (\bar{w}^0 \; \; \bar{w}^1)
\,,\qquad
	| w ] = \varsigma | w \ra = 
	\begin{pmatrix}
	-\bar{w}^1 \\
	\bar{w}^0 
	\end{pmatrix}
	 \,,
\end{equation}
where  $\varsigma$ is the $\SU(2)$ structure map. It is an anti-unitary map $\varsigma^2 = -1$, which commutes with the $\SU(2)$ action on the spinors, $g\,\varsigma\,| w \ra=\varsigma\,g\,| w \ra$ for all group elements $g \in \SU(2)$.
The dual spinors satisfy the following scalar product identities:
\begin{equation}
	[w|\tilde{w}] = \la \tilde{w} | w \ra \; , \quad [w|\tilde{w} \ra = - [\tilde{w} | w \ra \; .
	\nn
\end{equation}
A spinor $w$ is equivalent to a $\SU(2)$ group element up to a norm factor:
\be
|w\ra=\sqrt{\la w|w \ra}\,G_{w}\,|\uparrow\ra\,,
\qquad\textrm{with}\quad
|\uparrow\ra=\mat{c}{1 \\ 0}
\quad\textrm{and}\quad
G_{w}=\f1{\sqrt{\la w|w \ra}}\mat{cc}{w^{0} & -\bar{w}^{1}\\ w^{1} & \bar{w}^{0}}\,,
\ee
where the corresponding group element $G_{w}$ is  the $\SU(2)$ transformation mapping the orthonormal basis $(|\uparrow\ra,|\downarrow\ra)$ onto the orthonormal basis $(|w\ra,|w])$ up to the spinor norm factor.
Finally, we can associate a vector $\vec{u}\in\R^3$ to every spinor by projecting it over the Pauli matrices $\sigma_{i=1,2,3}$,
\begin{equation*}
	\vu_{w} = \la w | \vec{\sigma} | w \ra \,,\qquad
	|\vu_{w}|^{2}=\la w|w\ra
	\,,
\end{equation*}
where the Pauli matrices are normalized such that $(\sigma_i)^{2}=\id$. This map $\C^{2}\rightarrow\R^{3}$ is clearly not one-to-one but erases the phase of the spinor. Indeed it is $\U(1)$-invariant, $\vu_{w}=\vu_{e^{i\theta}w}$ for every phase $e^{i\theta}\in\U(1)$.

\medskip

First, we define $\SU(2)$ coherent states on the Hilbert space $\cV^{j}$ for the $\SU(2)$-representation of spin $j$.
Calling $J_{i}$  the $\su(2)$ generators, corresponding to half the Pauli matrices in the fundamental representation of spin $\f12$, the highest weight state of spin $j$ can be interpreted as a semi-classical 3-vector of norm $j$ aligned on the $z$-axis,
\be
|j,j\ra=|\f12,\f12\ra^{\otimes 2j}=|\uparrow\ra^{\otimes 2j}
\,,\quad
\la j,j | \vJ| j,j\ra =2j \la \uparrow|\f\vsigma2|\uparrow\ra=j\mat{c}{0 \\ 0 \\ 1}
\,,\quad
\sqrt{\la j,j | \vJ^2| j,j\ra -\la j,j | \vJ| j,j\ra^2 }=\sqrt{j}
\,.
\ee
We act on this highest weight state to define $\SU(2)$ coherent states defining semi-classical quantum 3-vectors of norm $j$ with arbitrary direction:
\be
\forall\xi\in\C^2\,,\qquad
|j,\xi\ra
=
\sqrt{\la\xi|\xi\ra}^{2j}\,G_{\xi}\,|j,j\ra
=
|\f12,\xi\ra^{\otimes 2j}
\,,\qquad
\la j,\xi| \vJ| j,\xi\ra
=
j\,\f{\vu_{\xi}}{|\vu_{\xi}|}
\,.
\ee

The second step is to introduce coherent intertwiners for every node of the boundary graph. Following \cite{Livine:2007vk,Freidel:2010tt,Dupuis:2010iq}, we consider tensor products of $\SU(2)$ coherent states and group average them to project them on the $\SU(2)$-invariant space of intertwiners. So, considering a node $n$ with $E_{n}$ links attached to it, all considered as outgoing to start with and each carrying a spin $j_{l}$ with $l$ running from 1 to $E_{n}$, a coherent intertwiner depends on the choice of $E_{n}$ spinors $\xi_{l}$ and is defined as:
\be
\cI_{n}[\{\xi_{l}\}_{l=1..E_{n}}]
\,=\,
\int_{\SU(2)}\rd g\,
g\act\bigotimes_{l=1}^{E_{n}} |j_{l},\xi_{l} \ra
\,=\,
\int_{\SU(2)}\rd g\,
\bigotimes_{l=1}^{E_{n}}g\,|j_{l},\xi_{l} \ra
\,.
\ee
To account for incoming links, we adjust this definition by taking the dual of the incoming spinors:
\be
\label{def:cohintertwiner}
\cI_{n}[\{\xi_{l}\}_{l \ni n}]
\,=\,
\int_{\SU(2)}\rd g\,
\bigotimes_{l,\,s(l)=n}
g\, |j_{l},\xi_{l} \ra
\otimes
\bigotimes_{l,\,t(l)=n}
[j_{l},\xi_{l}|\,g^{-1}
\,.
\ee

The next step is to glue those coherent intertwiners at fixed spins along the graph links. Following the original work \cite{Livine:2007vk}, we refer to such spin networks as Livine-Speziale (LS) spin networks. They are defined by choosing a spinor $\xi_{l}^{n}\in\C^{2}$ around each vertex $n$ for each link $l\ni n$ attached to it. This amounts to a spinor at every end of each link, so we equivalently write $\xi_{l}^{s,t}$ for the spinors at the source and target vertices of the link $l$. Gluing the corresponding coherent intertwiners by the $\SU(2)$ group elements $h_{l}$ along the links gives a spin network wave-function:
\be
\label{eq:LS_boundary_state}
\psi_{LS}^{\{j_{l},\xi^{n}_{l}\}}(\{h_l\}_{l\in\Gamma})
\,=\,
\int_{SU(2)} \prod_{n} \rd g_{n}
\,
\prod_{l}\, [j_l, \xi_{l}^{t} |g^{-1}_{t(l)} h_l g_{s(l)} |j_l, \xi_{l}^{s} \ra
\,=\,
\int_{SU(2)} \prod_{n} \rd g_{n}
\,
\prod_{l}\, [ \xi_{l}^{t} |g^{-1}_{t(l)} h_l g_{s(l)} | \xi_{l}^{s} \ra^{2 j_l}
\,.
\ee
Such a boundary state can be interpreted as a quantized discrete geometry on the 2D boundary cellular complex $\pp\Delta$ (see \cite{Freidel:2010aq} for their interpretation in terms of twisted geometries). The length of an edge $e\in\pp\Delta$ is given in Planck unit by the spin $j_{l}$ of the link $l$ dual to it. Each face $f\in\pp\Delta$ is a polygon whose shape is encoded in the intertwiner living at the node $n$ dual to the face. As illustrated in figure \ref{fig:spinornormal}, the 3-vector defined by the spinors  carried by the links around the face give the semi-classical edge vectors around the polygon.
\begin{figure}[h!]
		\begin{tikzpicture}[scale=1,>=latex]
		\coordinate (A1) at (-2.1,0.61);
		\coordinate (A2) at (-0.22,1.95);
		\coordinate (A3) at (2.99,1.07);
		\coordinate (A4) at (1.89,-1.06);
		\coordinate (A5) at (-0.93,-1.98);
		
		\coordinate (I) at (0,0);
		\coordinate (N1) at (-1.38,1.84);
		\coordinate (N2) at (2.11,2.23);
		\coordinate (N3) at (3.93,-0.06);
		\coordinate (N4) at (0.81,-2.69);
		\coordinate (N5) at (-2,-1);
		
		\draw[dashed,color=blue,->] (A1)-- node[pos=1,above]{$\vec{N}_1$} (A2);
		\draw[dashed,color=blue,->] (A2)-- node[pos=1,right]{$\vec{N}_2$} (A3);
		\draw[dashed,color=blue,->] (A3)-- node[pos=1,below right]{$\vec{N}_3$} (A4); 
		\draw[dashed,color=blue,->] (A4)-- node[pos=1,below]{$\vec{N}_4$} (A5); 
		\draw[dashed,color=blue,->] (A5)-- node[pos=1,left]{$\vec{N}_5$} (A1);
		
		\draw[dotted,color=red] (A2)--(A5); \draw[dotted,color=red] (A2)--(A4);
		
		\node at (I) {$\bullet$};
		\draw[->-=0.5,>=stealth] (I)-- (N1); \draw[->-=0.5,>=stealth] (I)--(N2); \draw[->-=0.5,>=stealth] (I)--(N3); \draw[->-=0.5,>=stealth] (I)--(N4); \draw[->-=0.5,>=stealth] (I)--(N5);
		
		\node at (N1) [above,scale=1.2]	{$|j_1, \xi_1 \ra$};
		\node at (N2) [above,scale=1.2]	{$|j_2, \xi_2 \ra$};
		\node at (N3) [right,scale=1.2]	{$|j_3, \xi_3 \ra$};
		\node at (N4) [below,scale=1.2]	{$|j_4, \xi_4 \ra$};
		\node at (N5) [below left,scale=1.2]	{$|j_5, \xi_5 \ra$};
		\end{tikzpicture}
	\caption{
Coherent intertwiners of a LS spin network as semi-classical polygons: example of a 5-valent intertwiner in black and its dual pentagon in dashed blue.
Each link $l$ (here chosen with outgoing orientation) carries a spinor $|j_l,\xi_l \ra$, which defines an edge vector of the polygon $\vec{N}_l =\la j_l,\xi_l | \vec{J} | j_l,\xi_l \ra$ of length $|\vec{N}_l|=j_{l}$.
Note that the edge vectors are not necessarily co-planar and that the corresponding polygon is thus not planar in general. For instance, splitting the pentagon into three triangles along the dotted red lines, there can be a non-vanishing dihedral angle between adjacent triangles. This splitting corresponds to the decomposition of the 5-valent intertwiners into three 3-valent intertwiners.}
	\label{fig:spinornormal}
\end{figure}
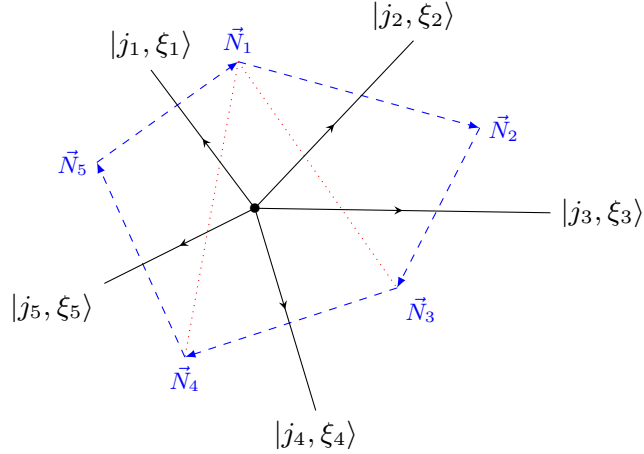

The final step is to define coherent spin networks with spin superpositions.
To this purpose, we introduce a complex coupling $y_l$ on each link, which will play the role of a conjugate variable to the spins $j_{l}$ and control its probability distribution:
\begin{equation}
	\psi^{\{y_{l},\xi^{n}_{l}\}}(\{h_l\})
	=
	\sum_{\{j_l\}} \omega\big{[}\{j_l\}\big{]} \left(\prod_{l} y_l^{2 j_l}\right) \psi_{LS}^{\{j_{l},\xi^{n}_{l}\}}(\{h_l\})
	\; ,
\end{equation}
where the weight $\omega\big{[}\{j_l\}\big{]}$ is a combinatorial factor depending on the spins,
\begin{equation}
	\omega\big{[}\{j_l\}\big{]} = \prod_{n} \f{(J_n+1)!}{\prod_{l\ni n} \sqrt{(2 j_l)}!}
	= \f{\prod_{n} (J_n+1)!}{\prod_{l} (2 j_l)!}
	\quad\textrm{where $J_n = \sum_{l \ni n} j_l$ is the sum of the spins around each node}.
	\label{eq:weight_general}
\end{equation}
Up to the node factorial contribution  $(J_n+1)!$, this weight defines a Poisson distribution on the spins.

Such coherent states have the natural mathematical interpretation as a generating function for LS spin networks \cite{Dupuis:2010iq,Dupuis:2011fz,Bonzom:2012bn}.
The sum over the spins $\{j_{l}\}$ defines a series with a non-zero radius of convergence \cite{Freidel:2012ji} and the boundary state is defined outside of the radius of convergence as the analytic continuation of the series. Such a coherent state is interpreted as a superposition of discrete geometries on the 2D boundary with different edge lengths \cite{Bonzom:2015ova}. The probability distribution of those edge lengths is given by the  weight $ \omega\big{[}\{j_l\}\big{]}$, the couplings $\prod_{l} y_l^{2 j_l}$ and the norm of the LS spin networks $\psi_{LS}^{\{j_{l},\xi^{n}_{l}\}}(\{h_l\})$.
Indeed, we can compute the expectation value of any observable $\cO$ of the spins on a state $\psi^{\{y_{l},\xi^{n}_{l}\}}$, 
\be
\la\cO[\{j_l\}]\ra_{\{y_{l},\xi^{n}_{l}\}}
=
\f{\la \psi^{\{y_{l},\xi^{n}_{l}\}}|\cO[\{j_l\}]|\psi^{\{y_{l},\xi^{n}_{l}\}}\ra}{\la \psi^{\{y_{l},\xi^{n}_{l}\}}|\psi^{\{y_{l},\xi^{n}_{l}\}}\ra}
=
\f{\sum_{\{j_l\}}\cO[\{j_l\}]\cP_{\{y_{l},\xi^{n}_{l}\}}[\{j_l\}]}{\sum_{\{j_l\}}\cP_{\{y_{l},\xi^{n}_{l}\}}[\{j_l\}]}
\,,
\ee
where the probability distribution for the spins is computed directly in terms of the norm of the coherent intertwiners:
\be
\cP_{\{y_{l},\xi^{n}_{l}\}}[\{j_l\}]
=
\omega\big{[}\{j_l\}\big{]}^{2} \left(\prod_{l} y_l^{4 j_l}\right)
\,\prod_{l}\f1{(2j_{l}+1)}
\,\prod_{n} \big{\la} \cI_{n}[\{\xi_{l}\}_{l \ni n}]
\,\big{|}\,
\cI_{n}[\{\xi_{l}\}_{l \ni n}]\big{\ra}
\,.
\ee
The norm of a coherent intertwiner was computed in \cite{Freidel:2010tt,Bonzom:2012bn} from its generating function, here keeping the node index $n$ implicit:
\be
\big{\la} \cI[\{\xi_{l}\}]
\,\big{|}\,
\cI[\{\xi_{l}\}]\big{\ra}
=
\int_{\SU(2)} \rd g\,
\prod_{l}\la\xi_{l}|g|\xi_{l} \ra^{2j_{l}}
\,,\qquad
\int_{\SU(2)} \rd g\,\,
e^{\sum_{l}\la\xi_{l}|g|\xi_{l} \ra}
=
\sum_{J\in\N}
\f1{J!(J+1)!}
\left(\f12\sum_{l,\tl}\big|[\xi_{l}|\xi_{\tl}\ra\big|^{2}\right)^{J}
\,,
\ee
from which we can extract the term of power $2j_{l}$ in each spinor $\xi_{l}$ to get the norm of a coherent intertwiner. Aside the exact expression, it is actually straightforward to derive its behavior at large spins by computing the saddle point approximation of the integral over $\SU(2)$ as one rescales homogeneously the spins by a large factor, $j_{l}\mapsto \lambda j_{l}$ with $\lambda \rightarrow +\infty$. As shown in \cite{Livine:2007vk}, the maximum of the integrand $\prod_{l}\la\xi_{l}|g|\xi_{l} \ra^{2j_{l}}$ is necessarily reached at $g=\id$ and this is a stationary point if and only if the spinors satisfy the closure condition
\be
\vcC=\sum_{l}j_{l}\f{\la\xi_{l}|\vsigma|\xi_{l}\ra}{\la\xi_{l}|\xi_{l}\ra}
=\sum_{l}j_{l}\f{\vu_{\xi_{l}}}{|\vu_{\xi_{l}}|}
=0
\,.
\ee
Let us stress that the closure conditions are invariant under homogeneous rescaling of the spins $j_{l}\mapsto \lambda j_{l}$.
Thus, as confirmed by numerics, the behavior is very different whether  the closure condition is satisfied or not. On the one hand, when the closure vector vanishes, $\vcC=0$, it means geometrically that the edge vectors around the node actually close and form a polygon. In that case, the norm $\big{\la} \cI[\{\xi_{l}\}] \,\big{|}\,\cI[\{\xi_{l}\}]\big{\ra}$ generically behaves as $\lambda^{-\f32}\prod_{l}\la\xi_{l}|\xi_{l}\ra^{2j_{l}}$ .  On the other hand, the case when the closure vector does not vanish, the integral is exponentially suppressed and the coherent intertwiner norm decreases as $\exp[-\lambda^{2}|\vcC|^{2}]$.

This means that the effect of the coherent intertwiner norm in the spin probability distribution is two-fold. First, it selects the spins configuration such that the closure conditions are satisfied around every node $n$.
Second, putting aside the algebraic term and focusing on the leading order exponential behavior, it produces spinor norm factors $\la\xi_{l}|\xi_{l}\ra^{2j_{l}}$ that can be reabsorbed in the couplings $y_{l}\rightarrow Y_{l}=\sqrt{\la\xi^{n}_{l}|\xi^{n}_{l}\ra}\, y_{l}$ upon normalizing all the spinors $\xi^{n}_{l}\rightarrow \sqrt{\la\xi^{n}_{l}|\xi^{n}_{l}\ra}\,\hat{\xi}^{n}_{l}$.

Putting all the contributing factors together, combinatorial weight, couplings and intertwiner norms, using the Stirling formula approximating the factorials in the combinatorial weight at large spins, assuming that the spins satisfy the closure conditions around every node, and finally focusing on the exponential factors and considering the algebraic factors as sub-leading,  the spin probability distribution behaves at large spins as:
\be
\cP_{\{y_{l},\xi^{n}_{l}\}}[\{j_l\}]
\sim
e^{2\Phi[\{j_l\}]}
\qquad\textrm{with}\quad
\Phi[\{j_l\}]
=
\sum_{l }2j_{l}(\ln Y_{l}-\ln j_{l} ) +\sum_{n}J_{n}\ln J_{n}
\,.
\ee
The extrema of this probability distribution is given by a vanishing derivative with respect to each spin $j_{l}$:
\be
\forall l\in\Gamma\,,\qquad
\pp_{j_{l}}\Phi=0
\quad\Longleftrightarrow\quad
\ln \f{Y_{l}^{2}J_{s(l)}J_{t(l)}}{4 j_{l}^{2}}=0
\quad\Longleftrightarrow\quad
\f{j_{l}^{2}}{J_{s(l)}J_{t(l)}}=\f{Y_{l}^{2}}4
\,.
\ee
This equation for the peaks of the spin  probability distribution has the crucial property of being invariant homogeneous rescaling of the spins $j_{l}\mapsto \lambda j_{l}$. This scale-invariance of the extrema of the   probability distribution defined by this class of coherent spin networks was put forward in \cite{Bonzom:2015ova}. The fixed point equations for the spins induce non-trivial constraints on the couplings $Y_{l}$. But, once  the couplings allow for a solution for the spins $\{j_{l}\}$, then there exist a whole line of solution obtained by arbitrary global rescalings of the spins.
 In the case of a planar 3-valent graph, these constraints determine a triangulation whose angles are fixed in terms of the couplings $Y_{l}$ and whose edge lengths give the spins \cite{Bonzom:2015ova}. Furthermore, in that case, it is understood that such fixed point couplings are related to the critical couplings of the Ising model on the considered graph \cite{Bonzom:2015ova,Bonzom:2019dpg}.
 
So we have three types of behavior for the probability distribution. Either, the couplings $Y_{l}$ are within the convergence radius of the series defining the coherent spin network and the spin probability distribution is peaked on low spins. Or the couplings $Y_{l}$ make the series diverge in which case the spin probability distribution favors large spins. And finally, in between, there is a critical regime, where the couplings $Y_{l}$ lead to a line of maxima of the spin probability distributions where the spins are the edge lengths, up to arbitrary global rescaling, of a planar 2D cellular decomposition whose angles are determined by the critical couplings. Such a line of stationary points corresponds to a pole for the coherent spin network wave-function \cite{Freidel:2010tt,Bonzom:2012bn,Bonzom:2015ova,Bonzom:2019dpg}. 

 
 \medskip
 
 We would like to conclude this presentation of coherent spin network states by providing  an expression of these states as Gaussian integrals, which allows for explicit and exact calculations and which will be at the heart of our work on the toroidal lattice.
The starting point is the expression \eqref{eq:LS_boundary_state} of the LS spin network wave-functions as group averaging at every node of the graph, which gives:
 \be
 \label{def:psi1}
\psi^{\{y_{l},\xi^{n}_{l}\}}(\{h_l\})
=
\sum_{\{j_l\}} \prod_{l} \f{y_l^{2 j_l}}{(2 j_l)!}\,\prod_{n} (J_n+1)!\,
\int_{SU(2)} \prod_{n} \rd g_{n}
\,
\prod_{l}\, [ \xi_{l}^{t} |g^{-1}_{t(l)} h_l g_{s(l)} | \xi_{l}^{s} \ra^{2 j_l}
\,.
\ee
Now, a lemma shown in \cite{Bonzom:2012bn} allows to switch the integrals over $\SU(2)$ group elements into integrals over spinors in $\C^{2}$ by absorbing the  node factors $(J_n+1)!$.
\begin{lemma}
From \cite{Bonzom:2012bn}, the integral of a homogeneous polynomial $P[g]$ in a group element $g\in\SU(2)$ of even degree $2J$ can be expressed of a Gaussian integral over $\C^{2}$:
\be
\int_{\SU(2)}\rd g\,P[g]
=
\f1{(J+1)!}\int_{\C^{2}}\f{e^{-\la w|w\ra}\rd^{4}w}{\pi^{2}}\,P[g(w)]
\,,
\ee
where $g(w)=|w\ra\la \uparrow| +|w][ \uparrow|$
is a $\SU(2)$ matrix rescaled by a norm factor $\la w|w\ra$.
\end{lemma}
This lemma embeds  $\SU(2)$ as the 3-sphere in the four-dimensional flat Euclidean space $\C^{2}\sim\R^{4}$ and realizes the integral over the 3-sphere as a Gaussian integral over $\R^{4}$. It allows to transform the group integrals entering the definition of the coherent spin network wave-functions into Gaussian integrals over spinor variables attached to the graph nodes:
\be
\label{def:psi2}
\psi^{\{y_{l},\xi^{n}_{l}\}}(\{h_l\})
=
\int_{\C^{2N}}\prod_{n}^{N}\f{e^{-\la w_{n}|w_{n}\ra}\rd^{4}w_{n}}{\pi^{2}}\,\,
e^{\sum_{l} \big{[} \xi_{l}^{t} \big{|}\,g(w_{t(l)})^{\dagger}\, h_l \,g(w_{s(l)})\, \big{|} \xi_{l}^{s} \big{\ra}}
\ee
Although the $\SU(2)$ gauge invariance of the coherent spin network wave-function at every node was obvious in the original definition \eqref{def:psi1} due to the explicit group averaging with group elements at every node, this new expression \eqref{def:psi2} in terms of  a Gaussian integral over spinors is still clearly gauge invariant. Indeed the 2$\times$ 2 matrices, $g(w)$ simply transform under the left $\SU(2)$ action:
\be
\forall k\in\SU(2)\,,\quad
kg(w)=g(kw)\,,
\nn
\ee
so that a $\SU(2)$ gauge transformation $h_{l}\mapsto k_{t(l)}^{-1}h_{l}k_{s(l)}$ for $\{k_{n}\}\in\SU(2)^{\times N}$ can be re-absorbed by a simple redefinition of the spinors $|w_{n}\ra\mapsto k_{n}\,|w_{n}\ra$ since the Gaussian measure $e^{-\la w|w\ra}\rd^{4}w$ is invariant under $\SU(2)$ transformations.

A slight variation of this expression consists in integrating over the inverse group elements at the nodes $g_{n}\mapsto g_{n}^{-1}$ in \eqref{def:psi1}, leading to switching the place of the $\dagger$ in \eqref{def:psi2}:
\beq
\label{def:psi3}
\psi^{\{y_{l},\xi^{n}_{l}\}}(\{h_l\})
&=&
\int_{\C^{2N}}\prod_{n}^{N}\f{e^{-\la w_{n}|w_{n}\ra}\rd^{4}w_{n}}{\pi^{2}}\,\,
e^{\sum_{l} \big{[} \xi_{l}^{t} \big{|}\,g(w_{t(l)})\, h_l \,g(w_{s(l)})^{\dagger}\, \big{|} \xi_{l}^{s} \big{\ra}}
\\
&=&
\int_{\C^{2N}}\prod_{n}^{N}\f{e^{-\la w_{n}|w_{n}\ra}\rd^{4}w_{n}}{\pi^{2}}\,\,
e^{
\sum_{l}
h_{l}^{\uparrow\uparrow}\la w_{s(l)}|X_{l}|w_{t(l)}\ra
+
h_{l}^{\uparrow\downarrow}[ w_{s(l)}|X_{l}|w_{t(l)}\ra
+
h_{l}^{\downarrow\uparrow}\la w_{s(l)}|X_{l}|w_{t(l)}]
+
h_{l}^{\downarrow\downarrow}[ w_{s(l)}|X_{l}|w_{t(l)}]
}
\,,\nn
\eeq
where $X_{l}=|\xi^{s}_{l}\ra [\xi^{t}_{l}|$ and the component $h_{l}^{mm'}$ stands for $\la m|h_{l}|m'\ra$ (i.e.  $h_{l}^{\uparrow\uparrow}=\la \uparrow|h_{l}|\uparrow\ra$ and so on).
Since  the left $\SU(2)$ action on  $g(w)^{\dagger}=|\uparrow \ra\la w|+|\uparrow ][w|$ is not obvious anymore, the $\SU(2)$ gauge invariance is harder to read. Nevertheless, this is resolved by a direct mapping between the matrices $g(w)$ and their complex conjugate $g(w)^{\dagger}$:
\be
g(w)^{\dagger}=g(w^{*})
\qquad\textrm{with}\quad
w^{*}=\mat{c}{\bar{w}^{0}\\-w^{1}}
\,,\quad
(w^{*})^{*}=w
\,,\quad
kg(w)^{\dagger}=kg(w^{*})=g(kw^{*})=g((k^{\dagger}w)^{*})
\,,
\ee
with the Gaussian measure $e^{-\la w|w\ra}\rd^{4}w$ clearly invariant under the change of variable $w\mapsto w^{*}$. This is this final expression \eqref{def:psi3} for the coherent spin network wave-function as the integral of a spinorial action that we will use to compute the Ponzano-Regge amplitude on the solid torus.

The coherent spin network wave-function can then be computed exactly as a complex Gaussian integral and expressed in terms of the corresponding determinant \cite{Freidel:2010tt,Bonzom:2015ova}.
To get the Ponzano-Regge amplitude for such a coherent spin network boundary state on the 2-sphere bounding the 3-ball, one evaluates this wave-function on the identity $h_{l}=\id\,,\,\,\forall l$, while to get the Ponzano-Regge amplitude on the 2-torus bounding a solid cylinder, as we consider in the next section, one will evaluate this wave-function on the holonomy around the non-contractible cycle of the cylinder and then further integrate over that holonomy.

\subsection{Quantum Boundary State on the Torus}
\label{sec:torus_notation}

This paper is dedicated to the study of the Ponzano-Regge amplitude on a manifold with the topology of a 3D solid torus with a 2D torus boundary. We actually consider a 3D cylinder with a disk as its basis. As illustrated in figure \ref{fig:cylinder}, this is meant to describe the evolution of the disk in time, from an initial disk geometry to a final disk geometry.
We choose a 2D square lattice on the boundary torus. Taking the (1-skeleton of the) dual of this boundary cellular complex gives the boundary graph $\Gamma$, which is again a square lattice.
%
%
Sidestepping the question of selecting specific initial and final states on the disk, we glue the final disk back onto the initial disk. This gluing is done with a twist angle $\gamma$---that is, we  rotate the final disk by an angle $\gamma$ before identifying it with the initial disk, obtaining a twisted filled torus with a 2D twisted torus boundary.
This amounts to  computing the trace of the Ponzano-Regge transition amplitude for a disk geometry contracted with a rotation operator (and thus a translation around the 1D boundary circle).

The advantage of using a square lattice, or quadrangulation, is not only that it is naturally adapted to the geometry of the 2D torus and easily extendable to a 3D cellular complex of the solid torus (cut into camembert parts or cake pieces as described in \cite{PRholo1,PRholo2}), but its regularity allows a simple formulation of coarse-graining and refinement towards the continuum limit. 
This will allow a detailed analysis of how the Ponzano-Regge amplitude depends on the details of the boundary graph and quantum state on the boundary.
Moreover, this choice will allow us to perform exact computations of the Ponzano-Regge partition function with for coherent spin network boundary states.
\begin{figure}[h!]
	\begin{center}
	\begin{tikzpicture}[scale=1]
		\coordinate (C1) at (1.5,2);
		\coordinate (D1) at (-1.5,2);
		\coordinate (C2) at (1.5,-2);
		\coordinate (D2) at (-1.5,-2);
		\coordinate (C3) at (1.5,-2.8);
		\coordinate (C0) at (1.5,2.8);
		\coordinate (D0) at (-1.5,2.8);
		
		\draw [thick] (C1) arc[x radius=1.5, y radius=0.5, start angle=0, end angle=-360]node[pos=65/360](E1){};
		\draw [thick] (C2) arc[x radius=1.5, y radius=0.5, start angle=0, end angle=-180];
		\draw [thick,dashed] (C2) arc[x radius=1.5, y radius=0.5, start angle=0, end angle=180];
		\draw (C1) node{$\bullet$} ++(1,0) node{initial disk};
		\draw (C2) node{$\bullet$} ++(1,0) node{final disk};
		\draw [thick] (C1)--(C2);
		\draw [thick] (D1)--(D2);
		\draw [dotted] (C0)--(C1);
		\draw [dotted] (D0)--(D1);

		\draw(C0) arc[x radius=1.5, y radius=0.5, start angle=0, end angle=-360]node[pos=65/360](E0){$\bullet$};		
		\draw(C3) arc[x radius=1.5, y radius=0.5, start angle=0, end angle=-360] node[pos=65/360]{$\bullet$};
		\draw[->,>=latex](1.5,-2.4) arc[x radius=1.5, y radius=0.5, start angle=0, end angle=-70] ;
		\draw (2.6,-2.6) node{twist angle $\gamma$};
		\draw [dotted] (E0)--(E1);
		
	\end{tikzpicture}\end{center}
	\caption{Closing the solid cylinder with a twist: the final disk is glued back onto the initial disk with a twist angle $\gamma$ in order to form a twisted solid torus.}
	\label{fig:cylinder}
\end{figure}
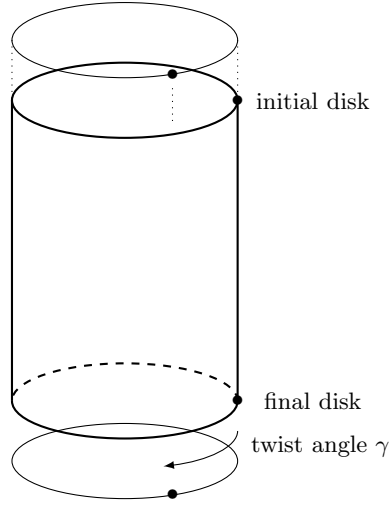

We consider $N_{t}$ vertical slides (the ``time'' direction) and $N_x$ horizontal slices (the ``space'').
Each node of the boundary graph  is denoted by $(t,x)$. This square lattice is closed with periodic conditions, 
\begin{equation}
	\forall (t,x)\in [0,N_t-1]\times[0,N_x-1]\quad (t,N_x) \equiv (t,0) \quad \text{and} \quad (N_t,x) \equiv (0,x + N_\gamma).
	\label{eq:periodicity_condition}
\end{equation}
The torus is trivially glued in the space direction, while the parameter $N_\gamma$ creates a shift in the identification in the time direction, as illustrated in figure \ref{fig:discretized_boundary_torus},  leading to a twist angle
\begin{equation}
	\gamma = 2 \pi \f{N_\gamma}{N_x} \; .
\end{equation}
It could be considered as slightly misleading to identify the vertical and horizontal directions as time and space since we are working in Euclidean signature. It is nevertheless legitimate to consider the vertical direction as the direction of evolution. The horizontal slices describe the spatial geometry (the disk and its circular boundary) and we can study its evolution as we slide along the cylinder. Then the (solid) torus is obtained by identifying the initial and final state of spatial geometry up to a twist.
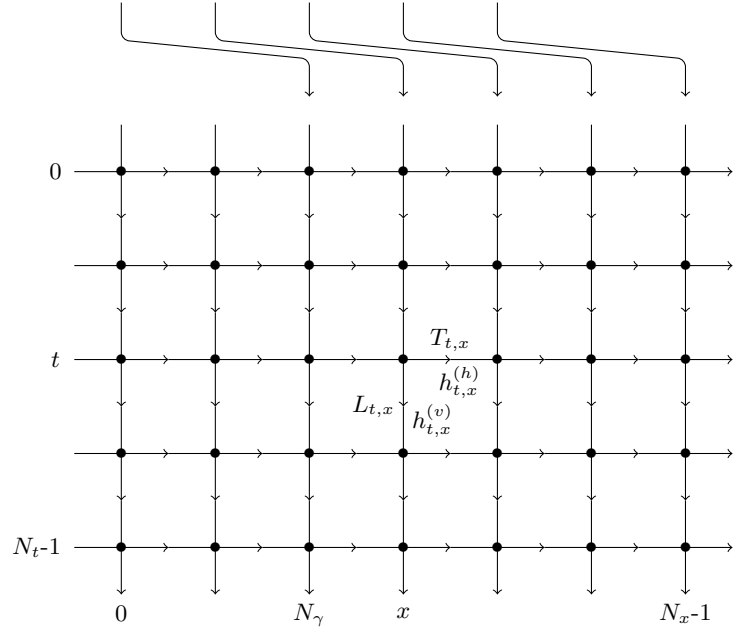
\begin{figure}[h!]
	\begin{center}
	\begin{tikzpicture}[scale=1.25]
	\foreach \i in {0,...,6}{
		\foreach \j in {-1,...,3}{
			\draw (\i,\j) node {$\bullet$};
			\draw[<-] (\i,\j-.5) --(\i,\j+.5);
			\draw[->] (\i-.5,\j) --(\i+.5,\j);
		}
	}
	
	\foreach \i in {0,...,4}{
		\draw[rounded corners=3 pt,->] (\i,4.5+0.3) --(\i,4+0.1+0.3)-- (\i+2,4-0.1+0.3)--(\i+2,3.5+0.3)   ;
	}
	
	\draw (0,-1.7) node{$0$};
	\draw (2,-1.72) node{$N_\gamma$};
	\draw (3,-1.7) node{$x$};
	\draw (6,-1.7) node{$N_x$-1};
	
	\draw(-.9,-1) node{$N_t$-1};
	\draw(-0.7,1) node{$t$};
	\draw(-0.7,3) node{$0$};
	
	\draw(3.6,1) node[below=-1]{$\hh_{t,x}$}; \draw (3.5,1) node[above]{$T_{t,x}$};
	\draw(3,0.35) node[right]{$\hv_{t,x}$}; \draw (3,0.5) node[left]{$L_{t,x}$};
	
	\end{tikzpicture}
\end{center}
	\caption{The boundary graph on the twisted torus defined as a square lattice: the space (``horizontal'') periodic condition is without twist while the time (``vertical'') periodic condition is with a twist angle is $\gamma=2\pi \f{N_{\gamma}}{N_{x}}$. The two links outgoing from a node $(t,x)$ carry respectively the group element  $\hh_{t,x}$ and spin $T_{t,x}$ for the  link in the space direction, and the group element  $\hv_{t,x}$ and spin $L_{t,x}$ for the  link in the time direction.}
	\label{fig:discretized_boundary_torus}
\end{figure}

\begin{figure}[h!]
	\begin{center}
	\begin{tikzpicture}[scale=1.25]
	
	\coordinate (Tr) at (8,1); \coordinate (Tl) at (12,1);
	\coordinate (Lu) at (10,2); \coordinate (Ld) at (10,0);
	\coordinate (C) at (10,1);
	
	\coordinate (A1) at (11,0.5) ; \coordinate (A2) at (11,1.5); \coordinate (A3) at (9,1.5); \coordinate (A4) at (9,0.5);
	
	\draw[gray!=40,thick,->] (A1)--(A2);
	\draw[gray!=40,thick,->] (A2)--(A3);
	\draw[gray!=40,thick,->] (A3)--(A4);
	\draw[gray!=40,thick,->] (A4)--(A1);
	
	\draw[->-=0.3] (Tr) -- node[pos = 0,above]{$|T_{t,x-1},\uparrow ]$} (C);
	\draw[->-=0.7] (C) --  node[above,pos = 1]{$| T_{t,x},\uparrow \ra$} (Tl);
	\draw[->-=0.3] (Lu) -- node[pos = 0,above]{$|L_{t-1,x},+ ]$} (C); 
	\draw[->-=0.7] (C) -- node[below,pos=1]{$|L_{t,x},+ \ra$} (Ld);
	\draw (C) node{$\bullet$};
	\end{tikzpicture}
\end{center}
	\caption{4-valent coherent intertwiner living at the node  $(t,x)$ of the square lattice, determined by the spins and spinors attached to each incoming and outgoing link. It is interpreted geometrically as a rectangle (in grey) dual to the node.}
	\label{fig:rectangleintertwiner}
\end{figure}
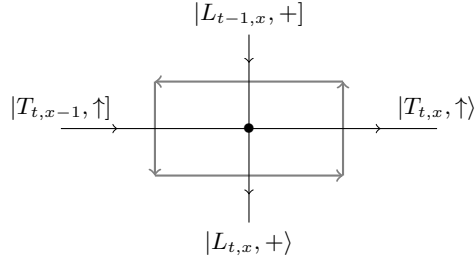

The boundary state $\psi$ is a wave-function of $SU(2)$ group elements living on the links of the boundary graph, respectively $\hh_{t,x}$ and $\hv_{t,x}$ along the  horizontal and vertical links outgoing from the node $(t,x)$. A spin network state is defined by the assignment of a spin $j_{l}$ to each link $l$ and an intertwiner $\cI_{n}$ to each node $n$.
So we assign to a spin $L_{t,x}\in\f\N2$ to each  vertical link between the nodes $(t,x)$ and $(t+1,x)$ and a spin $T_{t,x}$ to each horizontal link between the nodes $(t,x)$ and $(t,x+1)$.
As for  the nodes, since they are 4-valent, the intertwiners are not uniquely determined by the spins and we need to specify them. Following \cite{PRholo2}, we choose coherent intertwiners at every vertex $(t,x)$, representing semi-classical quantized rectangles on the boundary, 
\begin{equation}
	\cI_{t,x} = \int_{\SU(2)} \rd g_{t,x}\,
	|L_{t-1,x},+] \otimes |L_{t,x},+ \ra \otimes | T_{t,x-1}, \uparrow ] \otimes | T_{t,x}, \uparrow \ra
	\,,
	\label{eq:intertwiner_choice}
\end{equation}
with the four spinors around the node given by:
\be
|\uparrow \ra=\mat{c}{1 \\ 0}
\,,\quad
|\uparrow ]=|\downarrow \ra=\mat{c}{0 \\ 1}
\,,\quad
\sqrt{2}\,|+\ra=| \downarrow \ra + |\uparrow \ra =\mat{c}{1 \\ 1}
\,,\quad
\sqrt{2}\,|+]=| \downarrow \ra - |\uparrow \ra=\mat{c}{-1 \\ 1}
\,.
\nn
\ee
The spinor $|\uparrow \ra$ corresponds to a unit vector $\vu_{\uparrow}=\hat{e}_{z}$ upwards in the $z$-direction, while the spinor $|+ \ra$ corresponds to a unit vector $\vu_{+}=\hat{e}_{x}$ pointing in the $x$-direction. Taking the dual of a spinor flips the sign, thus the direction, of the corresponding vector. The coherent intertwiner defined above thus corresponds to a standing up  rectangle, as drawn in figure \ref{fig:rectangleintertwiner}.
With this choice of intertwiners, the LS spin network state, as defined in the previous section by equation \eqref{eq:LS_boundary_state}, reads\footnotemark{}
\begin{equation}
\label{eq:LS_boundary_state_torus}
\psi_{LS}(\{\hh_{t,x},\hv_{t,x}\})
=
\int_{\SU(2)}
\prod_{t,x = 0}^{N_t-1,N_x-1} \rd g_{t,x}\,
\la \uparrow |g^{-1}_{t,x+1} \hh_{t,x} g_{t,x} | \uparrow \ra^{2 T_{t,x}}
\,
\la + |g^{-1}_{t+1,x} \hv_{t,x}g_{t,x} | + \ra^{2 L_{t,x}}
\,.
\end{equation}
\footnotetext{%
Note that in this expression of the LS spin network, the dual spinors have disappeared and  the structure map does not appear explicitly. This is indeed compensated by the orientation of the links. From the definition of coherent intertwiners \eqref{def:cohintertwiner}, an ingoing edge carry an $\SU(2)$ coherent state with a dual spinor. Now, in the definition of the intertwiner \eqref{eq:intertwiner_choice}, all the links where considered outgoing. However, it is clear from the figure \ref{fig:discretized_boundary_torus} that half of the links are in fact  ingoing. The definition of the intertwiner must then be modified to take into account these ingoing links, by applying the structure map to those links and thus taking the dual spinor. Putting everything together, we end up with twice the structure map per link. Now recall that $\varsigma^2 = - \id$. This simply creates sign factors on each link. However, no sign actually appears in the definition of the LS spin network wave-function.
The reason behind this invariance is that for a non-trivial intertwiner to exist at a node $n$, the sum of the spins associated to the legs of the node must be an integer, here $(T_{t,x}+L_{t,x}+L_{t-1,x}+T_{t,x-1})\in\N$. Since switching the orientation of a link $l$ produces a sign  $(-1)^{2j_l}$, we see that the product of all of the signs contributions $(-1)^{2(T_{t,x}+L_{t,x})}$ simplifies to 1.
Interesting this even parity of the sum of the spins around each node leads to a $\Z_{2}$ invariance of the integrand under the local change of variable $g_{t,x} \mapsto -g_{t,x}$.
}

Finally, the definition of a coherent spin network involves the sum over the spins $T_{t,x}$ and $L_{t,x}$ to create a coherent spin superposition with weight defined in equation \eqref{eq:weight_general} as:
\begin{equation*}
\omega\big{[}\{T_{t,x},L_{t,x}\}_{t,x}x\big{]}
=
\prod_{t,x = 0}^{N_t-1,N_x-1} \frac{\lambda_{t,x}^{2 L_{t,x}} \tau_{t,x}^{2T_{t,x}}}{(2 L_{t,x})!(2T_{t,x})!} (J_{t,x}+1)!
\qquad \text{where} \quad J_{t,x} := T_{t,x} + L_{t,x} + T_{t,x-1} + L_{t-1,x} \; .
\end{equation*}
The parameters $\lambda_{t,x}$ and $\tau_{t,x}$ are the couplings controlling the probability distribution for the spins on  the links in respectively time and space directions, corresponding respectively to space and time edges.
The coherent spin network wave-function is then:
\begin{equation}
	\psi(\{\hh_{t,x},\hv_{t,x}\}) = \sum_{T_{t,x}} \sum_{L_{t,x}} \prod_{t,x} \left( \frac{\lambda_{t,x}^{2 L_{t,x}} \tau_{t,x}^{2T_{t,x}}}{(2 L_{t,x})!(2T_{t,x})!} (J_{t,x}+1)!\right) \, \psi_{LS}\left(\left\lbrace\hh_{t,x},\hv_{t,x} \right\rbrace \right) \; .
	\label{eq:boundary_state_scale_free}
\end{equation}
As shown previously, we can trade the sum over the spins for a Gaussian integral over spinors living at the graph nodes:
\begin{equation}
\label{def:psitorus}
\psi(\{\hh_{t,x},\hv_{t,x}\})
=
\int_{\C^2} \Bigg{[}\prod_{t,x}\f{e^{-\la w_{t,x} | w_{t,x} \ra}\rd^4 w_{t,x}}{\pi^2}\Bigg{]}
\,
e^{\tau_{t,x} \la \uparrow |g(w_{t,x+1}) \,\hh_{t,x} g(w_{t,x})^{\dagger} | \uparrow \ra + \lambda_{t,x} \la + |g(w_{t+1,x})\, \hv_{t,x} g_{t,x}(w_{t,x})^{\dagger} | + \ra }
\,,
\end{equation}
with the notation $g(w)=|w\ra\la \uparrow| +|w][ \uparrow|$.
This expression for the quantum boundary state on the 2-torus will be the starting point for our computation of the Ponzano-Regge partition function for 3D quantum gravity on the solid torus, as presented in section \ref{sec:PRontorus}.

\subsection{Critical Regime of the Boundary State and Scale Invariance}

Before moving on to the exact evaluation of the Ponzano-Regge amplitude for such coherent boundary states, let us delve into the geometrical interpretation of those spin network states. The coherent wave-function with spin superposition is written as:
\be
\psi(\{\hh_{t,x},\hv_{t,x}\})
=
\sum_{\{L_{t,x},T_{t,x}\}} f_{\{\lambda_{t,x},\tau_{t,x}\}}[\{L_{t,x},T_{t,x}\}]
\,\int_{\SU(2)}
\prod_{t,x} \rd g_{t,x}\,e^{-\Phi_{\{L_{t,x},T_{t,x}, \hh_{t,x}, \hv_{t,x}\}}[\{g_{t,x}\}]}
\,,
\ee
with the combinatorial pre-factor $f$ and the exponent $\Phi$ given by:
\beq
f_{\{\lambda_{t,x},\tau_{t,x}\}}[\{L_{t,x},T_{t,x}\}]
&=&
\prod_{t,x} \left( \frac{\lambda_{t,x}^{2 L_{t,x}} \tau_{t,x}^{2T_{t,x}}}{(2 L_{t,x})!(2T_{t,x})!} (J_{t,x}+1)!\right)
\,,\\
\Phi_{\{L_{t,x},T_{t,x}, \hh_{t,x}, \hv_{t,x}\}}[\{g_{t,x}\}]
&=&
-2\sum_{t,x}
T_{t,x} \ln\la \uparrow |g^{-1}_{t,x+1} \hh_{t,x} g_{t,x} | \uparrow \ra + L_{t,x} \ln\la + |g^{-1}_{t+1,x} \hv_{t,x} g_{t,x} | + \ra
\,.
\eeq
Following the logic developed  in section \ref{cohstates} for the general case, we would like to understand what are the main contributions to the series defining the wave-function $\psi(\{\hh_{t,x},\hv_{t,x}\})$. In order to analyze the dominant terms of the sum over  the spins $L_{t,x}$ and $T_{t,x}$, we will assume that the dominant contribution is obtained for sufficiently large spins. Thus, we will perform a saddle point approximation for the integral over the $\SU(2)$ group averaging variables $g_{t,x}$ and then look at the sum over the spins. This will clarify the role of the couplings $\lambda_{t,x}$ and $\tau_{t,x}$.


\medskip

The action $\Phi_{\{L_{t,x},T_{t,x}, \hh_{t,x}, \hv_{t,x}\}}[\{g_{t,x}\}]$ is the contribution coming from the group averaging in the definition of the coherent intertwiners used to build the LS spin networks. We review the analysis of the behavior of the integral $\int \rd g_{t,x}\, \exp[-\Phi]$  by a saddle point approximation at large spins previously done in \cite{PRholo2}.

Since $\Phi$ is a priori complex, the main contribution to the integral comes from a stationary point configuration $o$ such that the real part of the action is an absolute minimum and that the first derivatives of the action with respect to the group variables $g_{t,x}$ vanish,
\begin{equation}
\partial_{g_{t,x}} \Phi(\{g_{t,x}\})\Big{|}_o = 0
\,,\qquad
\mathfrak{Re}\,\Phi[\{g_{t,x}\}] \Big{|}_o \leq \mathfrak{Re}\,\Phi[\{g_{t,x}\}] 
\; .
\end{equation}
The first condition is not a condition on the variables $g_{t,x}$ but on the spins $L_{t,x}$ and $T_{t,x}$. It amounts to the closure constraints for the intertwiners at every node $(t,x)$,
\be
T_{t,x}\vu_{\uparrow}-T_{t,x-1}\vu_{\uparrow}
+
L_{t,x}\vu_{+}-L_{t-1,x}\vu_{+}
=0\,.
\ee
Since $\vu_{\uparrow}=\hat{e}_{z}$ and $\vu_{+}=\hat{e}_{x}$ are orthogonal, this implies that the spins are equal on opposite links, $L_{t,x}=L_{t-1,x}$ and $T_{t,x}=T_{t,x-1}$. This means that the spin $T_{t,x}$ is homogeneous around each time slice, $T_{t,x}=T_{t}$ independent from $x$, and the spin $L_{t,x}$ is homogeneous\footnotemark{} along each time line, $L_{t,x}=L_{x}$  independent from $t$.
\footnotetext{%
Note that the time lines can wrap around the twisted torus due to the twisted periodicity condition, which will identify the $L_{x}$'s for different $x$'s. E.g. if the shift $N_{\gamma}$ is 1 or  coprime with $N_{x}$, then the spin $L_{t,x}$ is completely homogeneous independent of the site $(t,x)$.
}
This represents a discretization of a flat cylinder 2D metric  $\rd s^{2}=A(t)^{2}\rd t^{2}+B(x)^{2}\rd x^{2}$ where the scale factors in both direction can be re-absorbed in the definition of the time and space coordinates.

On the other hand, the second condition characterizes the stationary points in the integration variables $\{g_{t,x}\}$.
Recall that since $|\uparrow\ra$ and $|+\ra$ are normalized spinors, the modulus of the matrix elements $|\la \uparrow | g | \uparrow \ra|$ and $ |\la + | g | + \ra|$ for an arbitrary group element $g\in\SU(2)$ are always positive and  bounded by 1.
%
%
Since  $\mathfrak{Re}\,\log(z) = \log(|z|)$ for any complex number $z\in\C$,  the absolute minimal of $\mathfrak{Re}\,\Phi$ is $0$ and  is reached when the arguments of the logarithms are 1 in modulus. So the stationary points are given by
\be
\la + |g^{-1}_{t+1,x} \hv_{t,x} g_{t,x} | + \ra = e^{i \phi^L_{t,x}}
\,,\qquad
\la \uparrow |g^{-1}_{t,x+1} \hh_{t,x} g_{t,x} | \uparrow \ra = e^{i \phi^T_{t,x}} \; .
\label{eq54}
\ee
for some  phases $\phi^{L,T}_{t,x}\in[0,2\pi]$.
%
On shell of these equations, and assuming the homogeneity of the spins resulting from the closure constraints, the action $\Phi$ evaluates to:
\begin{equation}
\Phi_{\{L_{t,x},T_{t,x}, \hh_{t,x}, \hv_{t,x}\}}[\{g_{t,x}\}]\bigg{|}_o
=
-2i \sum_{t,x} L_{x}  \phi^L_{t,x} + T_{t} \phi^T_{t,x}
\,,
\qquad\textrm{with}\quad
T_{t,x}=T_{t}\quad\textrm{and}\quad L_{t,x}=L_{x} \quad\forall t,x
\, .
\end{equation}
As shown in \cite{PRholo2}, since the spins give the edge lengths of the 2D square lattice, the phases $\phi^{L,T}_{t,x}$ correspond to (half) the dihedral angles between the rectangular faces, as shown in figure \ref{fig:dual_edge_dihedral_angle}. Then, the on-shell action $\Phi$ reproduces the  evaluation of the  Gibbons-Hawking-York boundary term for the 3D Regge action on the solid torus (see \cite{PRholo2} for details, proofs and discussions).
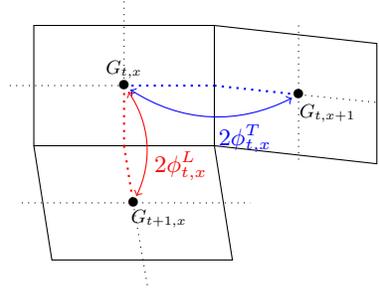
\begin{figure}[!htb]
	\begin{center}
		\begin{tikzpicture}[scale=0.80,line/.style={<->,shorten >=0.1cm,shorten <=0.1cm}]
		\coordinate(a1) at (0,0); \coordinate(b1) at (0.3,-1.9);
		\coordinate(a2) at (3,0); \coordinate(b2) at (3.3,-1.9); \coordinate(d1) at (5.7,-0.3);
		\coordinate(h1) at (0,2); 
		\coordinate(h2) at (3,2); \coordinate(d2) at (5.7,1.7);

		\draw (a1)--(a2)--(h2)--(h1)--cycle;
		\draw (a1)--(b1)--(b2)--(a2);
		\draw (a2)--(d1)--(d2)--(h2);
		
		\draw[dotted] (1.5,2.4)--(1.5,1);
		
		\draw[dotted] (-0.4,1)--(1.5,1); 
		
		\draw[dotted,blue,thick] (1.5,1)--(3,1); \draw[dotted, blue, thick] (3,1) -- (4.4,0.85);
		\draw[dotted,thick,red] (1.5,1)--(1.5,0); \draw[dotted,red,thick] (1.5,0)--(1.65,-0.95);

		\draw[dotted] (1.65,-0.95)--(1.9,-2.4);
		\draw[dotted] (-0.2,-0.95)--(3.6,-0.95);
		
		\draw[dotted] (4.4,0.85)--(5.9,0.7); \draw[dotted] (4.4,2)--(4.4,-0.3);
		
		\draw (1.5,1) node[above, scale=0.8]{$G_{t,x}$}; \draw (1.5,1) node{$\bullet$};
		\draw (1.5,-0.95) node[below right,scale=0.8]{$G_{t+1,x}$};\draw (1.65,-0.95) node{$\bullet$};
		\draw (4.3,0.8) node[below right,scale=0.8]{$G_{t,x+1}$};\draw (4.4,0.85) node{$\bullet$};
		
		\path [red,line,bend left] (1.5,1) edge node[midway,below right]{$2\phi_{t,x}^{L}$} (1.65,-0.95);
		\path [blue,line,bend right] (1.5,1) edge node[midway,below right]{$2\phi_{t,x}^{T}$} (4.4,0.85);
		\end{tikzpicture}
	\end{center}
	\caption{Three plaquettes - faces - of the boundary discretization dual to 3 vertices. The dotted lines are the links of the dual lattice. In red, the dihedral angle $2 \phi_{t,x}^{L}$ between the plaquette dual to the vertices $(t,x)$ and $(t+1,x)$ and in blue the dihedral angle $2 \phi_{t,x}^{T}$ between the plaquette dual to the vertices $(t,x)$ and $(t,x+1)$.}
	\label{fig:dual_edge_dihedral_angle}
\end{figure}


\medskip

Now it's time to deal with the sum over the spins $T_{t,x}$ and $L_{t,x}$. Let us look at the stationary configurations in spins to identify the dominant contributions to the series. 

The previous work \cite{Bonzom:2015ova} investigated the case of a triangulation and focused on the combinatorial pre-factor $f_{\{\lambda_{t,x},\tau_{t,x}\}}[\{L_{t,x},T_{t,x}\}]$ ignoring the coherent intertwiner contributions and thus the role of the extrinsic curvature. It nevertheless showed that stationary spins exist if and only if the couplings determine the angles of a planar triangulation. In order to be so, the couplings have to satisfy certain constraints. Then, in that critical regime, we do not get isolated stationary points but obtain instead a whole line of stationary points with the spins given by the edge lengths of the triangulation up to an arbitrary global scale factor. We obtain a similar result here, but extended to take into account the coherent intertwiner contributions with the dihedral angles and extrinsic curvature.

%

We combine the two terms in the coherent wave-function $\psi$ coming from the combinatorial pre-factor $f$ and the LS action $\Phi$. Proceeding to a large spin approximation for the pre-factor, we use the Stirling formula $n! \sim \sqrt{2 \pi n } \frac{e^n}{n^n}$ to get:
\beq
&&f_{\{\lambda_{t,x},\tau_{t,x}\}}
\big{[}\{L_{t,x},T_{t,x}\}\big{]}
\,=\,
e^{F}
\\
&\textrm{with}\quad&
F\sim \sum_{t,x}
2 L_{t,x} \ln\lambda_{t,x} + 2 T_{t,x} \ln\tau_{t,x}
-J_{t,x}+2L_{t,x}+2T_{t,x}
+ J_{t,x} \ln J_{t,x} - 2L_{t,x} \ln(2L_{t,x}) - 2T_{t,x} \ln(2T_{t,x})
\,.
\nn
\eeq
Using this large spin approximation, the stationary point equations for the spins read\footnotemark{}
\begin{align*}
\partial_{L_{t,x}} (F-\Phi) &= 0
\quad\Rightarrow\quad
2 \ln(\lambda_{t,x}) = \ln\left(\f{4L_{t,x}^2}{J_{t,x}J_{t+1,x}}\right) 
- 2\ln\la + |g^{-1}_{t+1,x} \hv_{t,x} g_{t,x} | + \ra
\; ,\\ 
\partial_{T_{t,x}}(F-\Phi) &= 0
\quad\Rightarrow\quad
2 \ln(\tau_{t,x}) = \ln\left(\f{4T_{t,x}^2}{J_{t,x}J_{t,x+1}}\right) 
- 2\ln\la \uparrow |g^{-1}_{t,x+1} \hh_{t,x} g_{t,x} | \uparrow \ra\; .
\end{align*}
\footnotetext{%
We focus on real geometries, thus real spins. The analytic continuation of the series to complex spins and the resulting contributions from complex saddle points was tackled in \cite{Bonzom:2019dpg}. It allows to strengthen the relation between the Ponzano-Regge boundary theory and the 2D Ising model but leads to a more complicated geometrical interpretation.
}
Plugging the solution of the saddle point equations  \eqref{eq54} for the group variables $g_{t,x}$, these stationary point equations for the spins translate into  equations for the couplings $\lambda_{t,x}$ and $\tau_{t,x}$:
\begin{equation}
\label{eq:saddle}
\lambda_{t,x} = \f{2L_{x}}{\sqrt{J_{t,x}J_{t+1,x}}} e^{-i \phi_{t,x}^L} \; ,
\qquad
\tau_{t,x} = \f{2T_{t}}{\sqrt{J_{t,x}J_{t,x+1}}}  e^{-i\phi_{t,x}^T} \;,
\qquad\textrm{with}\quad
J_{t,x}=2L_{x}+2T_{t}
\,.
\end{equation}
Note that we have also used the fact that the spins $T_{t,x}$ determining the time intervals are constant in space $T_{t,x}=T_{t}$ and the spins $L_{t,x}$ determining the space intervals are constant in time $L_{t,x}=L_{x}$, as  resulting from the  closure constraints induced by the saddle point equations in $g_{t,x}$.

We must not forget that the couplings $\lambda_{t,x}$ and $\tau_{t,x}$ are given a priori and that the stationary point equations determine extremal spin configurations $\{L_{x},T_{t}\}$ in terms of those couplings.
A first important point is that such stationary point do not always exist. Indeed, the existence of a solution to the stationary point equations requires the couplings to satisfy some constraints. We will refer to couplings that fulfill those conditions as {\it critical couplings}.
The second important point is that, once a set of critical couplings has been chosen and one solution for the spin configuration $\{L_{x},T_{t}\}$ has been identified, then any arbitrary global rescaling of  these spins still gives a solution to the stationary point equations. This means that we actually a whole line of stationary points. This induced scale invariance for critical couplings then leads to a divergence of the series defining the coherent spin network wave-function. For non-critical couplings, there is no finite real stationary points, we lose the scale-invariant line of stationary points in the spins and the dominant contributions are given  either by low spins or  or by spins growing to infinity, respectively corresponding to an absolutely convergent or totally divergent series.

So let us start by highlighting necessary constraints satisfied by the critical couplings. Since the phases of the couplings give the dihedral angles determined by the stationary points in $g_{t,x}$, the stationary point equations in the spins will give conditions on the modulus of the couplings.
%
%
Let us focus on the stationary point equations involving $T_{t}, T_{t+1}, L_{x}, L_{x+1}$:
\be
|\lambda_{t,x}| = \f{2L_{x}}{\sqrt{J_{t,x}J_{t+1,x}}}
\,,\quad
|\tau_{t,x}| = \f{2T_{t}}{\sqrt{J_{t,x}J_{t,x+1}}} 
\,,\quad
|\lambda_{t,x+1}| = \f{2L_{x+1}}{\sqrt{J_{t,x}J_{t+1,x+1}}}
\,,\quad
|\tau_{t+1,x}| = \f{2T_{t+1}}{\sqrt{J_{t+1,x}J_{t+1,x+1}}}
\,.
\ee
We can re-write these equations directly in terms of the ratios of $L$'s over $T$'s:
\be
|\tau_{t,x}| =
\f1{\sqrt{(1+A)(1+C)}}
\,,\qquad
|\tau_{t+1,x}| =
\f1{\sqrt{(1+B)(1+D)}}
\,,
\ee
\be
|\lambda_{t,x}| =
\f{\sqrt{AB}}{\sqrt{(1+A)(1+B)}}
\,,\qquad
|\lambda_{t,x+1}| =
\f{\sqrt{CD}}{\sqrt{(1+C)(1+D)}}
\,,\nn
\ee
with the notation:
\be
A:=\f{L_{x}}{T_{t}}
\,,\quad
B:=\f{L_{x}}{T_{t+1}}
\,,\quad
C:=\f{L_{x+1}}{T_{t}}
\,,\quad
D:=\f{L_{x+1}}{T_{t+1}}
\,,\quad
\textrm{satisfying the relation}\,\,
AD=BC
\,.
\ee
These expressions directly imply that the couplings satisfy a polynomial equation:
\be
\label{eq:critical}
1+|\lambda_{t,x}|^2|\lambda_{t,x+1}|^2+|\tau_{t,x}|^2|\tau_{t+1,x}|^2
=
2|\lambda_{t,x}||\lambda_{t,x+1}||\tau_{t,x}||\tau_{t+1,x}|
+|\lambda_{t,x}|^2+|\lambda_{t,x+1}|^2+|\tau_{t,x}|^2+|\tau_{t+1,x}|^2
\,.
\ee
This looks  similar to the equation for the zeroes of the 2D Ising model. In light of the recently uncovered relation between the geometrical stationary point of coherent spin networks on triangulations and the critical couplings of the 2D Ising model on dual graphs \cite{Bonzom:2015ova,Bonzom:2019dpg}, it would probably be enlightening to investigate this similarity further and understand if there is indeed a straightforward mapping between critical couplings of coherent spin networks and 2D Ising model on the square lattice.

Then, once the couplings satisfy such a condition, it is possible to compute the spin ratios A,B,C,D in terms of the couplings, as shown in appendix \ref{app:criticalcouplings}. Extending this from node to node and enforcing the periodicity conditions of the twisted torus, we finally get a spin configuration solution on the whole square lattice. 

\medskip

In this paper, we will not pursue the general case of arbitrary inhomogeneous couplings but will instead assume further homogeneity of the boundary state and focus on the homogeneous, though anisotropic, case where the couplings do not depend on the graph node\footnotemark:
\begin{equation*}
\tau_{t,x} = \tau  \qquad \text{and} \qquad \lambda_{t,x} = \lambda \qquad \forall (t,x) \; .
\end{equation*}
\footnotetext{%
Keeping arbitrary couplings $\tau_{t,x}$ and $\lambda_{t,x}$ depending on the graph node $(t,x)$ would lead in the continuum limit to fields $\tau(t,x)$ and $\lambda(t,x)$ defined on the 2D boundary. These fields would play the role of sources of the 2D geometry and thus of the 3D gravitational field. In the perspective of holography, it would be interesting to investigate what is the induced boundary action for those source field, whether the scale invariance translates into a conformal invariance on the boundary and if their correlations can be given a clear geometrical interpretation.
}
In this setting, the stationary point equations read:
\be
\label{eq:homosaddle}
\lambda= \f{2L_{x}}{\sqrt{J_{t,x}J_{t+1,x}}} e^{-i \phi_{t,x}^L} \; ,
\qquad
\tau= \f{2T_{t}}{\sqrt{J_{t,x}J_{t,x+1}}}  e^{-i\phi_{t,x}^T} \;,
\qquad\textrm{with}\quad
J_{t,x}=2L_{x}+2T_{t}
\,,
\ee
It is clear that this implies the homogeneity of the dihedral angles, $\phi_{t,x}^L=\phi^L$ and $\phi_{t,x}^T=\phi^T$, i.e. of the extrinsic geometry. It is also straightforward\footnotemark{} to show that it implies the homogeneity of the spins, $L_{x}=L$ and $T_{t}=T$.
Since the spins give the edge lengths in Planck unit, this translates geometrically into the homogeneity of the 2D intrinsic boundary geometry. Moreover this is possible if and only if the couplings satisfy $|\lambda|+|\tau|=1$.
One can check that this equation indeed corresponds to the condition \eqref{eq:critical} derived earlier in the inhomogeneous case applied to the homogeneous coupling ansatz.
So, the equation for criticality of the coherent spin networks for homogeneous couplings is simply that:
\be
|\lambda^{(c)}|+|\tau^{(c)}|=1
\,,
\ee
where the subscript ${}^{(c)}$ refers to ``critical''. When the couplings satisfy this condition, there exists phases $\phi^L$ and $\phi^T$ determining a solution to the stationary point equations, thus leading to a line of stationary points and a pole of the series defining the coherent spin network wave-function. This condition thus serves as radius of convergence of the series: whenever $|\lambda|+|\tau|<1$, the series is absolutely convergent.
%
\footnotetext{%
Starting for the equation of the couplings modulus,
\be
|\lambda|= \f{2L_{x}}{\sqrt{J_{t,x}J_{t+1,x}}}
\,,\quad
|\tau|= \f{2T_{t}}{\sqrt{J_{t,x}J_{t,x+1}}}
\,,\qquad
\textrm{and thus}
\quad
|\lambda|\sqrt{J_{t+1,x}}+|\tau|\sqrt{J_{t,x+1}}
=
\sqrt{J_{t,x}}\,,\,\,\forall t,x
\,,
\nn
\ee
we first sum the latter relation over all the nodes $(t,x)$ of the lattice, obtaining the necessary condition $|\lambda|+|\tau|=1$ for the existence of stationary spin configurations. Then using that $|\tau|=1-|\lambda|$, and keeping in mind that necessarily $0\le |\lambda|\le 1$,  we write that same relation as
\be
|\lambda|\big{[}\sqrt{J_{t+1,x}}-\sqrt{J_{t,x+1}}\big{]}=\sqrt{J_{t,x}}-\sqrt{J_{t,x+1}}
\,,\quad\textrm{which implies}\,\,\,
\sqrt{J_{t+1,x}}-\sqrt{J_{t,x+1}}\ge\sqrt{J_{t,x}}-\sqrt{J_{t,x+1}}
\,,\quad\textrm{i.e.}\,\,
\sqrt{J_{t+1,x}}\ge\sqrt{J_{t,x}}
\,.
\nn
\ee
Since the torus is periodic in the time direction, if we move up in time from the node $(t,x)$, we will eventually come back to the same node (even with a non-trivial twist), thus the sequence of inequalities implies that $J_{t,x}$ does not depend on the time coordinate $t$. Since $J_{t,x}=T_{t}+L_{x}$, this means that $T_{t}$ does not depend on $t$. Finally the equation for $\lambda$ allows us to conclude that $L_{x}$ is also independent from the space coordinate $x$.
}
%

\section{3D Quantum Gravity on the Solid Torus}
\label{sec:PRontorus}

In this section, we will focus on the computation of the Ponzano-Regge partition function for the solid torus with the quantum boundary state on the 2-torus given by a coherent spin network on the boundary square lattice. This Ponzano-Regge amplitude is given by the integral over the $\SU(2)$ holonomy along the non-contractible cycle of the evaluation of the boundary wave-function. We will show that it will be exactly computable and leads to a discretized, generalized and regularized  version of the BMS character formula for the partition function of 3D gravity as derived in \cite{Barnich:2015mui,Oblak:2015sea,Bonzom:2015ans}. In particular, we will discuss the continuum and asymptotic limits, in which the usual BMS character is recovered.



\subsection{Ponzano-Regge Amplitude on the Torus}

We apply the proposition \ref{prop:torus} derived in the first section about the Ponzano-Regge amplitude for a solid torus to the special case of a boundary graph defined by the square lattice on the 2-torus. This means evaluating the boundary wave-function setting all the group elements to the identity except on the final time slice and then integrate over this last group element interpreted as the holonomy along the non-contractible cycle of the solid torus:
\be
Z^{\PR}_{\cT}[\psi]=
\int_{\SU(2)} \rd H\,
\psi\Big{(}
\{\hh_{t,x} = \id\}_{t,x}, \{\hv_{t,x}= \id\}_{t < N_t-1,x}, \{\hv_{N_{t}-1,x}= H\}_{x}
\Big{)}
\,.
\ee
This formula was obtained through a $\SU(2)$ gauge fixing along a maximal tree on the boundary graph. If we do not proceed to a maximal gauge fixing, it is  interesting to take a step back and realize that we can redistribute the holonomy along the non-contractible cycle in an arbitrary way among all the temporal slices of the torus:
\be
Z^{\PR}_{\cT}[\psi]=
\int_{\SU(2)} \prod_{t}\rd H_{t}\,
\psi\Big{(}
\{\hh_{t,x} = \id\}_{t,x}, \{\hv_{t,x}= H_{t}\}_{t,x}
\Big{)}
\,.
\ee
Of course, the integrand as a function of the group elements $H_{t}$ is invariant under $\SU(2)$ transformations acting in between time slices:
\be
\phi(\{H_{t}\})=\psi\Big{(}
\{\hh_{t,x} = \id\}_{t,x}, \{\hv_{t,x}= H_{t}\}_{t,x}
\Big{)}\,,\qquad
\phi(\{H_{t}\})=\phi(\{g_{t}^{-1}H_{t}g_{t+1}\})
\,,\,\,
\forall g_{t}\in\SU(2)\,.
\ee
Gauge fixing this action, we can set all the $H_{t}$ to the identity except on the final slice, recovering the expression we started from.

In the previous work \cite{PRholo1,Riello:2018anu}, in the context of mapping the Ponzano-Regge amplitude onto a boundary spin chain model, it was interesting to restrict the group averaging to the final slice, in which case it led to a projector on the overall spin-0 sector of the chain. Here  it is actually more useful to distribute the holonomy evenly along the lattice, in order to compute the amplitude by a clean and simple Fourier transform. To this purpose, we use the following lemma:
\begin{lemma}
Let us consider a function of $N$ group elements in $\SU(2)^{\times N}$ such that it is invariant under a cyclic action of $\SU(2)$:
\be
\phi(H_{1},H_{2},..,H_{N})=\phi(g_{1}^{-1}H_{1}g_{2},g_{1}^{-1}H_{1}g_{2},..,g_{N}^{-1}H_{N}g_{1})
\,,\quad
\forall \{g_{n}\}\in\SU(2)^{\times N}
\,.
\ee
Then its integral over  $\SU(2)^{\times N}$ reduces to a integral over a single angle of its homogeneous evaluation:
\be
\int_{\SU(2)^{\times N}}\prod_{l}\rd H_{l}\,
\phi(\{H_{l}\})
=
\f1\pi\int_{0}^{2\pi}\rd\vphi\,\sin^{2}(\vphi)\,
\phi(\{H_{l}=e^{i\f{\vphi}{N}\sigma_{z}}\})
\,.
\ee
\end{lemma}
\begin{proof}
Due to the gauge invariance of the function under $\SU(2)$, it amounts to a function of the single group element $H=H_{1}H_{2}..H_{N}$:
\be
\phi(H_{1},H_{2},..,H_{N})=\phi(H,\id,..,\id):=f(H)\,.
\ee
This means that only the product $H=H_{1}H_{2}..H_{N}$ matters and we can redistribute the overall holonomy $H$ in an arbitrary way among the individual $H_{l}$'s. Moreover, the final function $f(H)$ is central, i.e. it has a remaining invariance under the $\SU(2)$ action by conjugation, $f(H)=f(g^{-1}Hg)$ for all $g\in\SU(2)$. We can gauge fix this action by fixing the direction of $H$ while integrating over its class angle:
\be
\int_{\SU(2)^{\times N}}\prod_{l}\rd H_{l}\,
\phi(\{H_{l}\})
=
\int \rd H\, f(H)
=
\f1\pi\int_{0}^{2\pi}\rd\vphi\,\sin^{2}(\vphi)\,
f(H=e^{i{\vphi}\sigma_{z}})
\,,
\ee
where the $\sin^{2}\vphi$ factor is the Cartan measure. Finally, we distribute the overall holonomy $H=e^{i{\vphi}\sigma_{z}}$ into $N$ equal pieces, giving the wanted formula.

\end{proof}
Applying this lemma to the Ponzano-Regge amplitude for the boundary square lattice gives:
\be
Z^{\PR}_{\cT}[\psi]=
\f1\pi\int_{0}^{2\pi}\rd\vphi\,\sin^{2}(\vphi)\,
\psi\Big{(}
\{\hh_{t,x} = \id\}_{t,x}, \{\hv_{t,x}= e^{i\f{\vphi}{N}\sigma_{z}}\}_{t,x}
\Big{)}
\,.
\ee
Applying this to the coherent spin network wave-function on the square lattice, as defined in \eqref{def:psitorus} finally gives the Ponzano-Regge amplitude for the solid torus as a Gaussian integral over spinor variables and an integral over the class angle of the non-trivial holonomy:
\begin{equation}
Z^{\PR}_{\cT}(\lambda,\tau)
=
\f{1}{\pi} \int_{0}^{2\pi} \text{d}\varphi \sin^{2}(\varphi) \int_{\C^2}\prod_{t,x=0}^{N_t-1,N_x-1} \frac{\text{d}^4 w_{t,x}}{\pi^2} e^{- S_{\lambda,\tau}[\{w_{t,x}\},\varphi]} \; .
\label{eq:PR_amplitude_starting}
\end{equation}
with  an action $S_{\lambda,\tau}$ depending on the couplings $\lambda$ and $\tau$ controlling the spin probability distribution:
\beq
S_{\lambda,\tau}[\{w_{t,x}\},\varphi]
&
\,=\,
\sum_{t,x} \la \omega_{t,x} | \omega_{t,x} \ra
&
- \tau \la \uparrow | \omega_{t,x+1} \ra \la \omega_{t,x} | \uparrow \ra -  \tau \la \uparrow | \omega_{t,x+1} ] [ \omega_{t,x} | \uparrow \ra
\nn\\
&&- \lambda e^{i \frac{\varphi}{N_t}} \la + | \omega_{t+1,x} \ra \la \omega_{t,x} | + \ra - \lambda e^{-i \frac{\varphi}{N_t}} \la + | \omega_{t+1,x} ] [ \omega_{t,x} | + \ra
\,.
\eeq
As we explained above, we have decided to uniformly redistribute the $\SU(2)$ holonomy $H$ among the time slices. As we have shown above, this is a specific gauge fixing. We could use any non-uniform splitting of the holonomy, monotonous or not, when computing the Ponzano-Regge amplitude $Z^{\PR}_{\cT}(\lambda,\tau)$ and it would not change the result. The advantage of the uniform splitting is that we can diagonalize the Gaussian integrand by a straightforward Fourier transform on the spinor variables. Using a non-uniform splitting leads to a more complicated decomposition of the action $S_{\lambda,\tau}$ in terms of Fourier modes, although it could be possible to adjust the definition of the Fourier transform to compensate the non-uniformity.

Thus performing a Fourier transform allows to perform the Gaussian integral over the spinor and write the Ponzano-Regge toroidal partition function as a trigonometric integral over the remaining class angle of the non-trivial holonomy:
\begin{prop}
The Ponzano-Regge partition function for the solid twisted torus with the coherent spin network boundary state defined on the square lattice with couplings $\lambda$ and $\tau$ can be written as an integral over the class angle of the holonomy wrapping around the solid torus' non-contractible cycle:
\be
\label{eq:starint_point_amp}
Z^{\PR}_{\cT}(\lambda,\tau)
\,=\,
\f{1}{\pi} \int_{0}^{2\pi} \text{d}\varphi \sin^{2}(\varphi) \prod_{\omega,k = 0}^{N_t-1,N_x-1} \f{1}{\det[Q_{\omega,k}(\varphi)]} 
\ee
with the determinants
\beq
\det[Q_{\omega,k}(\varphi)]
&=&
1 + \lambda^2 + \tau^2  -2\tau \cos\left(\f{2\pi}{N_x}k\right) - 2\lambda \cos\left(\f{\varphi}{N_t} + \f{2\pi}{N_t}\omega - \f{1}{N_t}\gamma k\right) \\
&&+ 2 \lambda \tau \cos\left(\f{2\pi}{N_x}k\right) \cos\left(\f{\varphi}{N_t} + \f{2\pi}{N_t}\omega - \f{1}{N_t}\gamma k\right)
\,.\nn
\eeq
\end{prop}

\begin{proof}

The first step  is to introduce a twisted Fourier transform on the torus to account for the peculiar periodic boundary conditions on the lattice. We define the discrete Fourier transform $\tilde{w}_{\omega,k}$ of the spinor $w_{t,x}$ by:
\be
\tilde{w}_{\omega,k}
\,=\,
\f{1}{\sqrt{N_t N_x}} \sum_{t,x=0}^{N_t-1,N_x-1} 
e^{i \f{2\pi}{N_x}kx}\,
e^{i \f{2\pi}{N_t}t\left(\omega - \f{N_{\gamma}}{N_{x}} k \right)}\,
w_{t,x}
\,,
\ee
\be
w_{t,x}
\,=\,
\f{1}{\sqrt{N_t N_x}} \sum_{\omega,k}^{N_t-1,N_x-1} 
e^{-i \f{2\pi}{N_x}kx}\,
e^{-i \f{2\pi}{N_t}t\left(\omega - \f{N_{\gamma}}{N_{x}} k \right)}\,
\tilde{w}_{\omega,k}
\,,
\ee
where we have absorbed the twist angle $\gamma = {2 \pi N_\gamma}/{N_x} $ as a shift in frequency in the temporal modes. This slight modification of the Fourier transform allows to translate the twisted periodicity conditions $w_{t+N_{t},x}=w_{t,x+N_{\gamma}}$ into standard periodicity conditions for the Fourier transformed spinor:
\be
\tilde{w}_{\om+N_t,k}
= \tilde{w}_{\om,k}=
\tilde{w}_{\omega,k+N_x}
\,,
\ee
meaning that it maps the twisted torus in coordinate space to the standard torus in momentum space.
The action is diagonalized by this transform:
\be
\label{eq:arg_exp_action}
S_{\lambda,\tau}[\{\tilde{w}_{\omega,k} \},\varphi]
\,=\,
\sum_{\omega,k}
\la \tilde{w}_{\omega,k} |Q_{\omega,k}(\varphi)| \tilde{w}_{\omega,k}\ra
\ee
with the 2$\times$2 matrices:
\be
Q_{\omega,k}(\varphi)
\,=\,
\id_{2}
- \tau
\mat{cc}{e^{i \f{2\pi k}{N_x}}& 0 \\ 0 & e^{-i \f{2\pi k}{N_x}}}
-\lambda
\mat{cc}{\cos\theta_{\om,k,\vphi} & i\sin\theta_{\om,k,\vphi}\\i\sin\theta_{\om,k,\vphi}&\cos\theta_{\om,k,\vphi}}
\qquad\textrm{with}\quad
\theta_{\om,k,\vphi}:=\f1{N_{t}}(\vphi+2\pi \om-\gamma k)
\,,\nn
\ee
which can be written in a more compact fashion in terms of the Pauli matrices as
\be
Q_{\omega,k}(\varphi)
\,=\,
\id_{2}
- \tau e^{i \f{2\pi k}{N_x} \sigma^z}
- \lambda e^{\f i{N_{t}} (\vphi+2\pi \om-\gamma k)\, \sigma^x}
\,.
\ee
For the sake of keeping notation simple, we keep implicit that those Hessian matrices $Q_{\omega,k}$ depend on the coupling constants $\lambda$ and $\tau$.
The Gaussian integration over the spinors $w_{t,x}$ in the expression \eqref{eq:PR_amplitude_starting} for the Ponzano-Regge amplitude gives the product of the determinants of the $Q_{\om,k}$ matrices, which are straightforward to compute:
\be
\det[Q_{\omega,k}(\varphi)]
=
1 + \lambda^2 + \tau^2  -2\tau \cos\left(\f{2\pi}{N_x}k\right) - 2\lambda \cos\theta_{\om,k,\vphi}
+ 2 \lambda \tau \cos\left(\f{2\pi}{N_x}k\right) \cos\theta_{\om,k,\vphi}
\,.
\ee 

\end{proof}
Interestingly the inverse determinant of the 2$\times$2 matrix $Q:=\id_{2}- \tau e^{i \alpha\sigma^z}- \lambda e^{i\beta \sigma^x}$, considered as a function of $\tau$ and $\lambda$, is the generating function of a bi-variate Chebyshev polynomials. As we explain in the appendix \ref{app:chebychev},  this point of view becomes especially convenient when unfreezing the spinors defining the coherent intertwiners, for instance when tilting the square lattice made of semi-classical rectangles into semi-classical parallelograms.

\subsection{Mode determinants, Mass and Criticality}

Now  working from the expression of the Ponzano-Regge amplitude as a trigonometric integral, $Z^{\PR}_{\cT}(\lambda,\tau)=\int\text{d}\varphi \sin^{2}(\varphi) \prod_{\omega,k } \det[Q_{\omega,k}(\varphi)]^{{-1}}$, the next step is to identify the poles in $\vphi$ and calculate their residue  in order to compute this integral as a contour integral in the complex plane. The class angle $\vphi$ encodes the holonomy along the non-contractible cycle of the solid torus and describes the curvature of the geometry. Poles in $\vphi$, corresponding to zeros of the determinants, $\det[Q_{\omega,k}(\varphi)]=0$, can be loosely thought as resonances in the connection.

From the perspective of standard (free) quantum field theory, zeroes of the determinants $\det[Q_{\omega,k}(\varphi)]=0$ would correspond to the classical dispersion relation between the frequency $\om$ and the momentum $k$. Here, not only this relation depends on the curvature angle $\vphi$, but we further have to integrate over that angle. Before computing that integral, it is nevertheless enlightening to investigate the physical meaning of those $\vphi$-dependant dispersion relations.

Let us start with the case of a trivial connection in the solid torus, $\vphi=0$. It turns out that we can factorize the determinants in terms of the eigenvalues of the discrete Laplacians in the space and time directions (taking into account the twist in the periodic boundary conditions in time),
\be
\Delta_{k} = 2-2\cos\left(\f{2\pi}{N_x}k\right)
\,,\qquad
\Delta_{\omega} = 2-2\cos\left(\f{2\pi}{N_t}\omega - \f{1}{N_t}\gamma k\right)
\,,
\ee
leading to a simple decomposition of the determinant in terms of Laplacian  and mass term:
\be
\det[Q_{\omega,k}(0)]
=
M_{0}^2 + \tau(1-\lambda)\Delta_{k}+ \lambda(1-\tau)\Delta_{\om}+\f12\lambda\tau\Delta_{k}\Delta_{\omega}\,,
\ee
with the mass given in terms of the couplings:
\be
M_{0}^{2}
=
\det[Q_{\omega,k}(0)]\Big|_{\om=k=0}
=
\big{(}1-\tau-\lambda\big{)}^{2}
\,.
\ee
Let us remember that the critical coupling equation was simply $|\tau|+|\lambda|=1$. This means that the case of real positive critical couplings\footnotemark, $\lambda+\tau=1$ with $\lambda, \tau\ge 0$, corresponds to a vanishing  mass $M_{0}=0$ for the boundary modes propagating on the background defined by a bulk curvature $\varphi=0$.
\footnotetext{
Real critical couplings naturally occur in the asymptotic limit, as $N_{x}, N_{t}$ grows to $\infty$, since the dihedral angles between the faces of the cylinder go to 0 in this continuum limit.
}
This remarkable fact extends in the complex plane to possible phases between $\tau$ and $\lambda$. Indeed, looking at the (complex) mass term for an arbitrary class angle $\vphi$, we get:
\be
M_{\vphi}^{2}
=
\det[Q_{\omega,k}(\vphi)]\Big|_{\om=k=0}
=
\big{(}1-\tau-\lambda e^{i\f\vphi{N_{t}}}\big{)}\big{(}1-\tau-\lambda e^{-i\f\vphi{N_{t}}}\big{)}
\,.
\ee
Thus  any critical couplings with $\tau\in\R$ satisfy  $\tau+|\lambda|=1$ and thus correspond to a vanishing mass for some value of the angle $\vphi$ depending on the relative phase of $\tau$ and $\lambda$. 
This provides a neat interpretation of criticality of the boundary state as massless 0-modes on the boundary leading to a divergence of the partition function.

\bigskip

In order to push the calculation of the Ponzano-Regge amplitude further, it is convenient to write the determinants in a semi-factorized form:
\begin{equation}
\label{eq:starting_point_det}
\det[Q_{\omega,k}(\varphi)]
=
(\tau^2 + \lambda^2 - 1) +2\left[ 1 - \tau \cos\left(\f{2\pi}{N_x}k\right) \right]\,  \left[1 - \lambda \cos\left(\f{\varphi}{N_t} + \f{2\pi}{N_t}\omega - \f{1}{N_t}\gamma k\right) \right] \,.
\end{equation}
In the following calculation of the Ponzano-Regge partition function, we will first assume a simplifying relation on the couplings:
\be
	\tau^2 + \lambda^2 = 1 \;,
\ee
in which case the determinants can all be factorized as:
\begin{equation}
\det[Q_{\omega,k}(\varphi)]
=
2\left[ 1 - \tau \cos\left(\f{2\pi}{N_x}k\right) \right]\,  \left[1 - \lambda \cos\left(\f{\varphi}{N_t} + \f{2\pi}{N_t}\omega - \f{1}{N_t}\gamma k\right) \right]
\,.
\end{equation}
%
%
%
The condition $\tau^{2}+\lambda^{2} = 1$ effectively decouples the spatial modes from the (twisted) temporal modes. Although this happens in momentum space and not directly in coordinate space on the original lattice, we can consider it as similar to the zero-spin recoupling channel for boundary spin network states at fixed spins considered in \cite{PRholo1}. Indeed it plays the same role of decoupling the spatial and temporal structures of the boundary state, trivializing in some sense the role of the number of time slices $N_{t}$ and thereby allowing allowing to focus solely on the effect of the twisted periodic conditions and how the Ponzano-Regge amplitude depends on the twisted angle $\gamma$. 

It is crucial to stress that this decoupling condition is not related to the criticality condition $|\tau|+|\lambda|=1$. In fact, criticality corresponds to extremal couplings satisfying $\tau^2 + \lambda^2 = 1$. For instance, considering real couplings $\tau,\lambda\in [0,1]$ satisfying $\tau^2=1- \lambda^2 $, critical couplings are $\lambda=0$ and $\lambda=1$. Thus, assuming $\tau^2 + \lambda^2 = 1$ for complex couplings does not restrict to specific behavior of the partition functions and allow to study the transition from a convergent boundary state to the critical regime. We will actually show that we recover in that critical regime the asymptotic BMS character formula for the 3D quantum gravity path  integral.

We also discuss in section \ref{sec:generalcase} the general calculation of the Ponzano-Regge amplitude when $\tau^2 + \lambda^2 \ne 1$. Evaluating the partition function as a contour integral over the class angle $\vphi$ in the complex plane, the poles and residues are not regularly spaced as in the decoupled case, but we can still write it as a finite sum over poles, which can be interpreted as a regularized and deformed BMS character.
These results, both in the decoupled and general cases, are at the  heart of the present work.

\subsection{Poles and Exact Residue Formula for the Free Ponzano-Regge Amplitude}

Assuming the decoupling condition $\tau^{2}+\lambda^{2}=1$ as discussed above, the Ponzano-Regge amplitude reads 
\be
Z^{\PR}_{\cT}(\lambda, \tau)\Big|_{\tau^{2}+\lambda^{2}=1}
\,=\,
\f{1}{\pi} \int_{0}^{2\pi} \text{d}\varphi \sin^{2}(\varphi) \prod_{\omega,k=0}^{N_t-1,N_x-1} \f{1}{\Big(2 - 2 \tau \cos\left(\f{2\pi}{N_x}k\right) \Big)  \Big(1 - \lambda \cos\left(\f{\varphi}{N_t} + \f{2\pi}{N_t}\omega - \f{1}{N_t}\gamma_k\right) \Big)}  \; .
\ee
We will refer to this case as the {\it free} Ponzano-Regge amplitude on the twisted torus, since it decouples the dispersion relations for modes propagating in the space and  time directions.
In this section, we will keep the condition $\tau^{2}+\lambda^{2}=1$ implicit and simply write  $Z^{\PR}_{\cT}$ for the Ponzano-Regge amplitude above without referring to the couplings.
The product of the $\cos$-factors over the Fourier modes can be handled using the following lemma, allowing to compute the products in terms of the Chebyshev polynomials:
\begin{lemma}
\label{prop:cosine_product}
Consider two integers $N,M\in\N^{*}$, we denote their greatest common divisor (GCD) by $P$ and call $n$ and $m$ the corresponding divisors, $N=nP$ and $M=mP$.
%
%
Then the following relation holds for all complex numbers $x,a\in\C$:
\begin{equation*}
\prod_{k=0}^{N-1} \left[2 a + 2 \cos\left( \f{2 \pi M}{N}k + x\right) \right]
=
2^{P}\Big{[} T_{n}(a) - (-1)^{n} \cos(n x) \Big{]}^P
\,,
\end{equation*}
where $T_{n}$ is the n-Chebyshev polynomial of the first kind.
\end{lemma}
\begin{proof}
In the case where the $P > 1$, we see that only $n$ terms are different in the product due to the periodicity of the cosine:
\begin{equation*}
\prod_{k=0}^{N-1} \left(2 a + 2 \cos\left( \f{2 \pi M}{N}k + x\right) \right)
=
\left[\prod_{k=0}^{n-1} \left(2a + 2\cos\left( \f{2 \pi m}{n} k + x \right) \right)\right]^{P}
\;.
\end{equation*}
Moreover, this product does not depend on $m$. Indeed, since $m$ and $n$ are coprime, we can use the periodicity of the cosine to absorb the factor $m$, leaving us with showing the equality
\begin{equation}
	\prod_{k=0}^{n-1} \left[2a + 2\cos\left( \f{2 \pi}{n} k + x \right) \right] =2 \Big{[} T_{n}(a) - (-1)^{n} \cos(n x) \Big{]} \; 
	\label{eq:proof_eq_ini}
\end{equation}
Since both sides are polynomials in $a$, it is enough to show that they have the same roots, thereby showing that the two sides are proportional. We can then finally conclude the proof by checking the proportionality factor.
%

As directly read from its expression, the left hand side has $n$ roots, parametrized by the integer $k \in [0,n-1]$:
\begin{equation*}
	a_{k}^{LHS} = - \cos\left( \f{2 \pi}{n} k + x \right)
	\,.
\end{equation*}
As for the right hand side, we use the trigonometric expression of the Chebyshev polynomial, $T_{n}(a) = \cos(n \arccos(a))$.
The roots of the right hand side are then given by an equality between cosine, with solutions labeled by an integer $k\in\N$:
\begin{equation*}
n \arccos(a_k^{RHS}) = \pm (n x + n\pi + 2\pi k) \quad \text{for} \; k \, \in \N \;,
\qquad\textrm{i.e.}\quad
a_k^{RHS} = \cos\left(\f{2 \pi k}{n} + x + \pi\right) = - \cos\left(\f{2 \pi k}{n} + x \right) \; ,
\end{equation*}
where  $k$ can be taken between $0$ and $n-1$.
Having the same roots for both sides, we  check the numerical factor for $a=\pm1$ (which can be done by looking at poles in $x$) and conclude the proof.

\end{proof}

Having in mind the special case of real positive couplings, in which case $0\le \lambda,\tau \le 1$, it is convenient to use the expression of the Chebyshev polynomial in term of hyperbolic functions:
\begin{equation*}
T_{n}(a) = \cosh(n X_a) \qquad \text{with the notation} \quad X_{a} = \arcosh(a)
\,,\quad\mathfrak{Re}(X_{a})\ge 0
\,.
\end{equation*}
Remembering the parity of the Chebyshev polynomials, $T_{n}(-a) = (-1)^{n} T_{n}(a)$, the product over the spatial modes $k$ gives:
\be
\prod_{k=0}^{N_x-1} \left[2 - 2 \tau \cos\left(\f{2\pi}{N_x}k\right) \right]
=
2\,\tau^{N_x} \left( \ch(N_x X_{\tau^{{-1}}}) - 1 \right)
=
4 \tau^{N_x} \sh^{2}\left( \f{N_x X_{\tau^{-1}}}{2}\right)
\; .
\ee
We can deal with the product over the temporal modes in a similar fashion:
\be
\prod_{\omega=0}^{N_t-1} \bigg[1 - \lambda \cos\left( \f{2\pi}{N_t}\omega+\f{\varphi}{N_t}  - \f{1}{N_t}\gamma k\right)\bigg]
=
\f{\lambda^{N_t}}{2^{N_t}} \, 2\Big[\ch(N_t X_{\lambda^{-1}}) -\cos\left( \varphi - \gamma k \right) \Big]
\; .
\ee
Combining these two results, the Ponzano-Regge partition function simplifies to:
\begin{equation}
\label{eq:starting_point_residue_case_trivial_disp_rel}
Z^{\PR}_{\cT}
=
\f{2^{N_tN_x-2N_{t}}}{(\lambda \tau)^{N_t N_x}}
\f{1}{\sh^{2N_t}\left(\f12{N_x X_{\tau^{-1}}}\right)}
\f{1}{\pi} \int_{0}^{2\pi} \text{d} \varphi \sin^{2}(\varphi)
\prod_{k=0}^{N_x-1} \f{1}{2\big[\ch(N_t X_{\lambda^{-1}})-\cos(\varphi-\gamma k)\big]} \;.
\end{equation}
%
%
This expression  clearly identifies the poles in $\vphi$ and is thus well-suited to perform the final integration.
%
%
Indeed, the poles in $\vphi$ in the complex plane are:
\be
\vphi\equiv\f{2\pi N_{\gamma}}{N_{x}}k\pm iN_{t}X_{\lambda^{-1}} \,\,[2\pi]\,.
\ee
This structure with a linear sequence of resonances, reminiscent of a Regge trajectory, considerably simplifies the calculation of the Ponzano-Regge amplitude. This behavior is proper to the decoupled case ($\lambda^{2}+\tau^{2}=1$). The generic pole structure is less regular when considering the general case of arbitrary couplings, as we will see in section \ref{sec:generalcase}.

Calling $K$  the GCD between $N_x$ and $N_\gamma$ and introducing the respective coprime divisors,
\be
N_x = K n_x\,, \qquad N_\gamma = K n_\gamma \;,
\nn
\ee
we see that each pole has the multiplicity $K$. Let us start with the case $K=1$, which happens for instance for a simple shift $N_{\gamma}=1$ and arbitrary $N_{x}$. The case $K\ge 2$ will only introduce extra combinatorial factors and $K$-powers.
%
%
\begin{prop}
\label{propK1}
Assuming the couplings satisfy $\lambda^2 + \tau^2 = 1$, and that the shift $N_{\gamma}$ defining the  twist angle is coprime with the number of spatial nodes $N_{x}>2$, the Ponzano-Regge partition function on the twisted solid torus with a coherent spin network boundary state on the square lattice admits a closed expression:
\beq
\label{expression1}
Z^{\PR}_{\cT}
&=&
\f{2^{N_t(N_x-2)-1}}{(\lambda \tau)^{N_t N_x}}
\f{1}{\sh^{2N_t}\left(\f{1}{2}N_x X_{\tau^{-1}}\right)}
\f{N_x}{\sh(N_t X_{\lambda^{-1}})}
\prod_{k=1}^{N_x-1} \f{1}{2(\ch(N_t X_{\lambda^{-1}}) - \cos(k\gamma+i N_t X_{\lambda^{-1}}))}
\\
&=&
\f{2^{N_t(N_x-2)-1}}{(\lambda \tau)^{N_t N_x}}
\f{1}{\sh^{2N_t}\left(\f{1}{2}N_x X_{\tau^{-1}}\right)}
\f{1}{\sh(N_{x}N_t X_{\lambda^{-1}})}
\,.
\label{expression2}
\eeq
\end{prop}
%

\begin{proof}

We compute the Ponzano-Regge amplitude applying the residue theorem to the integral over $\varphi$. 
The poles are directly read from equation \eqref{eq:starting_point_residue_case_trivial_disp_rel}. There are $2N_x$ simple poles at the following positions, labelled by an integer $k=0, \dots,(N_{x}-1)$,
\begin{equation}
\varphi_{k}^{\pm} = \gamma k \mp i N_t X_{\lambda^{-1}} \;.
\end{equation}
%
%
Note that the case $\lambda = 1$ is peculiar. For this critical value, we have $X_{1} = 0$ and the poles are real and are half as many, $\varphi_{k}^{+}=\varphi_{k}^{-}$. Moreover, the mode $\varphi_{0}=0$ is not a pole due to the sine square factor.
Thus assuming that we are away from the critical coupling, $\lambda \ne 1$,  we introduce the complex variable $z = e^{i \varphi}$.Under this change of variable, the integration contour becomes the unit circle $\U(1)$ in the complex plane:
\be
Z^{\PR}_{\cT}
=
\f{2^{N_t(N_x-2)}}{(\lambda \tau)^{N_t N_x}} \f{1}{\sh^{2N_t}\left(\f{1}{2}N_x X_{\tau}\right)}\f{1}{\pi} \int_{U(1)} \f{\text{d}z}{i} \f{-1}{4}\f{(z^2-1)^2}{z^2} z^{N_x-1} \prod_{k=0}^{N_x-1} \f{-e^{i \gamma_k}}{(z-z_k^+)(z-z_k^-)}
\,.
\nn
\ee
 As long as $N_x > 2$, the origin $z = 0$ is not a pole.
The poles are
\begin{equation*}
	z_{k}^{\pm} = e^{i \gamma k} e^{\pm N_t X_{\lambda^{-1}}}
	\;.
\end{equation*}
As long as  $|\lambda| \neq 1$, the positive branch of poles are outside the unit circle while the negative branch of poles are within the unit circle, $|z_k^{+}| \geq 1$ and $|z_k^{-}| \leq 1$.
%
%
Thus the residue formula only involves the negative branch:
\beq
Z^{\PR}_{\cT}
&=&
\f{2^{N_t(N_x-2)}}{(\lambda \tau)^{N_t N_x}} \f{1}{\sh^{2N_t}\left(\f{1}{2}N_x X_{\tau^{-1}}\right)} 2 \sum_{n=0}^{N_x-1} \f{-1}{4}\f{((z_n^-)^2 -1)^2}{(z_n^-)^2} (z_n^-)^{N_x-1} \f{-e^{i \gamma_n}}{z_n^- -z_n^+} \prod_{\substack{k=0 \\ k \neq n}}^{N_x-1} \f{-e^{i \gamma_k}}{(z_n^- -z_k^+)(z_n^- -z_k^-)}
\nn
\\
&=&
\f{2^{N_t(N_x-2)}}{(\lambda \tau)^{N_t N_x}} \f{1}{\sh^{2N_t}\left(\f{1}{2}N_x X_{{\tau^{-1}}}\right)} \f{1}{\sh(N_t X_{\lambda^{-1}})}\sum_{n=0}^{N_x-1} \sin^2\left(\gamma n + i N_t X_{\lambda^{-1}} \right)  \prod_{k=1}^{N_x-1} \f{1}{2(\ch(N_t X_{\lambda^{-1}}) - \cos(\gamma k+i N_t X_{\lambda^{-1}}))}
\;.
\nn
\eeq
%
The sum over $n$ can be explicitly computed\footnotemark{}, yielding a factor $N_x/2$.
\footnotetext{
This result holds only if $N_\gamma \neq N_x/2$. In that case, the sine does not depend on $n$ explicitly, and the sum returns $-N_x \sh^2(N_t X_{\lambda^{-1}})$.
}
The  result is thus
\begin{equation}
Z^{\PR}_{\cT}
=
\f{2^{N_t(N_x-2)-1}}{(\lambda \tau)^{N_t N_x}} \f{1}{\sh^{2N_t}\left(\f{1}{2}N_x X_{\tau^{-1}}\right)} \f{N_x}{\sh(N_t X_{\lambda^{-1}})} \prod_{k=1}^{N_x-1} \f{1}{2(\ch(N_t X_{\lambda^{-1}}) - \cos(\gamma k+i N_t X_{\lambda^{-1}}))} \; . 
	\label{eq:case_trivial_disp_rel_amp_final}
\end{equation}
Finally we can apply the product lemma \ref{prop:cosine_product} as before to compute the remaining product over poles to get the desired exact evaluation.

\end{proof}

Let us compare the two expressions \eqref{expression1} and \eqref{expression2} for $Z^{\PR}_{\cT}$ given in the proposition above.
On the one hand, the final formula \eqref{expression2} shows explicitly that  $Z^{\PR}_{\cT}$ does not depend on the twist angle $\gamma$ as long as the shift $N_{\gamma}$ is coprime with $N_{x}$. On the other hand, the less compact formula \eqref{expression1} gives the explicit expansion of the partition function in the momentum basis and shows  the poles in $\gamma$. As we will discuss in the next section, this latter expression allows for a clearer continuum limit, where we will recover the BMS character formula for the 3D quantum gravity path integral, as derived in \cite{Oblak:2015sea,Barnich:2015mui,Bonzom:2015ans}.

\medskip

We can generalize this formula to the case of an arbitrary non-trivial GCD between $N_{\gamma}$ and $N_{x}$:
\begin{prop}
\label{propK}
Calling $K$ the GCD between $N_{\gamma}$ and $N_{x}>2$, the Ponzano-Regge partition function on the twisted solid torus for a coherent spin network boundary state on the square lattice, with couplings satisfying $\lambda^2 + \tau^2 = 1$, admits a closed expression:
\beq
Z^{\PR}_{\cT}
&=&
\f{2^{N_t(N_x-2)-K}n_{x}^K}{(\lambda \tau)^{N_t N_x}\sh^{2N_t}\left(\f{1}{2}N_x X_{\tau^{-1}}\right) \sh^K(N_t X_{{\lambda^{-1}}})}
L_{K-1}\left( \f{\ch(n_x N_t X_{\lambda^{-1}})}{\sh(n_x N_t X_{\lambda^{-1}})} \right)
\prod_{k=1}^{N_x-1} \f{1}{2(\ch(N_t X_{{\lambda^{-1}}}) - \cos(\gamma k+i N_t X_{{\lambda^{-1}}}))}
\nn
 \\
&=&
\f{2^{N_t(N_x-2)-K}}{(\lambda \tau)^{N_t N_x}}
\f{1}{\sh^{2N_t}\left(\f{1}{2}N_x X_{\tau^{-1}}\right) \sh^K(n_x N_t X_{\lambda^{-1}})}
\,
L_{K-1}\left( \f{\ch(n_x N_t X_{\lambda^{-1}})}{\sh(n_x N_t X_{\lambda^{-1}})} \right)
\,,
\label{expressionK}
\eeq
where $L_n$ are the Legendre polynomials.
	
\end{prop}

\begin{proof}

We will not use compute the integral by residue, but  compute it directly by brute force.
We recall that the products over the spatial and temporal modes return:
\begin{equation*}
\prod_{\omega,k} \left[2 - 2 \tau \cos\left(\f{2\pi}{N_x}k\right) \right]
\,=\, 
\prod_{\om}4 \tau^{N_x} \sh^{2}\left( \f12{N_x X_{\tau^{-1}}}\right)
\,=\,
\left[
4 \tau^{N_x} \sh^{2}\left( \f12{N_x X_{\tau^{-1}}}\right)
\right]^{N_{t}}
\,,
\end{equation*}
and
\beq
\prod_{\omega,k} \left[1 - \lambda \cos\left(\f{\varphi}{N_t} + \f{2\pi\om}{N_t} - \f{2\pi N_{\gamma}k}{N_tN_{x}}\right)\right]
&=&
\prod_{k}\f{\lambda^{N_t}}{2^{N_t}} \, 2\left[\ch(N_t X_{\lambda^{-1}}) -\cos\left( \varphi - \f{2\pi N_{\gamma}k}{N_{x}} \right) \right]
\nn\\
&=&
\f{\lambda^{N_t N_x}}{2^{N_t N_x}} \Big[2\left(\ch(N_t n_x X_{\lambda^{-1}}) -\cos\left( n_x \varphi \right) \right) \Big]^K \; .
\eeq
The Ponzano-Regge amplitude  then reads:
\begin{equation}
Z^{\PR}_{\cT}
=
\f{2^{N_t(N_x-2)-K}}{(\lambda \tau)^{N_t N_x}} \f{1}{\sh^{2N_t}\left(\f{1}{2}N_x X_{\tau^{-1}}\right)} \f{1}{\pi} \int_{0}^{2\pi} \text{d} \varphi \sin^{2}(\varphi) \left[\f{1}{\ch(N_t n_x X_{\lambda^{-1}}) -\cos\left( n_x \varphi \right)}\right]^K \;.
\end{equation}
We  expand the fraction as an infinite sum over powers of cosines,
\begin{equation*}
\int_{0}^{2 \pi} \text{d}\varphi \sin^{2}(\varphi) \f{1}{\big[\ch(n_x N_t X_{\lambda^{-1}})-\cos(n_x \varphi)\big]^{K}}
=
\sum_{n=0}^{\infty} (-1)^n \binom{K+n-1}{n} \f{1}{\ch^{K+n}(n_x N_t X_{\lambda^{-1}})} \int_{0}^{2 \pi} \text{d} \varphi \sin^{2}(\varphi) \cos^{n}(n_x \varphi) \; .
\end{equation*}
and compute the trigonometric integrals\footnotemark, distinguishing even and odd power cases,
\be
\f1\pi\int_{0}^{2\pi}\text{d} \varphi \sin^{2}(\varphi) \cos^{2n}(n_x \varphi) = \f{\pi}{2^{2n}} \binom{2n}{n} 
\,,\qquad\qquad
\f1\pi\int_{0}^{2\pi}\text{d} \varphi \sin^{2}(\varphi) \cos^{2n+1}(n_x \varphi) = 0
\,.\nn
\ee
\footnotetext{
A quick way to proceed is to expand the sines and cosines into phases,
\begin{equation}
	\int_{0}^{2 \pi} \text{d} \varphi \sin^{2}(\varphi) \cos^{n}(n_x \varphi)= \int_{0}^{2 \pi} \text{d} \varphi \f{-1}{4}(e^{2i\varphi}+e^{-2i\varphi}-2) \sum_{m=0}^{n} \binom{n}{m} e^{i(n-2m)n_x \varphi}
	\,.\nn
\end{equation}
Assuming $n_x > 2$,  only the terms $2m = n$ contribute and give the desired result.
}
This turns the integral into a infinite series:
\begin{equation}
Z^{\PR}_{\cT}
=
\f{2^{N_t(N_x-2)-K}}{(\lambda \tau)^{N_t N_x}}
\f{1}{\sh^{2N_t}\left(\f{1}{2}N_x X_{\tau^{-1}}\right) }
\sum_{n=0}^{\infty} \frac{(K+2n-1)!}{(n!)^{2}(K-1)!} \f{1}{2^{2n}\ch^{K+2n}(n_x N_t X_{\lambda^{-1}})}
\,,
\end{equation}
where we recognize a hypergeometric function ${}_{2}F_1$:
\begin{equation}
Z^{\PR}_{\cT}
=
\f{2^{N_t(N_x-2)-K}}{(\lambda \tau)^{N_t N_x}}
\f{1}{\sh^{2N_t}\left(\f{1}{2}N_x X_{\tau^{-1}}\right) \ch^K(n_x N_t X_{\lambda^{-1}})}
{}_2F_1\left(\f{K}{2},\f{K+1}{2};1;\f{1}{\ch^2(n_x N_t X_{\lambda^{-1}})}\right)
\,.\nn
\end{equation}
This hypergeometric series can actually be exactly re-summed into a Legendre polynomial: 
\begin{equation}
Z^{\PR}_{\cT}
=
\f{2^{N_t(N_x-2)-K}}{(\lambda \tau)^{N_t N_x}}
\f{\mathrm{Sign}[\mathfrak{Re}(X_{\lambda^{-1}})]}{\sh^{2N_t}\left(\f{1}{2}N_x X_{\tau^{-1}}\right) \sh^K(n_x N_t X_{\lambda^{-1}})}
\,
L_{K-1}\left( \f{\ch(n_x N_t X_{\lambda^{-1}})}{\sh(n_x N_t X_{\lambda^{-1}})} \right)
\,.
\end{equation}
Since we defined the inverse hyperbolic cosine with positive real part, the sign disappears.

When the GCD is $K=1$, we recover the same result as from the product over poles in the previous calculation by residue. Here we can back-engineer the pole decomposition by expanding a posteriori $\sh^K(n_x N_t X_{\lambda})$ as a product over spatial modes:
\begin{equation*}
Z^{\PR}_{\cT}
=
\f{2^{N_t(N_x-2)-K}n_{x}^K}{(\lambda \tau)^{N_t N_x}\sh^{2N_t}\left(\f{1}{2}N_x X_{\tau^{-1}}\right) \sh^K(N_t X_{{\lambda^{-1}}})}
\prod_{k=1}^{N_x-1} \f{1}{2(\ch(N_t X_{{\lambda^{-1}}}) - \cos(\gamma k+i N_t X_{{\lambda^{-1}}}))}
L_{K-1}\left( \f{\ch(n_x N_t X_{\lambda^{-1}})}{\sh(n_x N_t X_{\lambda^{-1}})} \right)
\; ,
\end{equation*}
which is valid for any $K$. 

\end{proof}



The role of $K=\mathrm{GCD}(N_{\gamma},N_{x})$ is crucial in considering the the flow under refinement of boundary lattice and the continuum limit. As underlined in \cite{PRholo1}, the case $K=1$ corresponds to an irrational twist angle while the case $K\rightarrow\infty$ corresponds to a rational twist angle. More precisely, the twist angle is defined in the discrete setting from the shift as $\gamma=2\pi N_{\gamma}/N_{x}$. Then we consider the two cases, at fixed number of time slices $N_{t}$:
\begin{itemize}
\item the limit towards $\gamma\notin 2\pi\Q$: \\
\noindent
We construct a sequence of pairs of coprime integers $(N_{\gamma}^{(p)},N_{x}^{(p)})_{p\in\N}$ such that the limit of the ratios $2\pi N_{\gamma}^{(p)}/N_{x}^{(p)}$ converge towards $\gamma$ when $p$ goes to infinity. Then we define the Ponzano-Regge amplitude for a twisted torus with irrational twist as the limit of the partition function $Z^{\PR}_{\cT}$ for the square lattice with $N_{x}^{(p)}$ spatial nodes and a gluing shift $N_{\gamma}^{(p)}$. Since their two numbers are coprime by definition, their GCD is always 1, so the partition function $Z^{\PR}_{\cT}$ is given by proposition \ref{propK1} and does not depend on the twist angle $\gamma$ at the end of the day:
\be
Z^{\PR}_{\cT}
=
\f{2^{N_t(N_x-2)-1}}{(\lambda \tau)^{N_t N_x}}
\f{1}{\sh^{2N_t}\left(\f{1}{2}N_x X_{\tau^{-1}}\right)}
\f{1}{\sh(N_{x}N_t X_{\lambda^{-1}})}
=
\cA\,{ z}^{\PR}
\,,
\nn
\ee
where we have distinguished the pre-factor $\cA$ from a reduced Ponzano-Regge amplitude $z^{\PR}$, with the following asymptotics as $N_{x}$ is sent to infinity:
\be
\cA=
\f{2^{N_t(N_x-2)}}{(\lambda \tau)^{N_t N_x}}
\f{1}{\sh^{2N_t}\left(\f{1}{2}N_x X_{\tau^{-1}}\right)}
\underset{N_{x}\rightarrow\infty}{\sim}
\f1{2^{2N_{t}}}\,\left[\f{2}{(1+\lambda)\lambda} \right]^{N_{x}N_{t}}
\,,
\ee
\be
z^{\PR}
=
\f{1}{2\,\sh(N_{x}N_t X_{\lambda^{-1}})}
\underset{N_{x}\rightarrow\infty}{\sim}
\f1{2}\,\left[\f{1+\tau}{\lambda} \right]^{N_{x}N_{t}}
\,.
\ee

\item the limit towards $\gamma\in 2\pi\Q$: \\
\noindent
We choose the fundamental representation of the twist angle as a fraction $\gamma=2\pi n_{\gamma}/n_{x}$ with $n_{\gamma}$ and $n_{x}$ coprime. Then the infinite refinement limit is taken by considering the sequence of lattices with $(N_{\gamma},N_{x})=(Kn_{\gamma},Kn_{x})$ as the integer $K$ grows to infinity. The partition function $Z^{\PR}_{\cT}$ is given by proposition \ref{propK}, where we factorize out the same pre-factor $\cA$ as in the previous case:
\be
Z^{\PR}_{\cT}
=
\f{2^{N_t(N_x-2)-K}}{(\lambda \tau)^{N_t N_x}}
\f{1}{\sh^{2N_t}\left(\f{1}{2}N_x X_{\tau^{-1}}\right) \sh^K(n_x N_t X_{\lambda^{-1}})}
\,
L_{K-1}\left( \f{\ch(n_x N_t X_{\lambda^{-1}})}{\sh(n_x N_t X_{\lambda^{-1}})} \right)
=
\cA z_{K}^{\PR}
\,.
\nn
\ee
The reduced Ponzano-Regge amplitude now has a different behavior as $K$ is sent to infinity while keeping $n_{x}$ fixed and finite:
\be
z^{\PR}_{K}
=
\left[\f{1}{2\,\sh(n_x N_t X_{\lambda^{-1}})}\right]^{K}
\,
L_{K-1}\left( \f{\ch(n_x N_t X_{\lambda^{-1}})}{\sh(n_x N_t X_{\lambda^{-1}})} \right)
\underset{K\rightarrow\infty}{\sim}
\f1{\sqrt{\pi K}}\,\ch\left(\f{n_x N_t X_{\lambda^{-1}}}2\right)\,
\left[\f{1}{2\,\sh\f{n_x N_t X_{\lambda^{-1}}}2}\right]^{2K}
\,.
\nn
\ee
We see that this reduced Ponzano-Regge amplitude has a different scaling than the irrational case with $K=1$. The power $K$ is natural from a geometric perspective: the vertical lines on the twisted torus form $K$ large loops (of length $n_{x}N_{t}$) as illustrated in figure \ref{fig:drawingK}.

\end{itemize}
\begin{figure}[h!]
	\begin{center}
		\begin{tikzpicture}[scale=.8]
		\foreach \i in {0,2,4}{
			\foreach \j in {0,...,3}{
				\draw (\i,\j) node[color=red] {$\bullet$};
				\draw[<-,color=red] (\i,\j-.5) --(\i,\j+.5);
			}
		}
		
		\foreach \i in {1,3,5}{
			\foreach \j in {0,...,3}{
				\draw (\i,\j) node[color=blue] {$\bullet$};
				\draw[<-,color=blue] (\i,\j-.5) --(\i,\j+.5);
			}
		}
		
		\foreach \i in {-2,0,2}{
			\draw[rounded corners=3 pt,->,color=red] (\i,4.5+0.3) --(\i,4+0.1+0.3)-- (\i+2,4-0.1+0.3)--(\i+2	,3.5+0.3)   ;
		}
		
		\foreach \i in {-1,1,3}{
			\draw[rounded corners=3 pt,->,color=blue] (\i,4.5+0.3) --(\i,4+0.1+0.3)-- (\i+2,4-0.1+0.3)--(\i+2	,3.5+0.3)   ;
		}
		
		\draw (-2,4.5+0.3) node[above]{$4$};
		\draw (-1,4.5+0.3) node[above]{$5$};
		\draw (0,-0.5) node[below]{$0$};
		\draw (1,-0.5) node[below]{$1$};
		\draw (2,-0.5) node[below]{$2$};
		\draw (3,-0.5) node[below]{$3$};
		\draw (4,-0.5) node[below]{$4$};
		\draw (5,-0.5) node[below]{$5$};
		
		\draw(-.9,0) node{$N_t$-1};
		\draw(-0.7,3) node{$0$};
		
		\end{tikzpicture}
	\end{center}
	\caption{Vertical lines of the lattice for $N_{\gamma}=2$, $N_{x}=6$ and any $N_t$. There is $K=2$ large loops on the twisted torus, represented in red and blue}
	\label{fig:drawingK}
\end{figure}
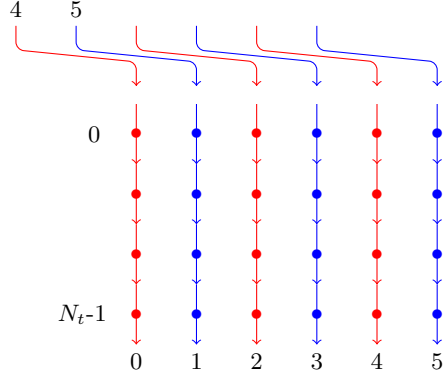
%

\subsection{General Residue Formula for the Ponzano-Regge Amplitude}
\label{sec:generalcase}

In this section, we focus on computing the Ponzano-Regge amplitude in the most general case, without assuming the condition $\lambda^{2}+\tau^{2}=1$ decoupling the temporal and spatial modes.
The computation is again done using the residue theorem, that allows us to obtain the mode expansion of the amplitude.
The starting point is equation \eqref{eq:starint_point_amp},
\begin{equation}
Z^{\PR}_{\cT}(\lambda,\tau)
=
\f{1}{\pi} \int_{0}^{2\pi} \text{d}\varphi \sin^{2}(\varphi) \prod_{\omega,k} \f{1}{\det\big[Q_{\omega,k}(\varphi)\big]} 
\nn
\end{equation}
where we recall that the mode determinants are
\begin{align}
\det\big[Q_{\omega,k}(\varphi)\big]
\,=\,
(\tau^2 + \lambda^2 - 1) +\Big[2 - 2 \tau \cos\left(\f{2\pi}{N_x}k\right) \Big]  \Big[1 - \lambda \cos\left(\f{\varphi}{N_t} + \f{2\pi}{N_t}\omega - \f{1}{N_t}\gamma k\right) \Big] \; .
\end{align}
We assume that $\lambda^{2}+\tau^{2}\ne 1$, so that the determinants are not straightforwardly factorizable. We can nevertheless write them as:
\begin{equation}
\det\big[Q_{\omega,k}(\varphi)\big]
\,=\,
\f{\tau\lambda}2\left[2\tau^{-1} - 2 \cos\left(\f{2\pi}{N_x}k\right) \right]
\left[2 C_{k} - 2\cos\left(\f{\varphi}{N_t} + \f{2\pi}{N_t}\omega - \f{1}{N_t}\gamma k\right)\right]
\; .
\end{equation}
where the coefficients $C_k$ are (to keep the notation a bit lighter, we keep the dependency on $\lambda$ and $\tau$ implicit)
\begin{equation}
	C_k = \lambda^{-1}\left[1 + \f{\tau^2 + \lambda^2 - 1}{2 - 2 \tau \cos \left(\f{2\pi}{N_x}k\right)}\right]
	\; .
\end{equation}
%
%
In the decoupled case when assuming  $\tau^2 + \lambda^2 = 1$, the second term of those coefficients drop out and the $C_k$ become all equal to $\lambda^{-1}$.
Now, in the general case with  $\tau^2 + \lambda^2 \ne 1$, the coefficients $C_{k}$ remain all different. We can nevertheless still perform the products and get the pole decomposition in the class angle $\vphi$.
The product over spatial modes remains unchanged,
\begin{equation}
\prod_{\omega,k}
\left[
2 \tau^{-1} - 2 \cos\left(\f{2\pi}{N_x}k\right)
\right]
=
2^{N_t} \Big[ \ch(n_x X_{\tau^{-1}}) - 1 \Big]^{N_t}
=
2^{2N_t} \sh^{2N_t}\left( \f{1}{2}n_x X_{\tau^{-1}}\right)
\; .
\end{equation}
The product over the temporal modes is slightly modified,
\begin{equation}
\prod_{\omega,k} \left[2 C_{k} - 2\cos\left(\f{\varphi+2\pi\omega - \gamma k}{N_{t}}\right) \right]
=
\prod_{k}2\Big[ \ch(N_t X_{C_k}) - \cos\left(\varphi - \gamma k\right) \Big]
\; .
\end{equation}
The full Ponzano-Regge amplitude thus reads as a trigonometric integral,
\begin{equation}
Z^{\PR}_{\cT}(\lambda,\tau)
=
\f{2^{N_t(N_x-2)}}{(\lambda \tau)^{N_t N_x}}
\f{1}{\sh^{2N_t}\left(\f{1}{2}N_x X_{\tau}\right)}
\f{1}{\pi} \int_{0}^{2\pi} \text{d} \varphi \sin^{2}(\varphi)
\prod_{k=0}^{N_x-1} \f{1}{2\left[ \ch(N_t X_{C_k}) - \cos\left(\varphi - \gamma k\right) \right]}
\;.
\end{equation}
The poles can be directly read from this expression:
\be
\varphi_{k}^{\pm} = \gamma k \mp i N_t X_{C_k}
\,.
\ee
As illustrated in figure \ref{fig:generalpoles}, these are not linear anymore as in the decoupled case with  $\tau^2 + \lambda^2 = 1$. Moreover the dependency on $k$ through the coefficients $C_k$ lifts the degeneracy of the poles due to the GCD $K$. Now, whatever the GCD between the shift $N_{\gamma}$ and the number of spatial nodes $N_{x}$, the poles in $\vphi$ remain simple and there is no degeneracy. This allows for a direct evaluation of the integral by residue, summarized in the following proposition.
\begin{figure}[h!]
\begin{subfigure}[t]{.4\linewidth}
\includegraphics[scale=0.45]{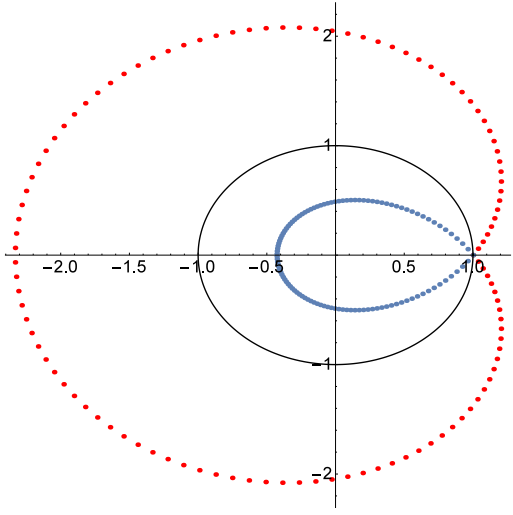}
\caption{Poles for a shift $N_{\gamma}=1$ and critical couplings $(\lambda,\tau)=(0.6,0.4)$}
		
\end{subfigure}
\hspace{8mm}
\begin{subfigure}[t]{.4\linewidth}
\includegraphics[scale=0.45]{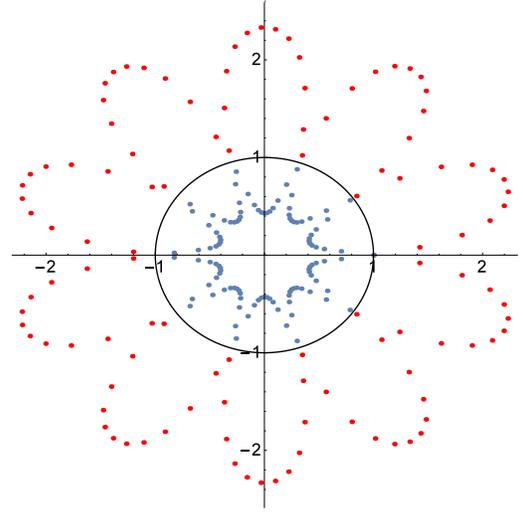}
\caption{Poles for a shift $N_{\gamma}=11$ and critical couplings $(\lambda,\tau)=(0.6,0.4)$}
		
\end{subfigure}
\begin{subfigure}[t]{.4\linewidth}
\includegraphics[scale=0.45]{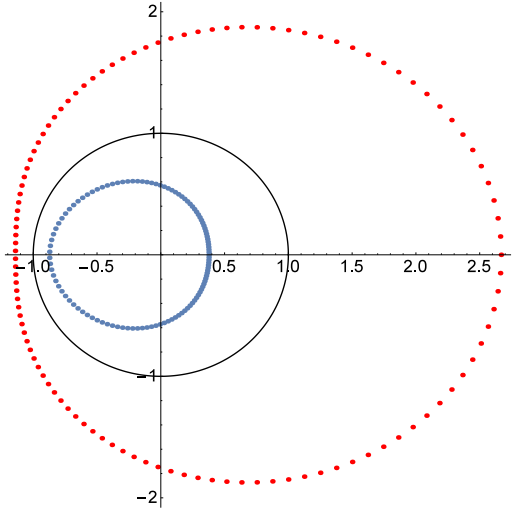}
\caption{Poles for a shift $N_{\gamma}=1$ and non-critical couplings $(\lambda,\tau)=(1.6,0.4)$}
		
\end{subfigure}
\hspace{8mm}
\begin{subfigure}[t]{.4\linewidth}
\includegraphics[scale=0.45]{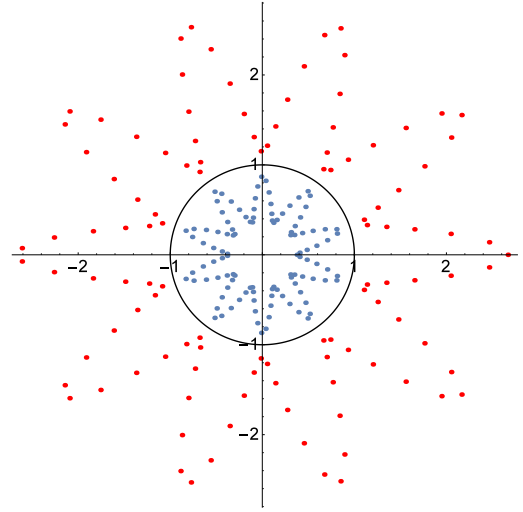}

\caption{Poles for  $N_{\gamma}=11$ and non-critical couplings $(\lambda,\tau)=(1.6,0.4)$}
		
\end{subfigure}

	\caption{Plot of the poles in the complex planes, $z_{k}^{\pm} = e^{i \gamma k} e^{\pm N_t X_{C_k}}$, for a square lattice on the twisted torus determined by the numbers of nodes $N_{t}=1$ and $N_{x}=111$. The red dots represent $z^{+}_{k}$ while the blue ones $z^{-}_{k}$. They are related by inversion with respect to the unit circle  (in black), since they have inverse modulus for equal phases. For critical couplings, on the top plots, $z=1$ is a pole, which leads to a divergent integral, while there is no pole on the unit circle for non-critical couplings as shown on the bottom plots.
	}
	\label{fig:generalpoles}
\end{figure}
\begin{prop}
The Ponzano-Regge partition function on the twisted solid torus for a coherent spin network boundary state on the square lattice, with couplings satisfying $\lambda^2 + \tau^2 \ne 1$, can be written as a finite sum with a clear pole structure in the twist angle $\gamma$:
\begin{equation}
\label{Zgeneral}
Z^{\PR}_{\cT}(\lambda,\tau)
=
\f{2^{N_t(N_x-2)}}{(\lambda \tau)^{N_t N_x}} \f{1}{\sh^{2N_t}\left(\f{1}{2}N_x X_{\tau^{-1}}\right)} \sum_{n=0}^{N_x-1} \f{\sin^2\left(\gamma n + i N_t X_{C_k} \right)}{\sh(N_t X_{C_n})}  \prod_{\substack{k=0 \\ k \neq n}}^{N_x-1} \f{1}{2\big[\ch(N_t X_{C_k}) - \cos(\gamma(n-k)+i N_t X_{C_n}\big]} \; .
\end{equation}
\end{prop}

\begin{proof}

To compute the integral by residue, we  perform the change of variable $z = e^{i \varphi}$ as before. The poles in the complex plane are then $z_{k}^{\pm} = e^{i \gamma k} e^{\pm N_t X_{C_k}}$, with
\begin{equation}
	e^{X_{C_k}} = C_k + \sqrt{C_k^2-1} \; .
\end{equation}
The factorized form of the amplitude is exactly the same as before and the integration over the unit circle $\U(1)$ gets contributions only from the poles $z_{k}^{-}$ whose norm is lesser than 1:
\beq
Z^{\PR}_{\cT}(\lambda,\tau)
&=&
\f{2^{N_t(N_x-2)}}{(\lambda \tau)^{N_t N_x}} \f{1}{\sh^{2N_t}\left(\f{1}{2}N_x X_{\tau}\right)} \f{1}{\pi} \int_{U(1)} \f{\text{d}z}{i} \f{-1}{4}\f{(z^2-1)^2}{z^2} z^{N_x-1} \prod_{k=0}^{N_x-1} \f{-e^{i \gamma k}}{(z-z_k^+)(z-z_k^-)} 
\nn\\
&=&
\f{2^{N_t(N_x-2)}}{(\lambda \tau)^{N_t N_x}} \f{1}{\sh^{2N_t}\left(\f{1}{2}N_x X_{\tau}\right)} 2 \sum_{n=0}^{N_x-1} \f{-1}{4}\f{((z^{-}_n)^2-1)^2}{(z^{-}_n)^2} \f{- e^{i \gamma n}}{z^{-}_n - z^{+}_n} \prod_{\substack{k=0 \\ k\neq n}}^{N_x-1} \f{- z_{n}^{-} e^{i \gamma k}}{(z_n^- -z_k^+)(z_n^- -z_k^-)}
\,,\nn
\eeq
which gives the announced result.
%

\end{proof}

We do not know how to re-sum over the pole label $n$ due to non-linearity of the poles in the complex plane, but the expression \eqref{Zgeneral} has a clear physical interpretation as an expansion over the poles in the twist angle $\gamma$ of the Ponzano-Regge partition function. This allows a clear comparison with other formulas for the partition function of 3D quantum gravity, such as the BMS character formula which we discuss below in the next section.

\subsection{The Continuum Limit: Recovering the BMS character from the Ponzano-Regge model}

We would like to compare the continuum limit of our Ponzano-Regge partition function formulas with the computation of the partition function of asymptotic flat 3D gravity for a solid twisted  torus as a BMS character as derived in  \cite{Barnich:2015mui}:
\begin{equation}
Z_{\textrm{3D flat gravity}}^{\textrm{asympt}}[\gamma] = \chi_{\text{BMS}}[\gamma]
\,,\qquad\textrm{with}\quad
\chi_{BMS}[\gamma]= e^{S_0}\prod_{k=2}^{\infty} \f{1}{2-2\cos(\gamma k)}\,,
\end{equation}
where $\chi_{BMS}$ is a character of the Bondi-Metzner-Sachs (BMS) group.
Here, $S_{0}$ is the on-shell action. For a real twist angle $\gamma\in\R$, the infinite product above is just a formal definition. It actually has poles at every rational twist angle, $\gamma\in2\pi\Q$. The BMS character is in fact well-defined on the upper complex plane, and can be identified as the Dedekind $\eta$-function up to a factor. 

This formula was also re-derived as an asymptotic limit of  Regge calculus for discretized 3D gravity in \cite{Bonzom:2015ans}. It was further recovered as leading order of a WKB approximation of the Ponzano-Regge model with LS spin network boundary states in \cite{PRholo2}. These LS spin networks consist in coherent intertwiners with fixed spins and allowed the study of the Ponzano-Regge amplitude by a saddle point approximation in the large spin regime (i.e. large edge lengths), which gives:
\begin{equation}
Z^{\textrm{1-loop PR}}
\,=\,
e^{i S_{0}} \left[\sum_{n=1}^{N_x-1} A_{n}\,4\sin^2\f{\gamma n}{2}\right]\, \prod_{k=1}^{\f{N_x-1}{2}} \f{1}{2-2 \cos(\gamma k)} \; .
\end{equation}
We refer to this leading order behavior of the Ponzano-Regge amplitude as 1-loop by similarity with quantum field theory calculations.
The twist angle $\gamma$ is defined from the discrete shift as $2\pi N_{\gamma}/N_{x}$.
The integer $n$ labels the saddle point and describes geometrically the winding number of the embedding of the boundary 2-torus (in the spatial direction).
The coefficient $A_{n}$ is a rather complicated pre-factor which depends on the winding number but is independent from the twist parameter.
%
%
There are two main differences with the BMS formula:
\begin{itemize}

\item There is a sum over winding numbers $n$ for the spatial geometry of the torus. This is a non-perturbative effect.

\item The product over modes start at $k=1$ instead of $k=2$. As explained in \cite{Oblak:2015sea}, this corresponds to a massive mode, i.e. working with a massive BMS representation. This is also a non-perturbative effect.

\end{itemize}
It is nevertheless possible to recover the BMS character formula as a truncation of $Z^{\textrm{1-loop PR}}$ in the asymptotic limit $N_{x}\rightarrow\infty$. In fact, the factor $\sin^2{\gamma n}/{2}$ for each winding number $n$ kills a factor in the product determinant, leading to
\begin{equation}
Z^{\textrm{1-loop PR}}
\,=\,
\sum_{n=1}^{N_x-1}A_{n}  \prod_{\substack{k=1 \\ k\neq n}}^{\f{N_x-1}{2}} \f{1}{2-2 \cos(\gamma k)} \; .
\end{equation}
If we truncate the sum to a trivial winding number $n=1$, which is the only allowed embedding of the ``spacetime'' cylinder in Euclidean 3D space, then this leads back as wanted to a finite version of the character for a massless irreducible representation of the BMS group. We can then take the asymptotic limit $N_{x}\rightarrow\infty$  corresponding to a torus whose radius grows to infinity.
The difficulty in making sense of this limit is that, in the Ponzano-Regge context, we necessarily deal with a real twist angle defined by the combinatorics of the boundary lattice and that the infinite product is ill-defined in that case.

\medskip

The present work considerably improves on this previous 1-loop approximation of the Ponzano-Regge partition function. First, we compute the exact partition function for a quantum boundary state, with no saddle point approximation or large spin limit. Second, we naturally get an imaginary shift, which regularizes the amplitude.
Let us look into this exact Ponzano-Regge amplitude in more details.

Starting with the decoupling ansatz, with couplings satisfying $\lambda^{2}+\tau^{2}=1$, it is convenient to write the Ponzano-Regge amplitude in a factorized form distinguishing the amplitude pre-factor and the sum over poles,
\be
Z^{\PR}_{\cT}= \cA\,z^{\PR}_{\cT}
\qquad\textrm{with}\quad
\cA=
\f{2^{N_t(N_x-2)}}{(\lambda \tau)^{N_t N_x}}
\f{1}{\sh^{2N_t}\left(\f{1}{2}N_x X_{\tau^{-1}}\right)}\,,
\ee
and the reduced Ponzano-Regge amplitude:
\be
z^{\PR}_{\cT}
=
\sum_{n=0}^{N_x-1} \f{\sin^2\left(\gamma n + i N_t X_{\lambda^{-1}} \right)}{\sh(N_t X_{\lambda^{-1}})}
\prod_{k=1}^{N_x-1} \f{1}{2(\ch(N_t X_{\lambda^{-1}}) - \cos(k\gamma+i N_t X_{\lambda^{-1}}))}
\ee
The pre-factor $\cA$ plays the role of the exponential of the on-shell action, while the reduced Ponzano-Regge $z^{\PR}_{\cT}$ encodes the whole pole structure of the amplitude. 
It is clear that we would like to take the critical limit $\lambda\rightarrow 1$, and thus $\tau\rightarrow0$, to recover the BMS character formula.
Geometrically, this corresponds to keeping the time span of the cylinder fixed while sending its radius to $\infty$, thus looking at the asymptotic limit in space while remaining at finite time.
In this limit, $X_{\lambda^{-1}}\rightarrow0$, the amplitude pre-factor diverges due to the factor $\tau^{-N_{x}N_{t}}$, while the  only divergent term in the reduced Ponzano-Regge amplitude is the factor $\sh(N_t X_{\lambda^{-1}})$:
\be
\sh(N_t X_{\lambda^{-1}})\,z^{\PR}_{\cT}
\underset{\substack{\lambda\rightarrow1 \\\tau\rightarrow0}}\sim
\left(\sum_{n=1}^{N_x-1} \sin^2\gamma n\right) \,
\prod_{k=1}^{N_x-1} \f{1}{2 - 2\cos (\gamma k )}
=
\sum_{n=1}^{N_x-1}\cos^{2}\f{\gamma n}2 \prod_{\substack{k=1 \\ k\neq n}}^{\f{N_x-1}{2}} \f{1}{2-2 \cos(\gamma k)}
\,,
\ee
where we recognize once again the term  with trivial winding number $n=1$ as a BMS character for a massless representation.
From this perspective, it is intriguing that the non-perturbative sum over winding numbers creates the equivalent of a mode $k=1$,
which would otherwise be absent from the semi-classical calculation, leading to a BMS character for a massive representation.

\medskip

To be more precise, we should take the infinite refinement limit $N_{t}\rightarrow \infty$ and describe the relative scaling of the coupling $\lambda$ to ensure this limit. If $\lambda=1-\eps$ with $\eps\rightarrow 0$, then $X_{\lambda^{-1}}\sim\sqrt{2\eps}$. So if $\lambda$ goes to 1 faster than $N_{t}^{-2}$, i.e. if $1-\lambda\propto 1/N_{t}^{2+\sigma}$ with $\sigma>0$, then we are clearly in the case described above.

However, there is a critical regime of the scaling limit, considering $\lambda=1-\alpha^{2} /2N_{t}^{2}$ for $\sigma=0$, in which case the couplings still converge to their critical value, but then we keep a finite imaginary shift in the poles:
\be
z^{\PR}_{\cT}
\quad\underset{N_{t}\rightarrow\infty}{\overset{\lambda=1-\f{\alpha^{2}}{2N_{t}^{2}}}{\sim}}\quad
\sum_{n=1}^{N_x-1} \f{\sin^2(\gamma n+i\alpha)}{\sh \alpha} \,
\prod_{k=1}^{N_x-1} \f{1}{2(\ch\alpha - \cos (\gamma k +i\alpha))}
\,.
\ee
Ignoring the sum over winding modes, and taking the limit $N_{x}\rightarrow\infty$, we obtain a deformation of the BMS character regularized by this complex shift:
\be
\chi_{\textrm{reg}}[\gamma,\alpha]=
\prod_{k=1}^{N_x-1} \f{1}{2(\ch\alpha - \cos (\gamma k +i\alpha))}
\,.
\ee
Interestingly, comparing to the derivation of the BMS character presented in \cite{Oblak:2015sea}, this amounts to the usual BMS character evaluated on the super-rotation defined by the twist angle $\gamma$ composed with a Bogoliubov transformation\footnotemark~ mixing the negative mode $-k$ with the positive mode $+k$.
\footnotetext{%
The interpretation in terms of Bogoliubov transformations begets the question of unitarity. In fact, the 2D boundary state that we are considering does not correspond to an initial or final canonical state, but describes the geometry of the ``time-like'' boundary and thereby determines the flow of time and the evolution of the bulk geometry. The natural question in that context is which boundary states ensures a unitary evolution for the geometry of the disk.
}
Indeed, the superrotation $R_{\gamma}$ acts on supermomenta $v_{k}$ as $v_{k}\mapsto e^{i\gamma k}$, while we define a Bogoliubov transformation pairing the modes propagating in opposite directions:
\be
B_{\alpha}=\mat{cc}{e^{-\alpha} & 2i\sh \f\alpha 2 \\ 2i\sh \f\alpha 2 & e^{+\alpha}} \qquad \textrm{acting on 2-vectors}\quad \mat{c}{v_{k}\\ v_{-k}}
\,.
\ee
The resulting character $\chi_{BMS}(R_{\gamma}B_{\alpha})$ is the inverse determinant of $\id-R_{\gamma}B_{\alpha}$ acting on the vector space of supermomenta:
\be
\chi_{BMS}(R_{\gamma}B_{\alpha})
=
\prod_{k\ge1}
\f1{\det_{k}(\id_{2}-R_{\gamma}B_{\alpha})}
=
\prod_{k\ge1} \f{1}{2(\ch\alpha - \cos (\gamma k +i\alpha))}
\,.
\ee

\medskip

Now, although the details of the previous limits and calculations do not go through in the general case beyond the decoupled ansatz, when $\lambda^2+\tau^2\ne 1$, the logic of going to critical couplings to recover the BMS character formula still applies. Indeed, in the general case, the amplitude pre-factor does not change at all, but the reduced Ponzano-Regge amplitude acquires a less regular pole structure:
\be
z^{\PR}_{\cT}[\lambda,\tau]
=
\sum_{n=0}^{N_x-1} \f{\sin^2\left(\gamma n + i N_t X_{C_k} \right)}{\sh(N_t X_{C_n})}  \prod_{\substack{k=0 \\ k \neq n}}^{N_x-1} \f{1}{2\big[\ch(N_t X_{C_k}) - \cos(\gamma(n-k)+i N_t X_{C_n}\big]}
\,,
\ee
with the coefficients 
\be
C_k = \lambda^{-1}\left[1 + \f{\tau^2 + \lambda^2 - 1}{2 - 2 \tau \cos \left(\f{2\pi}{N_x}k\right)}\right]
\,.
\nn
\ee
Choosing critical real positive couplings, with $\lambda+\tau=1$, it appears that the coefficients $C_{k}$ go to 1 for very low momenta $k\ll N_{x}$.  Thus the coefficients $X_{C_{k}}$ go to 0, and the general formula above once more reduces to a BMS character.
By periodicity, for very large momenta, for $k$ very close to $N_x$, the coefficient $C_k$ also approaches 1. The situation is however different for large momenta in the intermediate range,  with $k/N_x$ kept fixed in $]0,1[$, for example for $k\sim N_{x}/2$. In this case,  this approximation fails and we have to deal with the more complicated structure of the roots. Since the product over modes involves all the modes from $k=1$ to $k=N_{x}-1$, one can not ignore the effects of large momenta on the partition function, but could consider them as deep quantum gravity effects.

\medskip

Overall, we have discussed how to take the continuum limit of the boundary lattice. For critical couplings,  the exact Ponzano-Regge amplitude reads as a finitely truncated BMS character, which leads back to the BMS character---possibly with a complex shift regularizing it---in the infinite refinement limit $N_{x},N_{t}\rightarrow\infty$. Since the exact formulas and poles can be written explicitly, it is natural to wonder if  the  Ponzano-Regge amplitude can be identified as the character of a modified BMS group for finite boundary. Such a deformed BMS group would be identified as  the symmetry group of the Ponzano-Regge model with the 2D discrete boundary geometry, and would probably be the reason behind the simplification of the Ponzano-Regge partition function. This will be investigated in future work.

\section{Conclusion \& Outlook}

In the present work, we aimed at pursuing the exploration of holographic dualities in non-perturbative 3D quantum gravity  initiated in \cite{PRholo1,PRholo2,Short,Riello:2018anu}. Using the Ponzano-Regge model as the definition of a discrete quantum path integral for 3D Euclidean gravity (for vanishing cosmological constant) as a topological quantum field theory (TQFT), we focus on the case of 3d manifolds with boundaries in order to explore the correspondence between bulk geometry and boundary observables.
We are interested in how the Ponzano-Regge amplitude depends on the quantum state of the boundary geometry, having in mind both the reconstruction of bulk observables from boundary correlations and the study of  the renormalization of the 3D quantum gravity path integral under coarse-graining and refinement of the boundary state.

We focussed on the exact and explicit computation of the Ponzano-Regge path integral for a solid torus with a twist. Its 2D boundary state is defined as  a coherent spin network state living on a square lattice with the twist combinatorially  encoded in the periodicity conditions of the lattice. Such coherent boundary states define superpositions of spins, geometrically interpreted as edge lengths, which are controlled by couplings. We chose homogeneous, yet anisotropic, couplings, $\lambda$ and $\tau$ respectively controlling the lengths in the two directions of the boundary, referred to as the space and time directions despite the Euclidean signature. We showed that critical values of the couplings, satisfying $|\lambda|+|\tau|=1$, are poles of the series defining the the coherent spin network states and describe scale-invariant classes of 2D discrete geometries.

In the case of a 3-ball with a 2-sphere boundary, the Ponzano-Regge amplitude amounts to the evaluation of the boundary state on the trivial connection. In the case of the solid torus, the Ponzano-Regge amplitude still projects the boundary states on flat connections, but it retains the information of the $\SU(2)$ holonomy wrapping around the non-contractible cycle of the solid torus. This is the only remaining information about the bulk geometry. The Ponzano-Regge path integral is then the integral over the class angle $\vphi$ of that holonomy of the evaluation of the spin network state on the corresponding connection. We managed to reduce the whole combinatorial complexity, with $\SU(2)$ coherent intertwiners living on the nodes of the boundary lattice, and write the Ponzano-Regge amplitude as a trigonometric integral over the class angle. This integral appears in a factorized form with a clear structure of the poles in the class angle, allowing to write it as a complex contour integral and to compute it by residues. At the end of the day, the poles---resonances---in $\vphi$  become poles in the twist angle $\gamma$ of the Ponzano-Regge path integral.
The resulting exact formula for the Ponzano-Regge amplitude takes the form of a sum over an integer $n$ of the inverse of a factorized trigonometric polynomial whose roots are poles generically depending on $n$. This integer $n$ is interpreted as the winding number of the torus' spatial geometry.

Focussing on the $n=1$ term of this sum leads to a massless BMS character in the continuum limit (or thermodynamic limit) when the lattice size grows to infinity and the couplings become critical. This shows how to recover the standard QFT calculations of the 3D quantum gravity path integral from the non-perturbative Ponzano-Regge path integral. Moreover, we would like to underline two points on which the Ponzano-Regge amplitude goes beyond the  QFT calculations.
First, we focus on boundary state couplings satisfying $\lambda^2+\tau^2=1$, for which the boundary modes propagating in the space direction decouple from the modes propagating in the time direction, thus simplifying the pole structure and the Ponzano-Regge formula. In this case, the non-perturbative sum over winding number actually leads back to a (sum of) BMS character(s) corresponding to {\it massive} representations. In the general non-decoupled case, the pole structure is more complicated and the fit with the BMS character only works for low momenta. The behavior of the high momentum modes should be investigated further.

Second, the Ponzano-Regge amplitude can be seen as a truncated and regularized BMS character. Indeed, the BMS character is expressed in terms of the Dedekind $\eta$-function, which is ill-defined for real twist angle and actually truly defined on the upper complex plane. The Ponzano-Regge amplitude obviously truncates this function, due to the finite size of the lattice, but also introduces a complex shift of the twist angle. This can be made precise in the asymptotic limit where we identified a critical regime where the couplings approach their critical values with a precise scaling in the lattice size. The Ponzano-Regge amplitude then leads to a BMS character regularized by a imaginary shift, which can be interpreted as the trace of the super-rotation defined by the twist angle composed with a Bogoliubov transformation. The geometrical meaning of this Bogoliubov transformation and its implication for the unitarity of the amplitudes remains to be understood.

An interesting remark is nevertheless that the imaginary shift of the BMS character which we derived is analogous to -- though subtly different from -- the AdS regularization of the BMS character, which also corresponds to taking a  finite-radius regularization of the Minkowski space-time but, this time, accompanied by a corresponding modification of the gravitational dynamics. Indeed, as shown in \cite{Giombi:2008vd,Oblak:2016eij}, working out the 3D quantum gravity path integral on AdS leads to a correction of the form $[2- 2\cos k(\gamma + i \epsilon)]$ where the ratio between the time interval and the AdS curvature radius $\epsilon = T/\ell_{\textrm{AdS}}$ gives the imaginary shift of the poles similarly to the role of the parameter of the Bogoliubov transformation. Understanding whether this is meaningful or just a coincidence deserves further investigation, for instance by extending our calculations to take into account a non-vanishing cosmological constant and computing the Turaev-Viro amplitudes with 2D boundary spin network states.

Beside pushing further the investigation of the toroidal Ponzano-Regge amplitude, its various limits and regimes, possible extension to higher genus topologies or generalization to Lorentzian signature (using $\SU(1,1)$ holonomies and representations) and to a non-vanishing cosmological constant (using a $q$-deformed gauge group $\cU_{q}(\SU(2))$), an interesting direction of future investigation is the  duality of 3D quantum gravity with the 2D Ising on the boundary. Indeed it was proven in \cite{Bonzom:2015ova} that the Ponzano-Regge path integral on the 3-ball with a coherent spin network state on the boundary 2D sphere was related to the inverse of the 2D Ising partition function defined on the boundary. In particular the poles of the Ponzano-Regge amplitude correspond to critical Ising couplings. It would be enlightening to extend this duality formula to the toroidal topology, as suggested in \cite{costantino}, and provide the poles and resonances described in the present work with an interpretation in terms of the 2D Ising model on the twisted torus, as studied for instance in \cite{Matsuura:2015jda}.

Finally, it is intriguing to obtain such an exact and simple formula for the Ponzano-Regge amplitude, which does not reflect the complexity of the boundary graph. This hints towards the possibility of integrable structures, as uncovered for certain boundary  states with fixed spins on the boundary \cite{PRholo1,Riello:2018anu}. And the fact that we get a BMS character in the asymptotic limit begs the question of the symmetry group of the boundary state for a finite lattice: is the Ponzano-Regge model on the solid cylinder  with a discrete boundary invariant under a non-trivial symmetry group, such as a BMS-like symmetry group or a discrete current algebra? Understanding the boundary symmetry would not only shed light on the calculations of the Ponzano-Regge path integral but also on the structure of edge modes for 3D gravity, from the deep quantum regime to the semi-classical continuum regime.

\section*{Acknowledgement}
This work is supported by Perimeter Institute for Theoretical Physics. Research at Perimeter Institute is supported by the Government of Canada through Industry Canada and by the Province of Ontario through the Ministry of Research and Innovation.

\appendix

\section{Boundary linear forms on the 2-torus}
\label{app:torus}

In order to better understand the role of the cut cycle $\cC$ in the boundary amplitudes on the 2-torus, as derived from the Ponzano-Regge partition function in proposition \ref{prop:torus}, it is interesting to consider how to vary and choose the cut on the torus.
Indeed, instead of  the original initial and final boundary circles of the 2D cylinder which we glue onto each other to form the torus, we can choose other closed lines as cuts on the torus. Considering the original circle as the initial time slice $t=0$ and sliding it along the torus to define a periodic Morse function or time coordinate $t$ on the torus, and using the angle $x$ around the circle as space coordinate, this defines a 2D flat coordinate system $(t,x)$ on the torus. One can perform a modular transform $(t,x)\mapsto (u,y)=(at+bx,ct+dx)$ with $(a,b,c,d)\in\Z^{4}$ satisfying $ad-bc=1$ and use the new circle $u=0$ as cut on the torus. This new cut, $at+bx=0$ in the original coordinates, wraps around the torus with winding numbers $a$ and $b$ around the two original cycles of the torus.

Let us now consider the simplest graph drawn on the 2-torus, defined by a single node at the origin and the two lines $t=0$ and $x=0$, corresponding to a 2D cellular decomposition of the torus with a single face, as drawn on figure \ref{app:fig:simple_disc_torus}.
We denote by $\psi$ a quantum  boundary state on this graph. The corresponding wave-functions $\psi(g_1,g_2)$depends on two $\SU(2)$ group elements, $g_{1}$ and $g_{2}$ defining the two holonomies along the links $1$ and $2$ wrapping around the two cycles of the torus.
Following and naturally extending the result of proposition  \ref{prop:torus}, the boundary linear form defining the quantum gravity amplitude with boundary state $\psi$ counts the number of times $p_{1}$ and $p_{2}$ that the graph links wind cross the cut $\cC$ (or equivalently the number of times the cycle $\cC$ winds around the two links), thus yielding:
\begin{equation}
Z_{\cC_{(p_{1},p_{2})}}\big{[}
\psi\big{]}
=
\int \rd g \; \psi(g^{p_{1}},g^{p_{2}}) \; .
\end{equation}
It is clear that these linear forms with the two holonomies given in terms of a single group element, $(g_{1},g_{2})=(g^{p_{1}},g^{p_{2}})$, are integrals over different sectors of the moduli space of flat connections on the 2-torus.
\begin{figure}[!htb]

	\begin{subfigure}[t]{.3\linewidth}	
			\begin{tikzpicture}[scale=.7,>=stealth]
		\coordinate (A0) at (-2,-2);
		\coordinate (A1) at (-2,2);
		\coordinate (A2) at (2,2);
		\coordinate (A3) at (2,-2);
		
		\coordinate (O) at (0,0.4);
		\coordinate (O1) at (0,2);
		\coordinate (O2) at (0,-2);
		\coordinate (O3) at (2,0.4);
		\coordinate (O4) at (-2,0.4);
		
		\draw[dotted] (A0)--(A1)--(A2)--(A3)--cycle;
		
		\draw[->-=0.5,color=red] (O1)-- node[pos=0.5,left]{$1$} (O); \draw [->-=0.5,color=red] (O)-- node[pos=0.7,left]{$1$} (O2); 
		\draw[->-=0.5,color=blue] (O)-- node[pos=0.5,above]{$2$} (O3); \draw [->-=0.5,color=blue] (O4)-- node[pos=0.3,above]{$2$} (O); 
		
		\node at (O) {$\bullet$};
		
		\coordinate (C1) at (-2,-1);
		\coordinate (C2) at (2,-1);
		\coordinate (C11) at (-1,-2);
		\coordinate (C22) at (-1,2);
		\coordinate (C111) at (-2,-1.7);
		\coordinate (C222) at (2,0.4);
		
		\draw[color=orange] (C1)--(C2);
		\end{tikzpicture}
		\caption{The cut $\cC$ is purely spatial.\newline}
	\end{subfigure}
	\hspace{4mm}
	\begin{subfigure}[t]{.3\linewidth}
			\begin{tikzpicture}[scale=.7,>=stealth]
		\coordinate (A0) at (-2,-2);
		\coordinate (A1) at (-2,2);
		\coordinate (A2) at (2,2);
		\coordinate (A3) at (2,-2);
		
		\coordinate (O) at (0,0.4);
		\coordinate (O1) at (0,2);
		\coordinate (O2) at (0,-2);
		\coordinate (O3) at (2,0.4);
		\coordinate (O4) at (-2,0.4);
		
		\draw[dotted] (A0)--(A1)--(A2)--(A3)--cycle;
		
		\draw[->-=0.5,color=red] (O1)-- node[pos=0.5,left]{$1$} (O); \draw [->-=0.5,color=red] (O)-- node[pos=0.7,left]{$1$} (O2); 
		\draw[->-=0.5,color=blue] (O)-- node[pos=0.5,above]{$2$} (O3); \draw [->-=0.5,color=blue] (O4)-- node[pos=0.3,above]{$2$} (O); 
		
		\node at (O) {$\bullet$};
		
		\coordinate (C1) at (-2,-1);
		\coordinate (C2) at (2,-1);
		\coordinate (C11) at (-1,-2);
		\coordinate (C22) at (-1,2);
		\coordinate (C111) at (-2,-1.7);
		\coordinate (C222) at (2,0.4);
		
		\draw[color=orange] (C11)--(C22);
		\end{tikzpicture}
		\caption{The cut $\cC$ is purely temporal.\newline}
	\end{subfigure}
	\hspace{4mm}
	\begin{subfigure}[t]{.3\linewidth}
			\begin{tikzpicture}[scale=.7,>=stealth]
		\coordinate (A0) at (-2,-2);
		\coordinate (A1) at (-2,2);
		\coordinate (A2) at (2,2);
		\coordinate (A3) at (2,-2);
		
		\coordinate (O) at (0,0.4);
		\coordinate (O1) at (0,2);
		\coordinate (O2) at (0,-2);
		\coordinate (O3) at (2,0.4);
		\coordinate (O4) at (-2,0.4);
		
		\draw[dotted] (A0)--(A1)--(A2)--(A3)--cycle;
		
		\draw[->-=0.5,color=red] (O1)-- node[pos=0.5,left]{$1$} (O); \draw [->-=0.5,color=red] (O)-- node[pos=0.7,left]{$1$} (O2); 
		\draw[->-=0.5,color=blue] (O)-- node[pos=0.5,above]{$2$} (O3); \draw [->-=0.5,color=blue] (O4)-- node[pos=0.3,above]{$2$} (O); 
		
		\node at (O) {$\bullet$};

		\draw[color=orange] (A0)--(A2);
		\end{tikzpicture}
		\caption{The cut $\cC$ winds once around each of the two cycles.}
	\end{subfigure}

	
\caption{
The choice of cut cycle $\cC$ entering the boundary amplitude on the torus (opposite dotted lines are identified) on the
simplest boundary graph on the 2-torus, with a single node and two links in red and blue.}
\label{app:fig:simple_disc_torus}
\end{figure}
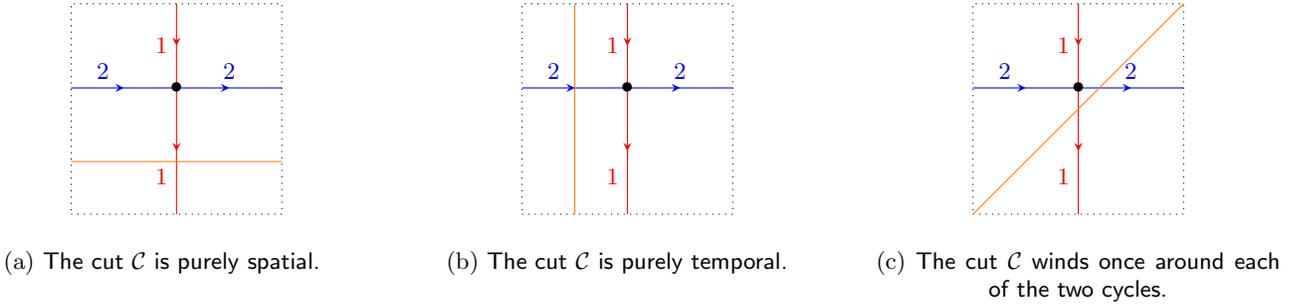

Three examples of cut cycles are depicted in figure \ref{app:fig:simple_disc_torus} and lead to three different, in the sense of non-gauge equivalent, boundary amplitudes:
\begin{align}
	Z_{\cC_{({1,0})}}\big{[}
\psi\big{]} = \int \rd g \; \psi(g,\id) \qquad \text{for the choice (a),} \nn\\
	Z_{\cC_{({0,1})}}\big{[}
\psi\big{]} = \int \rd g \; \psi(\id,g) \qquad \text{for the choice (b),} \nn\\
	Z_{\cC_{({1,1})}}\big{[}
\psi\big{]} = \int \rd g \; \psi(g,g) \qquad \text{for the choice (c),}	 \nn
\end{align}
In this paper, we have always considered the case (a) and derived this boundary amplitude from the Ponzano-Regge path integral for a solid torus obtained by filling the time slices into disks thus turning the line $t=0$ into a contractible loop and leaving the line $x=0$ as a non-contractible cycle.
Similarly, the case (b)  corresponds to a solid torus with $t=0$ as non-contractible cycle.
The last choice  does not have an obvious 3D completion, i.e. an obvious choice of bulk such that the Ponzano-Regge partition function leads back to this boundary amplitude.
\begin{figure}[!htb]
	\begin{subfigure}[t]{.44\linewidth}
	\begin{tikzpicture}[scale=.8,>=stealth]
	\coordinate (A0) at (-2,-2);
	\coordinate (A1) at (-2,2);
	\coordinate (A2) at (2,2);
	\coordinate (A3) at (2,-2);
	
	\coordinate (O) at (0,0.4);
	\coordinate (O1) at (0,2);
	\coordinate (O2) at (0,-2);
	\coordinate (O3) at (2,0.4);
	\coordinate (O4) at (-2,0.4);
	
	\draw[dotted] (A0)--(A1)--(A2)--(A3)--cycle;
	
	\draw[->-=0.5,color=red] (O1)-- node[pos=0.5,left]{$1$} (O); \draw [->-=0.5,color=red] (O)-- node[pos=0.7,left]{$1$} (O2); 
	\draw[->-=0.5,color=blue] (O)-- node[pos=0.5,above]{$2$} (O3); \draw [->-=0.5,color=blue] (O4)-- node[pos=0.3,above]{$2$} (O); 
	
	\node at (O) {$\bullet$};
	
	\coordinate (C1) at (-2,-1);
	\coordinate (C2) at (2,-1);
	\coordinate (C11) at (-1,-2);
	\coordinate (C22) at (-1,2);
	\coordinate (C111) at (-2,-2);
	\coordinate (C222) at (2,0);
	\coordinate (C333) at (-2,0);
	\coordinate (C444) at (2,2);

	\draw[color=orange] (C111)--(C222); \draw[color=orange] (C333)--(C444);	
	\end{tikzpicture}
	\caption{Boundary graph on the cylinder with only one rectangle face. The temporal link (1) is drawn in red while the spatial link (2) is drawn in blue. The cut cycle associated to the winding $(p_{1},p_{2})=(2,1)$ is drawn in orange.}
	\label{app:fig:p2_q1_simple_disc}
	\end{subfigure}
	\hspace{15mm}
	\begin{subfigure}[t]{.42\linewidth}
	\begin{tikzpicture}[scale=.8,>=stealth]
	\coordinate (A0) at (-2,-2);
	\coordinate (A1) at (-2,2);
	\coordinate (A2) at (2,2);
	\coordinate (A3) at (2,-2);
	
	\coordinate (O) at (0,0);
	\coordinate (O4) at (-2,2);
	\coordinate (O3) at (2,-2);
	
	\coordinate (O1) at (-1,2);
	\coordinate (O2) at (1,-2);
	\coordinate (O1a) at (1,2);
	\coordinate (O1b) at (2,0);
	\coordinate (O2a) at (-2,0);
	\coordinate (O2b) at (-1,-2);
	
	\draw[dotted] (A0)--(A1); \draw[dotted, orange,thick] (A1)--(A2); \draw[dotted] (A2)--(A3); \draw[dotted,orange,thick] (A3)--(A0);
	
	\draw[->-=0.5,color=red] (O1)--(O); \draw [->-=0.5,color=red] (O)--(O2); 
	\draw [->-=0.5,color=red] (O1a)--(O1b); \draw [->-=0.5,color=red] (O2a)--(O2b);

	\draw[->-=0.5,color=blue] (O)-- (O3); \draw [->-=0.5,color=blue] (O4)-- (O); 
	
	\node at (O) {$\bullet$};

	\end{tikzpicture}
	\caption{Boundary geometry of the torus associated associated to the Ponzano-Regge amplitude $Z_{\cC_{(2,1)}}[\psi]$.	The winding of the cycle $\cC$ translates into a Dehn twist on the temporal (in red) link.}
	\label{app:fig:resolved_lattice_dehn}
	\end{subfigure}
	\caption{}
\end{figure}

From our perspective on these boundary linear forms as derived from the 3D quantum gravity path integral, it is nevertheless crucial to understand  the geometry of the bulk that ``fills'' the corresponding 2-torus geometry.
To this purpose, we can unwind the cut and re-draw the torus using the cut $\cC$ as baseline. More precisely, we simply describe the torus in the coordinates $(u,y)$ obtained by a modular transformation. Instead of working with a tilted cut $\cC$, the cut now defines the horizontal time slice while the two graph links are now tilted. We illustrate this mapping on figure \ref{app:fig:p2_q1_simple_disc} in a simple case with $(p_{1},p_{2})=(2,1)$. This shows that this choice of cut corresponds to a  Dehn twist  of angle $2\pi$ around the link $1$. In general, the cycle with winding $(p_{1},p_{2})$ corresponds to a Dehn twist of angle $(p-1) 2\pi$ around the temporal cycle and angle $(q-1)2\pi$ around the spatial cycle.

We see that the boundary graph is now tilted and made of parallelograms. Thus, in order to take into account Dehn twists and large diffeomoprhisms, we should consider not only straight square lattices but tilted parallelogram lattices. One should extend these 2D lattices into bulk cellular decompositions in order to check that the Ponzano-Regge partition function effectively leads to the expected boundary amplitudes, and then compute these boundary amplitudes for coherent spin network states in order to study the behavior of the Ponzano-Regge model under large boundary diffeomorphisms.
We postpone such a study to future work.

\section{Stationary spin configurations from Critical couplings}
\label{app:criticalcouplings}

In this appendix, we solve for the spins $L_{x}$ and $T_{t}$ on the square lattice in terms of the couplings $\tau_{t,x}$ and $\lambda_{t,x}$ at criticality.
The stationarity point equations give the couplings,
\begin{align}
|\tau_{t,x}| =
\f1{\sqrt{(1+A)(1+C)}}
\,,\qquad
|\tau_{t+1,x}| =
\f1{\sqrt{(1+B)(1+D)}}
\\
|\lambda_{t,x}| =
\f{\sqrt{AB}}{\sqrt{(1+A)(1+B)}}
\,,\qquad
|\lambda_{t,x+1}| =
\f{\sqrt{CD}}{\sqrt{(1+C)(1+D)}}
\,,
\end{align}
in terms of the spins ratio $A$, $B$ $C$ and $D$ defined by
\begin{equation}
A:=\f{L_{x}}{T_{t}}
\,,\quad
B:=\f{L_{x}}{T_{t+1}}
\,,\quad
C:=\f{L_{x+1}}{T_{t}}
\,,\quad
D:=\f{L_{x+1}}{T_{t+1}}
\,,\quad
\textrm{satisfying the relation}\,\,
AD=BC
\,.
\end{equation}
These relations impose a polynomial condition on the couplings:
\begin{equation}
1+|\lambda_{t,x}|^2|\lambda_{t,x+1}|^2+|\tau_{t,x}|^2|\tau_{t+1,x}|^2
=
2|\lambda_{t,x}||\lambda_{t,x+1}||\tau_{t,x}||\tau_{t+1,x}|
+|\lambda_{t,x}|^2+|\lambda_{t,x+1}|^2+|\tau_{t,x}|^2+|\tau_{t+1,x}|^2
\,,
\label{app:eq:critical}
\end{equation}
ensuring that the couplings and the spin ratios carry the same number of independent degrees of freedom.

Our goal is to invert these relations and express the ratios  $A$, $B$, $C$ and $D$ in terms of  the couplings.
We introduce the more compact notation:
\begin{equation*}
\alpha = |\tau_{t,x}| \; , \quad \beta = |\tau_{t+1,x}| \; \quad x = |\lambda_{t,x}| \; , \quad y=|\lambda_{t,x+1}| \;.
\end{equation*}
We start by expressing the ratios $B$, $C$ and $D$ in terms of $A$ and the couplings.
Using the first two relations for the coouplings, it is immediate to derive the expression for $B$ and $C$,
\begin{equation}
B = \f{x^2(1+A)}{A-(1+A)x^2}
\,,\qquad
C = \f{1-\alpha^2 (1+A)}{\alpha^2 (1+A)}
\end{equation}
The remaining ratio $D$, can similarly be expressed in terms of $B$ or $C$ using the two remaining relations:
\begin{equation}
D = \f{y^2 (1+C)}{C-(1+C)y^2} = \f{1-\beta^2 (1+B)}{\beta^2 (1+B)} \; .
\end{equation}
Inserting the expression of $B$ and $C$ into the previous formulas, we obtain a constraint on $A$:
\begin{equation}
 \f{y^2}{(1-\alpha^2-y^2)-\alpha^2 A} = D =\f{(1-x^2-\beta^2) A - x^2}{\beta^2 A} \,,
\end{equation}
which gives a second order equation in $A$:
\begin{equation}
\alpha^2\left(1-x^2-\beta^2\right) A^2 - 2 \alpha x( \alpha x + \beta y)A + x^2(1-\alpha^2-y^2) = 0\,,
\end{equation}
where we have used the constraint \eqref{app:eq:critical} to simplify the linear term in $A$.
Further using the relation \eqref{app:eq:critical}, it turns out that  the discriminant  vanishes and that this equation is a perfect square giving a unique solution for $A$:
\begin{equation*}
A = x \f{\alpha x + \beta y}{\alpha(1-x^2 - \beta^2)}
\,.
\end{equation*}
Finally inserting this solution for $A$ into  the expressions of $B$, $C$ and $D$, we get:
\begin{equation*}
B = \frac{x \left(-\alpha  \beta ^2+\alpha +\beta  x y\right)}{\beta  \left(x^2 (-y)+\alpha  \beta  x+y\right)} \; .
\end{equation*}
\begin{equation*}
C = \frac{x^2+\alpha  \beta  x y-\left(\alpha ^2-1\right) \left(\beta ^2-1\right)}{\alpha  \left(\alpha  \left(\beta ^2-1\right)-\beta  x y\right)} \; .
\end{equation*}
\begin{equation*}
D = \frac{\left(1-\beta ^2-x^2\right) \left(x (\alpha  x+\beta  y)-\alpha x^2\right)}{\beta ^2 x (\alpha  x+\beta  y)} \; .
\end{equation*}
One can check that these are indeed solutions of the starting equations and naturally satisfy the  condition $AD = BC$.

Furthermore, one could have derived the solutions A,B,C and D without assuming the condition \eqref{app:eq:critical}. The second order equation on A would then return two solutions, $A_{\pm}$, leading to  two branches of solutions $(A,B,C,D)_{\pm}$. Finally imposing the condition $AD=BC$  on the solutions then leads back to the constraint \eqref{app:eq:critical} on the critical couplings.

\section{Ponzano-Regge determinants and bi-variate Chebyshev polynomials}
\label{app:chebychev}

We would like to point out that the inverse determinants $\det\big[Q_{\omega,k}(\varphi)\big]^{-1}$ can be interpreted in terms of  the generating function to a generalization of the Chebyshev polynomial to bivariate polynomials. Indeed, these determinants can be written as polynomials of traces of $\SU(2)$ group elements:
\begin{equation}
\label{eq:det_SU(2)_element_expression}
\det\big[Q_{\omega,k}(\varphi)\big]
= 
1 + \lambda^2 + \tau^2 - \tau \tr(G_{k}^\tau) - \lambda \tr(G_{\omega,k}^\lambda) +\lambda\tau \tr(G_k^\tau G_{\omega,k}^\lambda)
\,,
\end{equation}
in terms of the pair of $\SU(2)$ group elements defined as:
\be
G_{k}^\tau
=
e^{i \f{2\pi}{N_x}k \sigma_z}
\,,\qquad
G_{\omega,k}^\lambda
=
e^{\f i{N_{t}}(\varphi+ {2\pi}\omega - \gamma k) \sigma_x} \; .
\ee
%
%
This expression actually allows to easily generalize the formulas for the Ponzano-Regge amplitude with toroidal boundary coherent spin networks to a tilted square lattice, built of parallelograms instead of rectangles. Indeed, we introduce  the spinor $|\theta \ra$ associated to the vector $\vec{u}_{\theta} = (\cos(\theta),0,\sin(\theta))$. The case $\theta = 0$ corresponds to the rectangular case we work on in this paper. Replacing $|+\ra$ in the coherent intertwiners by $|\theta\ra$ corresponds to replacing $G^{\lambda}_{\omega,k}$ in the determinant $\det\big[Q_{\omega,k}(\varphi)\big]$ of the Hessian matrices by
\begin{equation*}
G^{\lambda}_{\omega,k}(\theta)
=
e^{\f i{N_{t}}(\varphi+ {2\pi}\omega - \gamma k) \vec{\sigma}.\vec{u}_{\theta}}
\; .
\end{equation*}

Here, we would to focus on the function of two couplings $\lambda,\tau$ and two group elements $G,\tG$:
\be
\cF_{2}(\lambda,\tau;G,\tG)
:=
\f1{1 + \lambda^2 + \tau^2 - \tau \tr(G) - \lambda \tr(\tG) +\lambda\tau \tr(G\tG)}
\,.
\ee
When one of the coupling vanishes, say $\lambda=0$, this function depends only on one group element, here $G$, and it reduces to the generating function of the  Chebyshev polynomial of the second kind $U_{n}$:
\be
\cF_{1}(\tau;G)
:=
\cF_{2}(0,\tau;G,\tG)
=
\f1{1 + \tau^2 - \tau \tr(G)}
=
\f1{1 + \tau^2 -2 \tau\cos\theta}
=
\sum_{n\in\N} U_n(\cos\theta) \tau^{n}\,,
\ee
where $\theta$ is the class angle of the group element $G$ and the  Chebyshev polynomial $U_n(\cos\theta)$ turns out to be the character of $G$ in the irreducible representation of spin $j$:
\be
\tr G=2\cos\theta
\,,\qquad
 U_n(\cos\theta)
 =
 \f{\sin (n+1)\theta}{\sin\theta}
 =
 \chi_{\f n2}(G)
 \,.
\ee
The orthogonality of the  Chebyshev polynomial  simply amounts to the orthonormality of the characters with respect to the Haar measure on $\SU(2)$:
\be
\int_{\SU(2)} \rd G\, \chi_{\f n2}(G) \chi_{\f {n'}2}(G)=\delta_{n,n'}\,.
\ee
From this perspective, we define  bi-variate polynomials generated by $\cF_{2}$:
\be
\cF_{2}(\lambda,\tau;G,\tG)
=
\sum_{n,m\in\N} U_{n,m}(G,\tG) \tau^{n}\lambda^{m}
\,.
\ee
It is straightforward to check that the polynomials $U_{n,m}(G,\tG)$ are orthogonal,
\be
\int_{\SU(2)^{\times 2}} \rd G\,\rd\tG\,\,
U_{n,m}(G,\tG)U_{n',m'}(G,\tG)
\propto \delta_{n,n'} \delta_{m,m'}
\,.
\nn
\ee
It would be interesting to understand further the properties of those orthogonal polynomials, which intertwiner they define between $G$ and $\tG$ and so on.
We leave this detailed study to future investigation.


\bibliographystyle{bib-style}
\bibliography{PonzanoRegge}

\end{document}